\PassOptionsToPackage{unicode}{hyperref}
\PassOptionsToPackage{hyphens}{url}
\PassOptionsToPackage{dvipsnames,svgnames,x11names}{xcolor}
\documentclass[
  12pt]{article}
\usepackage{amsthm}
\usepackage{amsmath,amssymb}
\usepackage{iftex}

\ifPDFTeX
  \usepackage[T1]{fontenc}
  \usepackage[utf8]{inputenc}
  \usepackage{textcomp} 
\else 
  \usepackage{unicode-math}
  \defaultfontfeatures{Scale=MatchLowercase}
  \defaultfontfeatures[\rmfamily]{Ligatures=TeX,Scale=1}
\fi
\usepackage{lmodern}
\ifPDFTeX\else  
\fi
\IfFileExists{upquote.sty}{\usepackage{upquote}}{}
\IfFileExists{microtype.sty}{
  \usepackage[]{microtype}
  \UseMicrotypeSet[protrusion]{basicmath} 
}{}
\makeatletter
\@ifundefined{KOMAClassName}{
  \IfFileExists{parskip.sty}{%
    \usepackage{parskip}
  }{
    \setlength{\parindent}{0pt}
    \setlength{\parskip}{6pt plus 2pt minus 1pt}}
}{
  \KOMAoptions{parskip=half}}
\makeatother
\usepackage{xcolor}
\makeatletter
\ifx\paragraph\undefined\else
  \let\oldparagraph\paragraph
  \renewcommand{\paragraph}{
    \@ifstar
      \xxxParagraphStar
      \xxxParagraphNoStar
  }
  \newcommand{\xxxParagraphStar}[1]{\oldparagraph*{#1}\mbox{}}
  \newcommand{\xxxParagraphNoStar}[1]{\oldparagraph{#1}\mbox{}}
\fi
\ifx\subparagraph\undefined\else
  \let\oldsubparagraph\subparagraph
  \renewcommand{\subparagraph}{
    \@ifstar
      \xxxSubParagraphStar
      \xxxSubParagraphNoStar
  }
  \newcommand{\xxxSubParagraphStar}[1]{\oldsubparagraph*{#1}\mbox{}}
  \newcommand{\xxxSubParagraphNoStar}[1]{\oldsubparagraph{#1}\mbox{}}
\fi
\makeatother

\usepackage{longtable,booktabs,array}
\usepackage{calc} 
\usepackage{etoolbox}
\makeatletter
\patchcmd\longtable{\par}{\if@noskipsec\mbox{}\fi\par}{}{}
\makeatother
\IfFileExists{footnotehyper.sty}{\usepackage{footnotehyper}}{\usepackage{footnote}}
\makesavenoteenv{longtable}
\usepackage{graphicx}
\makeatletter
\def\maxwidth{\ifdim\Gin@nat@width>\linewidth\linewidth\else\Gin@nat@width\fi}
\def\maxheight{\ifdim\Gin@nat@height>\textheight\textheight\else\Gin@nat@height\fi}
\makeatother
\setkeys{Gin}{width=\maxwidth,height=\maxheight,keepaspectratio}
\makeatletter
\def\fps@figure{htbp}
\makeatother

\makeatletter
\@ifpackageloaded{caption}{}{\usepackage{caption}}
\AtBeginDocument{%
\ifdefined\contentsname
  \renewcommand*\contentsname{Table of contents}
\else
  \newcommand\contentsname{Table of contents}
\fi
\ifdefined\listfigurename
  \renewcommand*\listfigurename{List of Figures}
\else
  \newcommand\listfigurename{List of Figures}
\fi
\ifdefined\listtablename
  \renewcommand*\listtablename{List of Tables}
\else
  \newcommand\listtablename{List of Tables}
\fi
\ifdefined\figurename
  \renewcommand*\figurename{Figure}
\else
  \newcommand\figurename{Figure}
\fi
\ifdefined\tablename
  \renewcommand*\tablename{Table}
\else
  \newcommand\tablename{Table}
\fi
}
\@ifpackageloaded{float}{}{\usepackage{float}}
\floatstyle{ruled}
\@ifundefined{c@chapter}{\newfloat{codelisting}{h}{lop}}{\newfloat{codelisting}{h}{lop}[chapter]}
\floatname{codelisting}{Listing}

\makeatother
\makeatletter
\@ifpackageloaded{caption}{}{\usepackage{caption}}
\@ifpackageloaded{subcaption}{}{\usepackage{subcaption}}
\makeatother

\ifLuaTeX
  \usepackage{selnolig}  
\fi
\usepackage[]{natbib}
\usepackage{bookmark}

\IfFileExists{xurl.sty}{\usepackage{xurl}}{} 
\hypersetup{
  pdftitle={Title},
  pdfauthor={Author 1; Author 2},
  pdfkeywords={3 to 6 keywords, that do not appear in the title},
  colorlinks=true,
  linkcolor={blue},
  filecolor={Maroon},
  citecolor={Blue},
  urlcolor={Blue},
  pdfcreator={LaTeX via pandoc}}

\newtheorem{theorem}{Theorem}
\newtheorem{proposition}[theorem]{Proposition}%
\newtheorem{remark}{Remark}%

\usepackage[doublespacing]{setspace}
\usepackage[fontsize=12pt]{fontsize}
\usepackage{bookmark}
\usepackage{amssymb}
\usepackage{xcolor}
\usepackage{setspace}
\usepackage{makecell}
\usepackage{graphicx}
\usepackage{natbib}
\usepackage{booktabs}
\usepackage{multirow} 
\usepackage{tabularx} 
\usepackage{ragged2e}

\newcommand{\widgraph}[2]{\includegraphics[keepaspectratio,width=#1]{#2}}

\newcommand{\diag}{\text{diag}}

\newcommand{\Gh}{\widehat{G}}

\newcommand{\Fh}{\widehat{F}}

\renewcommand{\SS}{\mathcal{S}}
\newcommand{\EE}{\mathbb{E}}

\renewcommand{\Re}{\mathbb{R}}
\newcommand{\PP}{\mathcal{P}}
\newtheorem{lemma}{Lemma}

\newcommand{\anon}{1}

\newcommand{\RPE}{\text{RE}}

\renewcommand{\SS}{\mathcal{S}}

\renewcommand{\Re}{\mathbb{R}}

\newcommand{\SWB}{\text{SWB}}
\newcommand{\SW}{\text{SW}}
\newcommand{\Gt}{\widetilde{G}}
\newtheorem{assumption}{Assumption}

\DeclareMathOperator*{\argmin}{arg\,min}

\begin{document}

\def\spacingset#1{\renewcommand{\baselinestretch}%
{#1}\small\normalsize} \spacingset{1}


\if1\anon
{
  \title{\bf Bayesian Multiple Multivariate Distribution-Distribution Regression}
  \author{Khai Nguyen$^{3}$\footnote{Part of this work was done while Khai Nguyen was at UT Austin.}, Yang Ni$^1$, Peter Mueller$^{1,2}$\hspace{.2cm}\\
    $^1$Department of Statistics and Data Science, University of Texas at Austin\\
    $^2$Department of Mathematics, University of Texas at Austin \\
    $^3$Department of Statistics and Data Science, Texas A\&M University
    }
  \maketitle
} \fi

\if0\anon
{
  \bigskip
  \bigskip
  \bigskip
  \begin{center}
    {\LARGE\bf Bayesian Multiple Multivariate Distribution-Distribution Regression}
\end{center}
  \medskip
} \fi

\bigskip
\begin{abstract}
We propose the first approach for multiple multivariate
distribution–distribution regression (MDDR), making it possible to consider the regression of a
multivariate distribution–valued response on multiple multivariate
distribution–valued predictors. The core idea is to define a fitted
distribution using a sliced Wasserstein barycenter (SWB) of
push-forwards of the predictors and to quantify deviations from the
observed response using the sliced Wasserstein (SW)
distance. Regression functions, which map predictors’ supports to the
response support, and barycenter weights are inferred within a
generalized Bayes framework, enabling principled uncertainty
quantification without requiring a fully specified likelihood. The
inference process can be seen as an instance of an inverse SWB
problem. We establish theoretical guarantees, including
sample complexity of the generalized likelihood and posterior
consistency, under stability and uniqueness assumptions on the SWB.
For practical inference, we introduce a differentiable
approximation of the SWB and a smooth reparameterization to handle the
simplex constraint on barycenter weights, allowing efficient
gradient-based MCMC sampling. We demonstrate MDDR in
 an application to inference for 
population-scale single-cell data.
Posterior analysis under the MDDR model in this example identifies which source/sender cell types are predictive of a target/receiver cell type. The proposed approach provides accurate fits and reliable predictions, and a leave-one-predictor-out measure of each predictor's contribution yields a sparse network of predictive associations between cell types, suggesting candidate cell-cell communication patterns.
\end{abstract}

\noindent%
{\it Keywords:} Sliced Wasserstein Barycenter, Optimal Transport, Single-Cell Data, Generalized Bayes
\vfill

\newpage
\spacingset{1.8} 

\section{Introduction}
\label{sec:introduction}

We propose \emph{multiple multivariate
distribution–distribution regression},  of a
\emph{multivariate}  distribution-valued response on several
\emph{multivariate} distribution-valued predictors. The proposed approach
does not impose Riemannoannian geometry on the space of predictors and
responses, which is a key requirement for existing distribution-distribution
regression methods. The proposed approach is computationally
and statistically scalable, and can handle predictors and responses of
any form, including discrete and continuous. The main feature is the
use of a sliced Wasserstein barycenter~\citep{bonneel2015sliced} (SWB)
to define a fitted distribution based on push-forwards (marginals) of multiple predictors, and the use of sliced Wasserstein
distance~\citep{rabin2012wasserstein} (SW) to measure the deviation of
the fitted distribution from the observed response distribution.
We implement a generalized Bayes~\citep{bissiri2016general} framework for inference on the barycenter weights and on the regression functions, which map the support of the predictors to that of the response to define the push-forwards.
 The generalized Bayes setup  enables principled
belief updating without requiring a fully specified generative model.  While our primary goal is to develop a regression model,
the approach can also be interpreted as the first instance of an
\emph{inverse SWB} problem, namely, inferring the marginals and their
corresponding barycenter weights given a noisy  observation of the
barycenter. 

Beyond the methodological contributions, we  establish several
new theoretical properties of the proposed regression approach. We derive two key results:the sample complexity of the generalized likelihood and posterior consistency.
Specifically, we show that the generalized
likelihood can be reliably estimated even when the distributions are
observed only through samples, and we prove that the posterior concentrates on the set of risk minimizers, which reduces to the true parameters when they are identified, both for the posterior based on the underlying distributions and for the posterior computed from samples.

For posterior  inference, we develop practicable posterior
simulation schemes. 
First, we introduce a differentiable approximation
of SWB using an iterative algorithm. Second, we handle the simplex
constraint on the barycenter weights through a smooth and invertible
change-of-variable transformation. Together
with the differentiability of SW distance, this makes it possible to use
 computation-efficient 
gradient-based Markov chain Monte Carlo (MCMC) samplers,
such as the Metropolis-adjusted Langevin algorithm
(MALA)~\citep{girolami2011riemann,ning2025metropolis}.

Distribution regression encompasses settings where distribution-valued predictors are used to predict real-valued responses~\citep{szabo2016learning,matabuena2023distributional}, as well as settings where real-valued predictors are used to model distribution-valued responses~\citep{tokdar2004bayesian,dunson2007bayesian,shen2016adaptive,ghosal2025distributional,petersen2022modeling}. Distribution-distribution regression (DDR) generalizes this to scenarios where both the predictor and the response are distributions.  For example, simple univariate DDR, with
a one-dimensional predictor and a one-dimensional response, is considered
in~\citet{zhao2023density} using a Riemannian geometry of
distributions on the Wasserstein-2 space. Similarly, univariate DDR using autoregressive models is proposed in~\citet{ghodrati2024distributional,jiang2026wasserstein,zhang2022wasserstein}, or~\citet{zhu2023autoregressive}. The key challenge of
extending these approaches to multivariate distributions is the
computational intractability of Riemannian geometrical tools in higher
dimensions. Moreover, the approach requires distributions to be
continuous for accessing the Riemannian structure with the existence
of optimal transport maps. To address
 these limitations, 
\citet{nguyen2025bayesian} propose a simple multivariate DDR model
using a generalized Bayes framework~\citep{bissiri2016general}. The
key modeling strategy is to employ  sliced Wasserstein (SW)
distance~\citep{rabin2012wasserstein} to construct a generalized
likelihood, enabling scalable inference without requiring any
Riemannian structure on either the predictors or the responses. 


Considering more than one predictor gives rise to the multiple distribution-distribution regression problem. In the context of DDR, multiple univariate DDR is
proposed in~\citet{chen2024distribution}. Again, the approach requires
the continuity of distributions to utilize the Wasserstein geometry of
distributions that is impractical for computation in multivariate
cases. We  instead develop an extension of the approach  in~\citet{nguyen2025bayesian} to the
multiple regression setting. First, predictor distributions are
transformed to the same dimensions as the response distribution using
regression functions.  After that, we form the fitted distribution as
the SWB of the transformed predictor distributions. An inference model is then
defined with a generalized likelihood, based on SW distance between fitted distribution and the response. 

The use of the SWB and SW distance addresses both of the aforementioned
multivariate and finite-sample challenges.
Under an assumption about the stability of the SWB, we show that the generalized
likelihood achieves an estimation rate of
$\mathcal{O}(n^{-1/(2p^2)})$, which is the parametric rate when $p=1$ ($n$ is the minimal number of atoms of the involved distributions, and $p$ is the order of the Wasserstein distance; see Section~\ref{sec:background}), under the simple plug-in estimation when the
distributions have compact supports and the barycenter is stable under perturbations of the marginals. Furthermore, we establish that the posterior concentrates on the set of risk minimizers under parametric settings, which are the true parameters when the model is correctly specified and identified. Establishing posterior consistency of the proposed model and analyzing the sample complexity of the generalized likelihood is also complicated because the fitted distribution is obtained via a barycenter. The
theoretical studies are novel for both SW literature and the
generalized Bayes literature. In addition to the theoretical analysis, applying MALA directly is nontrivial because the gradient of the barycenter is defined implicitly. We approximate it via unrolled optimization.

Our approach is motivated by
 the analysis of  population-scale single-cell data, where
gene expression levels are measured across large numbers of cells
sampled from  multiple  subjects. Using established markers to identify cell types, these data provide a unique
opportunity to uncover communication patterns between cell types
through the expression of ligand–receptor pairs. Unlike DDR, which
infers communication between pairs of cell types independently,
multiple DDR (MDDR) enables joint  regression  of a target
cell type on multiple source cell types simultaneously.
Importantly, we only
observe a finite (though potentially large) number of cells for each
cell type, making the estimation of optimal transport maps, and thus
the use of Riemannian geometric tools as
in~\cite{chen2024distribution,zhu2023geodesic}, impractical. In
addition, our approach avoids continuity assumptions on the underlying
distributions and is capable of operating beyond Euclidean geometry on
atoms of distributions.

The proposed computational techniques implement inference for the MDDR model. We show that the model is particularly
meaningful when predictors lie in lower-dimensional spaces than the
response and are observed separately without a joint distribution, as is typical in single-cell data. The model provides
accurate fits and predictions for both simulated and real single-cell
datasets. 
Moreover, MDDR induces a weighted graph between cell types, in which the weight of the edge from a sender to a receiver is the share of the improvement of MDDR over a reference model that is lost when the sender is removed. Keeping, for each receiver, the edges with above-average weight gives a sparse network of predictive associations, which suggests candidate communication patterns

The remainder of the article is organized as follows. 
Section~\ref{sec:background} reviews the definitions of 
Wasserstein distance, SW distance, and the SWB. 
In Section~\ref{sec:MDDR}, we introduce the Bayesian MDDR framework, which uses SW and SWB to define a generalized likelihood, establish key theoretical properties, and develop posterior simulation via a MALA sampler.
Section~\ref{sec:simulation} presents results on a simulated dataset,
 highlighting the fundamental difference in the multiple regression under
MDDR compared to DDR. 
In Section~\ref{sec:cell_cell_communication}, we apply Bayesian MDDR
to single-cell data,  exploring  fitting, predictive
performance, and  the discovery of  a cell-cell communication
network. 
Section~\ref{sec:conclusion} concludes the paper and discusses
directions for future work. Technical proofs and additional
experimental results are provided in the supplementary material. 

For notation, let $\delta_x$ denote the Dirac delta measure at $x$. For any $d \ge 2$, define the unit hypersphere $\mathbb{S}^{d-1} = \{\theta \in \mathbb{R}^d : \|\theta\|_2 = 1\}$. For two sequences $a_n$ and $b_n$, we write $a_n = \mathcal{O}(b_n)$ if there exists a universal constant $C$ such that $a_n \le C\, b_n$ for all $n \ge 1$.  Given two measurable spaces $(\mathcal{X}_1, \Sigma_1)$ and $(\mathcal{X}_2, \Sigma_2)$, a measurable function $f : \mathcal{X}_1 \to \mathcal{X}_2$, and a measure $\mu$ on $(\mathcal{X}_1, \Sigma_1)$, the push-forward of $\mu$ through $f$ is defined by
$
f{\sharp} \mu (B) = \mu(f^{-1}(B)),\, \text{for all } B \in \Sigma_2. 
$  We denote the $K$-simplex as $\Delta^K$, i.e., $(\pi_1,\ldots,\pi_K)\in \Delta^K$ implies $0\leq \pi_k\leq 1,\, \forall k=1,\ldots,K$, and $\sum_{k=1}^K \pi_k=1$.  Finally, $\PP_p(\Re^d)$ is the set of all distributions supported on $\Re^d$ with finite $p$-th moment. Additional notation will be introduced as needed.

\section{Background}
\label{sec:background}
By way of a brief review of Wasserstein distance, sliced Wasserstein distance (SW), and sliced Wasserstein barycenter (SWB), we introduce some notation and definitions. Given $p\geq 1$, let $G_1, G_2 \in \PP_p(\Re^d)$, where
$\PP_p(\Re^d)$ be the set of all distributions supported on
$\Re^d$ 
with finite $p$-th moment.
Wasserstein-$p$ 
distance~\citep{villani2009optimal,peyre2019computational} between  $G_1$ and $G_2$ is
defined as:
\begin{align}
\label{eq:Wasserstein}
    W_p^p(G_1, G_2) = \inf_{\pi \in \Pi(G_1, G_2)} \int_{\Re^d \times \Re^d} \|x- y\|_p^p \, \mathrm{d}\pi(x, y),
\end{align}
where $\Pi(G_1, G_2) = \left\{\pi \in \PP(\Re^d \times \Re^d) \mid \pi(A, \Re^d) = G_1(A), \ \pi(\Re^d, B) = G_2(B) \ \forall A, B \subset \Re^d \right\}$ is the set of all transportation plans/couplings. In one dimension  Wasserstein distance admits the closed-form:
\begin{align}
    W_p^p(G_1, G_2) = \int_0^1|  Q_{G_1}(t) -Q_{G_2}(t) |^p \mathrm{d}t,
\end{align}
where $Q_{G_1}$ and $Q_{G_2}$  are the quantile function of $G_1$ and $G_2$. 

 The availability of a closed-form solution for $d=1$ is exploited 
by \cite{rabin2012wasserstein} to define sliced Wasserstein (SW) distance. 
In particular, SW distance
between two distributions $G_1,G_2 \in 
\mathcal{P}_p(\mathbb{R}^d)$ is defined as:
\begin{align}
        \label{eq:SW}
\SW_p^p(G_1, G_2) = \EE_{\theta \sim
  \mathcal{U}(\mathbb{S}^{d-1})}[W_p^p(\theta \sharp G_1, \theta
\sharp G_2)],
\end{align}
where  $\mathcal{U}(\mathbb{S}^{d-1})$ is the uniform distribution over $\mathbb{S}^{d-1}$, and $\theta \sharp G_1$
and $\theta \sharp G_2$ denote the pushforward distribution of $G_1$
and $G_2$ through  projections 
$f_\theta(x) =  \theta^\top x$.   Using Monte Carlo estimation for the expectation gives:
\begin{align}
    \label{eq:MC_empirical_SW}
    \widehat{\SW}_p^p(G_{1}, G_{2};L) = \frac{1}{L}\sum_{l=1}^LW_p^p(\theta_l \sharp G_{1}, \theta_l \sharp G_{2}),
\end{align}
where $\theta_1,\ldots,\theta_L \overset{i.i.d.}{\sim} \mathcal{U}(\mathbb{S}^{d-1})$ with $L$ being the number of Monte Carlo samples.

The key benefit of SW compared to Wasserstein distance arises when
 the distributions are represented by i.i.d. samples. 
In particular, when
observing $x_1,\ldots,x_{m_1} \overset{i.i.d.}{\sim} G_1$ and
$y_1,\ldots,y_{m_2} \overset{i.i.d.}{\sim} G_2$, we have
$\EE\left[\left|W_p(\hat{G}_{1}, \hat{G}_{2})-W_p(G_1,
  G_2)\right|\right] =
\mathcal{O}(m_1^{-1/d}+m_2^{-1/d})$ for $d>2p$ and compactly supported $G_1,G_2$~\citep{fournier2015rate} and
$\EE\left[\left|\SW_p(\hat{G}_{1},
    \hat{G}_{2})-\SW_p(G_1, G_2)\right|\right] =
\mathcal{O}(m_1^{-1/(2p)} +m_2^{-1/(2p)})$  independently of the dimension~\citep{nadjahi2020statistical,nietert2022statistical},
where $d>1$ is the dimension, and $\hat{G}_1=\frac{1}{m_1}
\sum_{i=1}^{m_1} \delta_{x_i}$ and $\hat{G_2}=\frac{1}{m_2}
\sum_{j=1}^{m_2} \delta_{y_j}$.
In addition to better sample complexity, SW has a lower
time complexity $\mathcal{O}\left\{
(m_1+m_2)\log(m_1+m_2)\right\}$~\citep{peyre2020computational}
compared to Wasserstein
$\mathcal{O}\{(m_1+m_2)^3\log(m_1+m_2)\}$~\citep{peyre2020computational}.

SW distance induces a metric on the space of distribution
$\PP_p(\Re^d)$. The notion of weighted average (Fréchet mean) is
generalized to the concept of SW barycenter~\citep{bonneel2015sliced}
(SWB). The SWB of $K\geq 2$
marginals $G_1,\ldots,G_K \in \PP_p(\Re^d)$ with marginal weights
$(\pi_1,\ldots,\pi_K) \in \Delta^K$ is defined as: 
\begin{align}
  \label{eq:SWB}
  \SWB_p(G_1,\ldots,G_K,\pi_1,\ldots,\pi_K)=
  \text{argmin}_{G \in \PP_p(\mathbb{R}^d)}\sum_{k=1}^K \pi_k \SW_p^p (G,G_k).
\end{align}
We will discuss the computation of SWB in later sections.

\section{Multiple Multivariate Bayesian Distribution-Distribution Regression}
\label{sec:MDDR}

\subsection{Proposed Model}
\label{subsec:MDDR}
We consider the inference problem of regressing a distribution-valued
response  $G_i \in \PP_p(\Re^{d})$ $(p\geq 1, d\geq 2)$ on multiple
distribution-valued predictors $F_{i1},\ldots,F_{iK}$ ($K\geq 2$) with
$F_{ik} \in \PP_p (\Re^{h_k})$ for  predictors  $k=1,\ldots,K$ and  experimental units  $i=1,\ldots,N$
($N>0$). We denote the data as
$\SS=\{(F_{i1},\ldots,F_{iK},G_i)\}_{i=1}^N$.
We  introduce MDDR  by way of a {\em generalized likelihood} 
\citep{bissiri2016general}   $\ell(f_1,\ldots,f_K;G_i,F_{i1},\ldots,F_{iK})$ based on a loss function for a fitted
approximation  $\Gt_i$  of $G_i$: 
\begin{align}
  \label{eq:MDDRmodel}
  &\ell(f_1,\ldots,f_K,\pi_1,\ldots,\pi_K;G_i,F_{i1},\ldots,F_{iK}) =
  \exp\left\{-w\, \SW_p^p[\Gt_i(f,\pi),G_i)]\right\}, \\
  &\Gt_i(f,\pi) = \SWB_p(f_1\sharp F_{i1},\ldots,f_K\sharp F_{iK}, \pi_1,\ldots,\pi_K),  \label{eq:MDDRmodel2}
\end{align}
where $f_k$ is a measurable function that maps from $\mathbb{R}^{h_k}$ to
$\mathbb{R}^{d}$ for $k=1,\ldots,K$, $w > 0$, and $(\pi_1,\ldots,\pi_K)\in \Delta^K$. We assume throughout that the minimizer in \eqref{eq:SWB} defining $\Gt_i(f,\pi)$ is unique, so that the fitted distribution $\Gt_i(f,\pi)$ is well-defined.

In~\eqref{eq:MDDRmodel}, the fitted distribution $\Gt_i$ represents the (in-sample) prediction of the response distribution $G_i$ given the predictors $F_{i1},\ldots,F_{iK}$. It is defined as the SWB of $f_1 \sharp F_{i1},\ldots,f_{K}\sharp F_{iK}$, which we refer to as \emph{push-forwards of the predictors}. The fitted distribution $\Gt_i$ can be seen as a Fréchet mean (generalized weighted average) of the $f_k\sharp F_{ik}$. Using $\Gt_i$ and the SW distance, \eqref{eq:MDDRmodel} defines a loss function that serves as generalized (negative log) likelihood.
 The generalized Bayes framework
\eqref{eq:MDDRmodel} and \eqref{eq:MDDRmodel2}  offers a principled
approach for multiple multivariate distribution-distribution regression while
naturally incorporating uncertainty within a Bayesian framework
(priors to be discussed in the next subsection).  The model can be viewed as an
instance of a novel inference paradigm, which we refer to as
\emph{inverse sliced Wasserstein barycenter}. In particular, we
consider the response $G_i$ as a noisy barycenter and aim to infer the
underlying marginals and weights that generate it. In our formulation,
the barycenter weights  $\pi_k$ 
are shared across $i$, and the marginals $f_k\sharp F_{ik}$ are constructed through shared regression functions $f_k$.

We restrict $\pi$ to the simplex as the default of SWB. The reason is that the objective in \eqref{eq:SWB} is convex in $G$ only for nonnegative weights (Proposition~\ref{proposition:convexity_SWB}). Moreover, nonnegative weights summing to one make the weights $\pi$ (after suitable normalization, see later)  interpretable as relative contributions. Negative weights would add no flexibility since  the sign and scale of each push-forward are already free through $f_k$. As we shall discuss later, the scale of $\pi$ is only identified jointly with $f_k$.

Unlike a classical regression model, \eqref{eq:MDDRmodel} specifies no error distribution for $G_i$. The fitted distribution $\Gt_i(f,\pi)$ is a deterministic function of the predictors and the parameters, and $\SW_p^p[\Gt_i(f,\pi),G_i]$ plays the role of a residual. The prior is updated by the exponentiated negative loss with learning rate $w$ \citep{bissiri2016general}, the Bayesian counterpart of $M$-estimation with the SW loss. We use this formulation because no tractable sampling model for a random distribution centered at a barycenter is available, while the SW loss is computationally cheap and differentiable.

To assess goodness of fit, we compute the residual error between the
fitted and observed response distributions
 in \eqref{eq:MDDRmodel} and \eqref{eq:MDDRmodel2}: 
$
PE(\SS) = \frac{1}{N}\sum_{i=1}^N \SW_p^p(\Gt_i, G_i).
$
To calibrate $PE$, we introduce a reference model that represents a baseline fit (we introduce specific examples later): 
$
\Gt_i'(f', \pi') = \SWB_p\big(f_1'\!\sharp F_{i1}, \ldots, f_K'\!\sharp F_{iK};\, 1/K, \ldots, 1/K\big),
$
where $f'$ denotes regression functions under strong restrictions. We then normalize the prediction error relative to this reference to obtain the relative residual error (RE):
  \begin{align}
  \label{eq:RE}
     \RPE(\SS)=\frac{1}{N} \sum_{i=1}^N \frac{\SW_p^p(\Gt_i(f,\pi),G_i)}{\SW_p^p(\Gt_{i}'(f',\pi'),G_i)}.
   \end{align}
Here $\RPE(\SS)=0$ means a perfect fit as SW is a distance between distributions, $\RPE(\SS)<1$ means that the model fits better than the reference model, and $\RPE(\SS)\geq 1$ means no improvement over the reference model. The reference model is fitted once on the training data and the same reference is used for all models that are compared, so that RE values of different models share the same denominator. In practice, we will approximate SW distances using Monte Carlo samples as discussed. Later, in Section~\ref{subsec:parsimonious_model}, we introduce a specific example of a reference model that ignores the predictors and defines $f'_k\sharp F_{ik}=\delta_{b_k}$ using a constant intercept $b_k$ only.

\subsection{Likelihood Sample Complexity and Posterior Consistency}
\label{subsec:posterior_consitency}
We investigate two theoretical aspects of the proposed model. The
first  result characterizes  sample complexity of the
generalized likelihood
when distributions are estimated through i.i.d samples.
The second  result establishes posterior
consistency, both for the posterior based on the underlying distributions and for the posterior computed from samples. In short, with increasing sample size the posterior concentrates on the set of risk minimizers under parametric settings of regression functions, which are the true parameters when they are identified. Throughout this subsection, we consider distributions supported in the ball $B_R:=\{x\in\Re^d:\|x\|_2\leq R\}$ for some $R>0$, denote by $\PP(B_R)\subseteq\PP_p(\Re^d)$ the set of such distributions, and take the minimum in \eqref{eq:SWB} over $\PP(B_R)$ (a minimizer always exists). Since all projections of distributions in $\PP(B_R)$ are supported in $[-R,R]$, every SW distance between them is at most $D:=2R$.

We first start with the sample complexity of  a simple plug-in estimator of  the (generalized) likelihood in \eqref{eq:MDDRmodel}.

\begin{theorem}
  \label{theorem:sample_complexity}
  For a radius $R>0$, let $F_1,\ldots,F_K,G \in \PP(B_R)$
  ($p\geq 1$),
  and
  $\hat{F}_1,\ldots,\hat{F}_K,\hat{G}$ are the corresponding empirical
  distributions with at least $n$ i.i.d samples, $(\pi_1,\ldots,\pi_K) \in
  \Delta^K$, $\bar{F} =
  \SWB_p(F_1,\ldots,F_K,\pi_1,\ldots,\pi_K)$, and $\hat{\bar{F}}$ is any minimizer defining
  $\SWB_p(\hat{F}_1,\ldots,\hat{F}_K,\pi_1,\ldots,\pi_K)$. Assume that there is $\kappa>0$ such that, for all $F_1',\ldots,F_K'\in\PP(B_R)$, every minimizer $\bar F'$ defining $\SWB_p(F_1',\ldots,F_K',\pi_1,\ldots,\pi_K)$ satisfies
  \begin{align}
      \SW_p^p(\bar F,\bar F')\leq \kappa\sum_{k=1}^K\pi_k\SW_p(F_k,F_k'). \label{eq:stability}
  \end{align}
  The following inequality holds:
    \begin{align}
        \EE\left[\left|\exp \{-w  \SW_p^p(\bar{F},G)\} - \exp \{-w  \SW_p^p(\hat{\bar{F}},\hat{G})\} \right|\right] \leq C_{p,w,R,\kappa} \frac{1}{n^{1/(2p^2)}},
    \end{align}
    where $C_{p,w,R,\kappa}=wp(2R)^{p-1}\big((2R\kappa)^{1/p}+2R\big)\sqrt{\pi/2}$.
\end{theorem}

The proof 
is given in
Supplementary Material~\ref{subsec:proof:theorem:sample_complexity}. Inequality \eqref{eq:stability} assumes that the barycenter is stable under perturbations of the marginals. It is a form of uniqueness of the SWB and is only needed to obtain a rate. It holds, for example, when the SWB problem has a quadratic growth property. In particular, for $p=2$ it holds when the marginals are rescaled and translated versions of a common distribution (Supplementary Material~\ref{subsec:stability_sufficient}).
Theorem~\ref{theorem:sample_complexity} suggests that the empirical
estimation of the generalized likelihood converges at the order of
$\mathcal{O}(n^{-1/(2p^2)})$, which is a parametric rate when $p=1$. It applies directly to the push-forwards of the predictors, since $f_k\sharp\hat F_k$ is the empirical distribution of i.i.d samples from $f_k\sharp F_k$.
We again recall that  the result does not require  any continuity
assumption on the population distributions $F_1,\ldots,F_K,G$, i.e.,
they can be discrete or continuous.

Next, we discuss posterior consistency.
 We assume that the push forward functions are indexed by parameters
$\phi=(\phi_1,\ldots,\phi_K) \in \Phi:=\Phi_1\times \ldots\times\Phi_K$ , as  $f_\phi = (f_{1,\phi_1},\ldots,f_{K,\phi_K})$ and we are interested in $p(\phi,\pi\mid \SS)$. Let $P$ denote a true generating model for the predictors $(F_1,\ldots,F_K)$ and a true response $G$. Here, $P$ is a distribution on tuples of distributions equipped with the product of the topologies of weak convergence, and $\mathrm{supp}(P)$ is its support. Recognizing the negative log likelihood in \eqref{eq:MDDRmodel} as a loss function we define empirical risk and population risk
\begin{align}
    &R_N (\phi,\pi) = \frac{1}{N}\sum_{i=1}^N \SW_p^p\{\Gt_i(\phi,\pi),G_i\}, \quad
    R(\phi,\pi) = \EE_{(F_1,\ldots,F_K,G)\sim P}[\SW_p^p\{\Gt(\phi,\pi),G\}].
\end{align}
We make the following assumptions for posterior consistency. First, we require some assumptions.

\begin{assumption}[Compactness]\label{assumption:compact}
The parameter space $\Phi$ is compact.
\end{assumption}

\begin{assumption}[Regularity of regression functions]\label{assumption:continuity} There are compact sets $\mathcal{X}_k\subset\Re^{h_k}$ such that for all $(F_1,\ldots,F_K,G) \in \mathrm{supp}(P)$, $F_k$ is supported in $\mathcal{X}_k$ for $k=1,\ldots,K$, and $G$ is supported in $B_R$. For $k=1,\ldots,K$, the map $(\phi_k,x)\mapsto f_{k,\phi_k}(x)$ is continuous on $\Phi_k\times\mathcal{X}_k$ with values in $B_R$. For every $(\phi,\pi)\in\Phi\times\Delta^K$ and $(F_1,\ldots,F_K,G)\in\mathrm{supp}(P)$, the minimizer over $\PP(B_R)$ in \eqref{eq:SWB} with marginals $f_{1,\phi_1}\sharp F_1,\ldots,f_{K,\phi_K}\sharp F_K$ and weights $\pi$ is unique.
\end{assumption}

Let $\Theta := \Phi\times \Delta^K$ with the metric $\rho\big((\phi,\pi),(\phi',\pi')\big) = \|\phi-\phi'\|_p+\|\pi-\pi'\|_p$, and let $\rho\big((\phi,\pi),\Theta'\big)=\inf_{(\phi',\pi')\in\Theta'}\rho\big((\phi,\pi),(\phi',\pi')\big)$ for any $\Theta'\subset\Theta$. Under Assumptions~\ref{assumption:compact} and~\ref{assumption:continuity}, the population risk $R$ is continuous on the compact set $\Theta$ (see Supplementary Material~\ref{subsec:proof:theorem:posterior_consistency}). Therefore, the set of risk minimizers
\begin{align}
    \Theta_0 := \argmin_{(\phi,\pi)\in\Theta} R(\phi,\pi)
\end{align}
is non-empty and compact. We do not require $\Theta_0$ to be a singleton.

\begin{assumption}[Prior positivity]\label{assumption:prior}
The prior assigns positive mass to every neighborhood of some $(\phi_0,\pi_0)\in\Theta_0$. In particular, for every
$\epsilon > 0$, we have
$
    p(B_\epsilon(\phi_0,\pi_0)) > 0,
    \quad \text{where} \quad
    B_\epsilon(\phi_0,\pi_0) = \{(\phi,\pi) \in \Theta : \rho\big((\phi,\pi),(\phi_0,\pi_0)\big)< \epsilon\}.
$
\end{assumption}

\begin{theorem}\label{theorem:posterior_consistency}
 Under
Assumptions~\ref{assumption:compact}–\ref{assumption:prior},
for every $\epsilon>0$, the posterior measure $p_N$
 satisfies
\begin{align}
    p_N\big(\{ (\phi,\pi) \in \Theta: \rho\big((\phi,\pi),\Theta_0\big) \geq \epsilon \}\big) \xrightarrow{a.s} 0 ,
\end{align}
as $N \to \infty$, almost surely under i.i.d sampling from the true generating model $P$.
\end{theorem}
The proof of Theorem~\ref{theorem:posterior_consistency} is provided in Supplementary Material~\ref{subsec:proof:theorem:posterior_consistency}. In particular, the continuity of the SWB under uniqueness assumption (Lemma~\ref{lemma:continuity_SWB} in Supplementary Material~\ref{subsec:proof:theorem:posterior_consistency}) is important to establish a uniform law of large numbers for $R_N(\phi,\pi)$. This enables us to prove consistency under the stated assumptions. Theorem~\ref{theorem:posterior_consistency} states that the posterior concentrates on the set of risk minimizers.

Theorem~\ref{theorem:posterior_consistency} concerns the posterior based on the underlying distributions. However, in practice, we might computing posterior when replacing $F_{ik}$ and $G_i$ by the empirical distributions $\hat F_{ik}^{(n)}$ and $\hat G_i^{(n)}$ of $n$ i.i.d samples from them. We denote this posterior by $\hat p_N^{(n)}$, i.e., $R_N$ and $\Gt_i(f,\pi)$ are replaced by
\begin{align}
    &\hat R_N^{(n)}(\phi,\pi)=\frac{1}{N}\sum_{i=1}^N\SW_p^p\{\widehat{\Gt}_i(\phi,\pi),\hat G_i^{(n)}\}, \\ &\widehat{\Gt}_i(\phi,\pi)=\SWB_p(f_{1,\phi_1}\sharp\hat F_{i1}^{(n)},\ldots,f_{K,\phi_K}\sharp\hat F_{iK}^{(n)},\pi).
\end{align}

\begin{theorem}\label{theorem:approx_posterior_consistency}
Under the assumptions of Theorem~\ref{theorem:posterior_consistency}, assume that, given the units $(F_{i1},\ldots,F_{iK},G_i)$, $i=1,\ldots,N$, the samples are i.i.d from the corresponding distributions, and that the units together with their samples are i.i.d across $i$. If $n=n_N\to\infty$ as $N\to\infty$, at any rate, then for every $\epsilon>0$,
\begin{align}
    \hat p_N^{(n_N)}\big(\{ (\phi,\pi) \in \Theta: \rho\big((\phi,\pi),\Theta_0\big) \geq \epsilon \}\big)\xrightarrow{a.s}  0,
\end{align}
as $N\to\infty$, almost surely under the true generating model.
\end{theorem}
The proof of Theorem~\ref{theorem:approx_posterior_consistency} is given in Supplementary Material~\ref{subsec:proof:theorem:approx_posterior_consistency}. The theorem does not require any rate relation between $N$ and $n$. The two sources of asymptotics interact only through a plug-in error, which depends on $n$ but not on $N$ and vanishes as $n\to\infty$. 
 We would like to recall that we sometimes observe distributions in discrete population forms, hence, plug-in estimators are not needed. 

The model is stated for distributions $F_{ik}$ and $G_i$ because they are the estimands i.e., the regression does not depend on how many observations are available for a unit. So, units with different sample sizes are directly comparable. The empirical distributions are plug-in estimators. Theorem~\ref{theorem:approx_posterior_consistency} shows that the plug-in error vanishes as $n\to\infty$, and Theorem~\ref{theorem:sample_complexity} gives a dimension-free rate under \eqref{eq:stability} (a fully observed-data formulation would also model the sampling of cells within units, however, since the plug-in error is already small for moderate $n$ (Section~\ref{sec:simulation}), we do not pursue this.)

\subsection{Posterior Inference}
\label{subsec:posterior_inference}

The generalized posterior is given by:
\begin{align}
  p(f,\pi \mid  \SS) &\propto
  p(f,\pi)\,
  \prod_{i=1}^N \ell(f_1,\ldots,f_K,\pi_1,\ldots,\pi_K;G_i,F_{i1},\ldots,F_{iK}).
\end{align}
 We implement posterior Markov chain Monte Carlo (MCMC) simulation
using Metropolis-Hastings (MH) transition probabilities
with proposals to mimic
$p(f \mid\pi,  \SS)$  and $p(\pi\mid f,\SS)$. 
In particular, we construct proposals $f^*\sim q(f^*\mid
f)$ and $\pi^*\sim q(\pi^*\mid \pi)$, and accept with probabilities: 
\begin{align}
    \min \left\{1,\frac{p(f^*\mid \pi,\SS)q(f\mid f^*)}{p(f\mid \pi,\SS)q(f^*\mid f)}\right\}, \quad \min \left\{1,\frac{p(\pi^*\mid f,\SS)q(\pi\mid \pi^*)}{p(\pi\mid f,\SS)q(\pi^*\mid \pi)}\right\},
\end{align}
 respectively. 
In practice, we only observe i.i.d samples from distributions and approximate the generalized likelihood as:
\begin{align}
   & \ell(f_1,\ldots,f_K,\pi_1,\ldots,\pi_K;G_i,F_{i1},\ldots,F_{iK}) \nonumber\\ &\approx \ell(f_1,\ldots,f_K,\pi_1,\ldots,\pi_K;\hat{G}_i,\hat{F}_{i1},\ldots,\hat{F}_{iK})
    = \exp \left\{-w  \SW_p^p\left[\widehat{\Gt}_i(f,\pi),\hat{G}_i\right]\right\}, 
\end{align}
where  $\widehat{\Gt}_i(f,\pi)=\SWB_p(f_1\sharp
\hat{F}_{i1},\ldots,f_K\sharp \hat{F}_{iK}, \pi_1,\ldots,\pi_K)$, and
$\hat{G}_i,\hat{F}_{i1},\ldots,\hat{F}_{iK}$ are empirical
distributions of $G_i,F_{i1},\ldots,F_{iK}$. As  stated  in
Theorem~\ref{theorem:sample_complexity}, this approximation converges
well. Next, we approximate the SWB
 as discussed in Supplementary Materials~\ref{sec:computation_SWB},
approximating $\widehat{\Gt}_i(f,\pi)$ by
$\widehat{\Gt}_i^{(T)}(f,\pi)$ defined by way of an iterative algorithm with $T>0$ iterations. We denote the approximated
generalized likelihood as: 
\begin{align}
  \hat{\ell}(f_1,\ldots,f_K,\pi_1,\ldots,\pi_K;G_i,F_{i1},\ldots,F_{iK},T)=
  \exp
    \left\{-w  \SW_p^p\left[\widehat{\Gt}_i^{(T)}(f,\pi),\hat{G}_i\right]\right\}. 
\end{align}
The projections and the SWB initialization $z^{(0)}$ (the initial atoms of $\widehat{\Gt}_i$, see Supplementary Material~\ref{sec:computation_SWB}) are drawn once and kept fixed, so that $\hat{\ell}$ is a deterministic function of $(f,\pi)$ and the chain targets a fixed approximate posterior. The same fixed projections and initialization are used when computing RE.

\textbf{ Updating $\phi$:}
We now discuss the construction of the proposal distributions
 for $\phi$ and $\pi$, respectively.
We parameterize  the regression function $f_\phi$ as
$f_\phi=(f_{1,\phi_1},\ldots,f_{K,\phi_K})$ with $\phi=(\phi_1,\ldots,\phi_K)$.
We use a MALA proposal $q(\phi^* \mid
\phi) \propto \exp\left(-\frac{1}{4\eta_1} \|\phi^* - 
 \phi - \eta_1 \nabla_\phi \log \hat{p}(\phi \mid \pi,\SS)\|_2^2\right)$,
for a fixed step size $\eta_1 > 0$, and 
\begin{align}
    \nabla_\phi \log p(\phi \mid  \pi,\SS) &\approx \nabla_\phi \log
    p(\phi)  -w \sum_{i=1}^N   \nabla_\phi
    \EE\left\{W_p^p\left[
      \theta \sharp \widehat{\Gt}_i^{(T)}(\phi,\pi),\theta\sharp
      \hat{G}_i\right]
    \right\}, \nonumber \\
    &=\nabla_\phi \log p(\phi)  -w \sum_{i=1}^N
    \EE\left\{ \nabla_\phi
    W_p^p\left[\theta \sharp \widehat{\Gt}_i^{(T)}(\phi,\pi),\theta\sharp \hat{G}_i\right]\right\}.
\end{align}
Let $\psi_{i,\theta}^\star$ be the optimal transport plan between $\theta \sharp \widehat{\Gt}_i^{(T)}(\phi,\pi)$ and $\theta\sharp \hat{G}_i$, and $M_i$ be the number of atoms $z_{i,j}^{(T)}(\phi,\pi)$ of $\widehat{\Gt}_i^{(T)}(\phi,\pi)$. We have:
\begin{align}
   & \nabla_\phi W_p^p(\theta \sharp \widehat{\Gt}_i^{(T)}(\phi,\pi),\theta\sharp \hat{G}_i) = \nabla_\phi\sum_{j=1}^{M_{i}} \sum_{j'=1}^{M_{G_i}} \psi_{i,\theta,jj'}^\star |\theta^\top(z_{i,j}^{(T)}(\phi,\pi) -y_{i,j'})|^p \\
    &=p
\sum_{j=1}^{M_{i}} \sum_{j'=1}^{M_{G_i}}
\psi_{i,\theta,jj'}^\star
\left|\theta^\top(z_{i,j}^{(T)}(\phi,\pi) -y_{i,j'})\right|^{p-2}
\left( \theta^\top(z_{i,j}^{(T)}(\phi,\pi) -y_{i,j'})\right) \nonumber\\&\left(\frac{\partial z_{i,j}^{(T)}(\phi,\pi)}{\partial \phi }\right)^\top \theta ,
\end{align}
where $\frac{\partial z_{i,j}^{(T)}(\phi,\pi)}{\partial \phi }$ is discussed in Section~\ref{sec:computation_SWB}. We sample $\phi^* \sim q(\phi^* \mid \phi)$ as
$\phi^* = \phi + \eta_1 \nabla_\phi \log p(\phi \mid  \pi,\SS) +
 \sqrt{2\eta_1} \epsilon_0$,
with $\epsilon_0 \sim \mathcal{N}(0, I)$, a standard multivariate
Gaussian distribution  of  dimension matching the parameter
$\phi$.

\textbf{ Updating $\pi$:}
 To construct a gradient based transition probability for $\pi$ we
use a  
change of variables to remove the simplex constraint of $\pi$:
\begin{align}
    &\pi_k =  \frac{e^{\varpi_k}}{1+\sum_{k'=1}^{K-1}
    e^{ \varpi_{k'}}} \quad \text{for} \quad k=1,\ldots,K-1, \quad
  \pi_K =  \frac{1}{1+\sum_{k'=1}^{K-1} e^{\varpi_{k'}}}. 
\end{align}
 Let $D=\diag(\pi_1,\ldots,\pi_{K-1})$. 
We have the Jacobian $J =  D - \pi_{-K}\pi_{-K}^\top$, with $\pi_{-K}=(\pi_1,\ldots,\pi_{K-1})^\top$, and $\det (J)= \det(D) (1-\pi_{-K}^\top  D^{-1} \pi_{-K}) =  \prod_{k=1}^K \pi_k$,
implying  
$p(\varpi) = p(\pi) \prod_{k=1}^K \pi_k$.
We use 
the proposal $q(\varpi^* \mid \varpi) \propto \exp\left(-\frac{1}{4\eta_2} \|\varpi^* -
 \varpi - \eta_2 \nabla_\varpi \log p(\varpi \mid \phi,\SS)\|_2^2\right)$,
for a fixed step size $\eta_2 > 0$, and 
\begin{align}
       \nabla_\varpi \log p(\varpi \mid  \phi,\SS) &= \nabla_\varpi \log
    p(\varpi)  -w \sum_{i=1}^N   \nabla_\varpi
    \EE\left\{W_p^p\left[
      \theta \sharp \widehat{\Gt}_i^{(T)}(\phi,\varpi),\theta\sharp
      \hat{G}_i\right]\right\}, 
\end{align}
which can be derived with the gradient  $\frac{\partial  z_{i,j}^{(T)} (\phi,\pi)}{\partial \pi}$ as  discussed in Section~\ref{sec:computation_SWB}. Since $\varpi$ and $\pi$ are connected through invertible deterministic mappings, sampling for $\varpi$ implies sampling for $\pi$.

\subsection{A Parsimonious Model}
\label{subsec:parsimonious_model}

We specify a parsimonious model for use in later simulation and data analysis. First, we use a linear parameterization for regression functions:
\begin{align}
    f_{k,\phi_k}(x)=f_{k,A_k,b_k}(x) = A_kx +b_k, \,\forall k=1,\ldots,K,
     \label{eq:linear_regression}
\end{align}
where $A_k \in \mathbb{R}^{d \times h_k}$ is the coefficient matrix and $b_k \in \mathbb{R}^d$ is the intercept. With a compact parameter space and compactly supported predictors, \eqref{eq:linear_regression} satisfies Assumption~\ref{assumption:compact} and the continuity part of Assumption~\ref{assumption:continuity} for $R$ large enough. The uniqueness of the SWB in Assumption~\ref{assumption:continuity} is an additional assumption. For $p=2$ it holds, for example, under the condition of Proposition~\ref{proposition:QG_sufficient} in Supplementary Material~\ref{subsec:stability_sufficient}.  As discussed in  Supplementary Material~\ref{sec:computation_SWB} and Section~\ref{subsec:posterior_inference}, we have  $ \frac{\partial f_{k,A_k,b_k}(x)}{\partial A_k}:= I_{d} \otimes x^\top$ (with $A_k$ vectorized row-wise to match the vector form of $\phi$ in Section~\ref{sec:computation_SWB}) and $ \frac{\partial f_{k,A_k,b_k}(x)}{\partial b_k}=I_d$  for $k=1,\ldots,K$ for posterior simulation.

We complete the model  with a prior on  the regression functions:
\begin{align}
        &A_k[ij] \sim Laplace(0,1),\, b_k[i] \sim \mathcal{N}(0,10^3), \,\forall i=1,\ldots,d, \, \forall j =1,\ldots,h_k,
\end{align}
and the barycenter weights:
\begin{align}
    (\pi_1,\ldots,\pi_K) \sim Dir(\alpha),\, \alpha =(0.01,\ldots,0.01),
\end{align}
implying the following prior on $\varpi$: $p(\varpi) =
\frac{1}{B(\alpha)} \prod_{k=1}^K \pi_k^{\alpha_{k}}$.
Any alternative family of differentiable functions $f_{k,\phi_k}$
could be used, such as deep neural
networks~\citep{wilson2020bayesian}, if desired or required. We adopt
the linear form because it is interpretable, parsimonious and, as shown later,
sufficiently expressive for our analysis. For RE~\eqref{eq:RE}, we define the reference model as $f'_k(x) =  b_k$, i.e., an intercept model that ignores predictors. For $p=2$, the push-forwards $f_k'\sharp F_{ik}=\delta_{b_k}$ are point masses. So, the reference fitted distribution is $\Gt_i'=\delta_{c}$ with $c=\frac{1}{K}\sum_{k=1}^K b_k$. Since $\SW_2^2(\delta_c,G)=\frac{1}{d}\big(\|c-m_G\|_2^2+\mathrm{Tr}(\Sigma_G)\big)$, where $m_G$ and $\Sigma_G$ are the mean and covariance of $G$, fitting the reference model on the training units gives $c=\frac{1}{N}\sum_{i=1}^N m_{G_i}$, the average of the mean responses. We use this fitted reference for all models.

\begin{remark}
\label{remark:invariance}
For $p=2$, $K\geq 2$, and linear regression functions $f_{k}(x)=A_kx+b_k$, $\Theta_0$ is in general not a singleton. For $s_k>0$ and $t_k\in\Re^d$ with $\sum_{k=1}^K\pi_k/s_k=1$ and $\sum_{k=1}^K(\pi_k/s_k)t_k=0$, if we replace $(A_k,b_k,\pi_k)$ by $(s_kA_k,s_kb_k+t_k,\pi_k/s_k)$, the SWB and the risk are unchanged (Supplementary Material~\ref{subsec:proof:invariance}) when the transformed parameters remain in $\Phi$. The scale-normalized weights $\tilde{\pi}_k = \pi_k\|A_k\|_F/\sum_{j=1}^K\pi_j\|A_j\|_F$ and the weighted intercept $\bar b=\sum_{k=1}^K\pi_kb_k$ are invariant under these transformations, and $\tilde\pi=\pi$ whenever $\|A_1\|_F=\ldots=\|A_K\|_F$.
\end{remark}

\section{Simulation}
\label{sec:simulation}
We assess the proposed MDDR model in a simulation study. In particular, we highlight how the use of the SWB generates multi-variate responses
from  lower-dimensional  predictors
 (in this case, bivariate responses from predictors with 1-dimensional
structure). 

We first sample $A_1,A_2,A_3 \overset{i.i.d}{\sim} \mathcal{U}(SO(2))$
and we set $(\pi_1,\pi_2,\pi_3)=(2/3,1/6,1/6)$.  We then
 generate a simulation truth and hypothetical data by the
following steps: 
\begin{enumerate}
    \item We sample $v_1,v_2,v_3\overset{i.i.d}{\sim} \mathcal{U}(\mathbb{S})$ from the unit circle $\mathbb{S}$,  $\mu_1,\mu_2,\mu_3 \overset{i.i.d}{\sim} \mathcal{N}(0,9)$, and $\sigma_1^2,\sigma_2^2,\sigma_3^2 \overset{i.i.d}{\sim} IG(3,1)$.
        \item We define predictors $F_1 = f_{v_1}\sharp \mathcal{N}(\mu_1,\sigma_1^2)$, $F_2 = f_{v_2}\sharp \mathcal{N}(\mu_2,\sigma_2^2)$, $F_3 = f_{v_3}\sharp \mathcal{N}(\mu_3,\sigma_3^2)$ where $f_v(x) = vx$, namely, predictors are univariate Gaussians lifted to 2 dimensions. We  then define $\hat{F}_1,\hat{F_2},\hat{F}_3$ to be empirical distributions over $100$ i.i.d samples  of $F_1,F_2,F_3$ respectively.

    \item We define $\Gt = \SWB_2(f_{A_1} \sharp F_1,f_{A_2} \sharp F_2,f_{A_3} \sharp F_3, \pi_1,\pi_2,\pi_3)$ where $f_{A}(x)=Ax$, and define the response $G = \Gt* \mathcal{N}(0,0.01)$. Since $G$ is intractable, we obtain an empirical version of $G$: $\hat{G} = \SWB_2(f_{A_1} \sharp \hat{F}_1,f_{A_2} \sharp \hat{F}_2,f_{A_3} \sharp \hat{F}_3, \pi_1,\pi_2,\pi_3)* \mathcal{N}(0,0.01 I_2)$ where the SWB is approximated with $100$ atoms as discussed in Section~\ref{sec:computation_SWB}.

\end{enumerate}

\begin{table}[!t]
    \centering
    \scalebox{0.9}{
    \begin{tabular}{c|c|c|c|c}
    \toprule
         Model&  Train RE & Train RE 95\% HCI & Test RE & Test RE 95\% HCI  \\
         \midrule
        DDR & $0.2485 \pm 0.0171$ & 0.2365-0.2600 & $0.3069 \pm 0.0819$ & 0.2791-0.3399 \\
        Aggregated DDR & $0.4293 \pm 0.0233$ & 0.4062-0.4496 & $0.4388 \pm 0.0492$ & 0.4131-0.4654 \\
        MDDR & $\mathbf{0.0175 \pm 0.0008}$ & \textbf{0.0079-0.0283} & $\mathbf{0.0229 \pm 0.0065}$ & \textbf{0.0098-0.0410} \\
        \bottomrule
    \end{tabular}
    }
    \caption{Comparison of the single-predictor DDR model, the aggregated DDR model, and the multi-predictor MDDR model with simulated data. Reported are training and test relative errors (RE) averaged over 3 replicated datasets, and the 95\% highest credible intervals (HCI) averaged over the replications.}
    \label{table:simulation_result}
\end{table}

Although the theory in Section~\ref{subsec:posterior_consitency} assumes compact supports, we use Gaussian distributions for simplicity. They can be viewed as approximations of their truncations to a large ball.
 Steps 1-3 generate one observation
$(\Fh_{i,1},\ldots,\Fh_{i,K},\Gh_i)$.
We repeat the process  70 times to obtain a dataset. We select $70\%$ of samples for the training set $\SS=\{(\hat{F}_{11},\hat{F}_{12},\hat{F}_{13},\hat G_1),\ldots,(\hat{F}_{N1},\hat{F}_{N2},\hat{F}_{N3},\hat G_N) \}$ (with $N=49$) and the other $30\%$ of samples for the testing set. We generate 3 such datasets with different random seeds and report results averaged over them. We compare MDDR with two baselines. The first one is DDR~\citep{nguyen2025bayesian}, for which we use a linear regression function to regress the response on the first predictor since it does not allow multiple predictors. The second one is an aggregated DDR model: we fit three separate DDR models, one for each predictor, and define the fitted distribution as the SWB of the three DDR fitted distributions with uniform weights $(1/3,1/3,1/3)$. Aggregated DDR uses all predictors, but its regression functions are estimated separately and its weights are uniform. For all models,  we use $w=10$ for the generalized likelihood and run MALA, discussed in Section~\ref{subsec:posterior_inference}, for 500 iterations, discarding the first half as burn-in. We use step sizes $\eta_1 = 0.0005$ and $\eta_2 = 0.005$ for MDDR, and $\eta_1=0.0005$ for DDR. SW distances are approximated using 1000 Monte Carlo projections, which are drawn once and kept fixed, both during inference and for computing RE. For computing SWB, we use $T=100$, $\eta=0.1$, and $(\beta_1,\beta_2)=(0.9,0.999)$. DDR takes about 10 minutes while MDDR takes about 90 minutes on a NVIDIA GeForce 3060.

\begin{figure}[!t]
\begin{center}
    \begin{tabular}{cccc}
  \widgraph{0.23\textwidth}{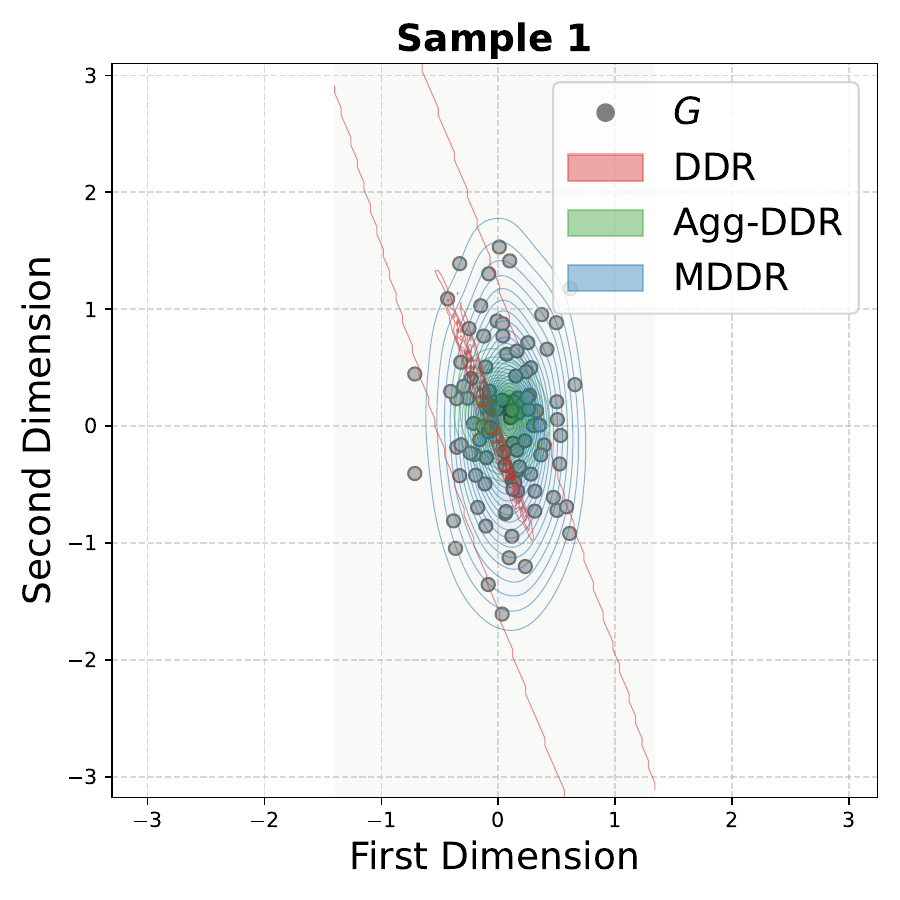} 
&
\widgraph{0.23\textwidth}{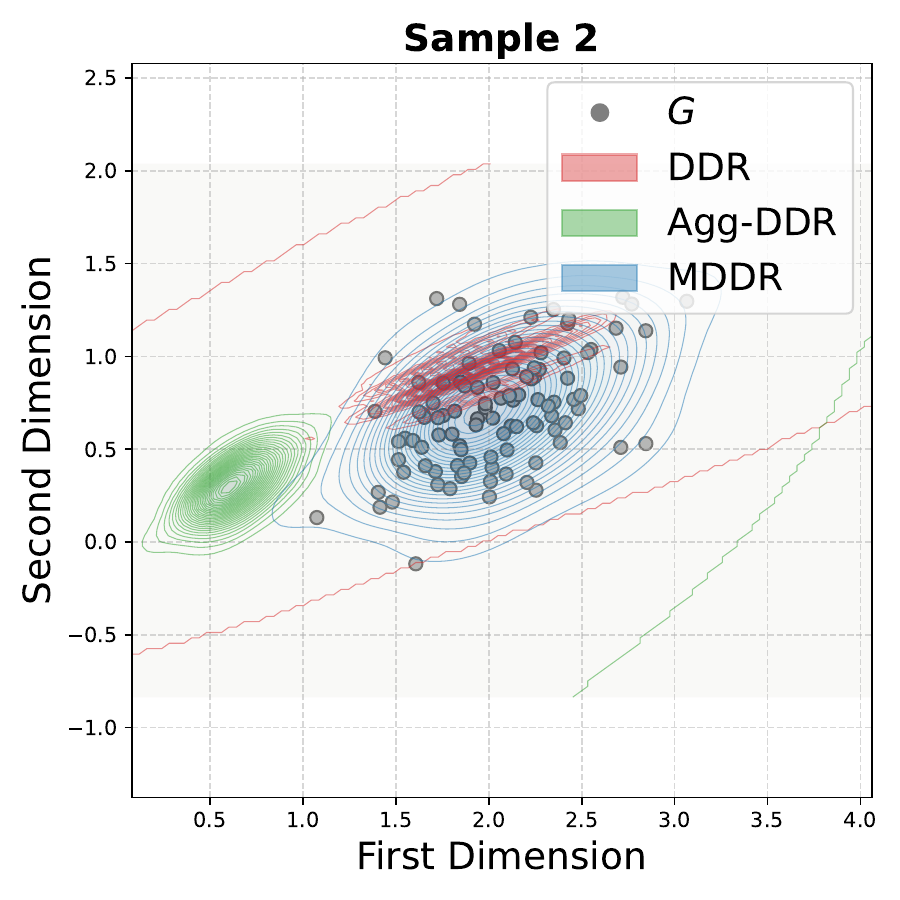} 
&
\widgraph{0.23\textwidth}{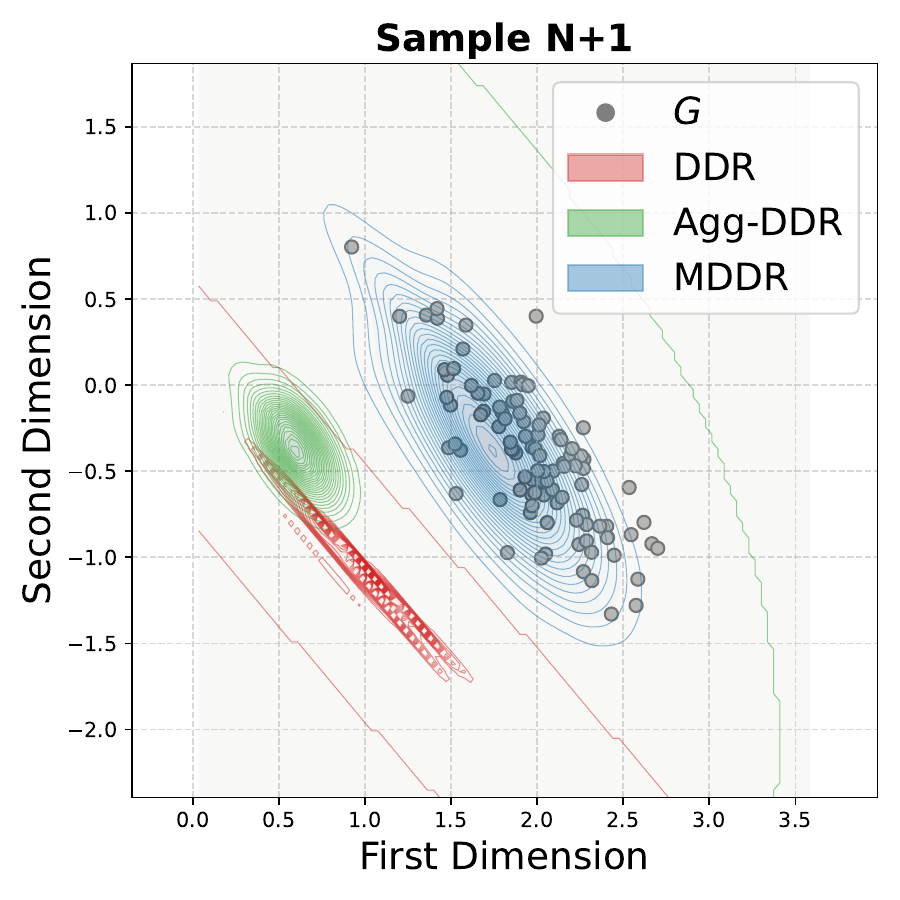} 
&
\widgraph{0.23\textwidth}{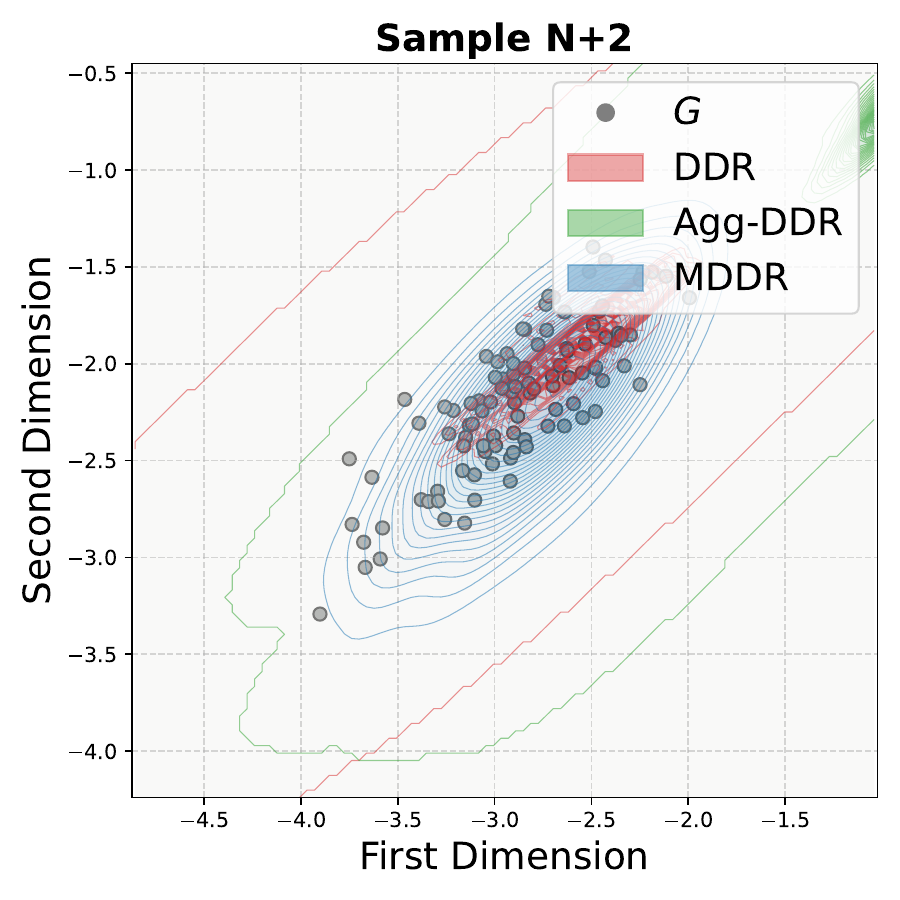}  \\
  \widgraph{0.23\textwidth}{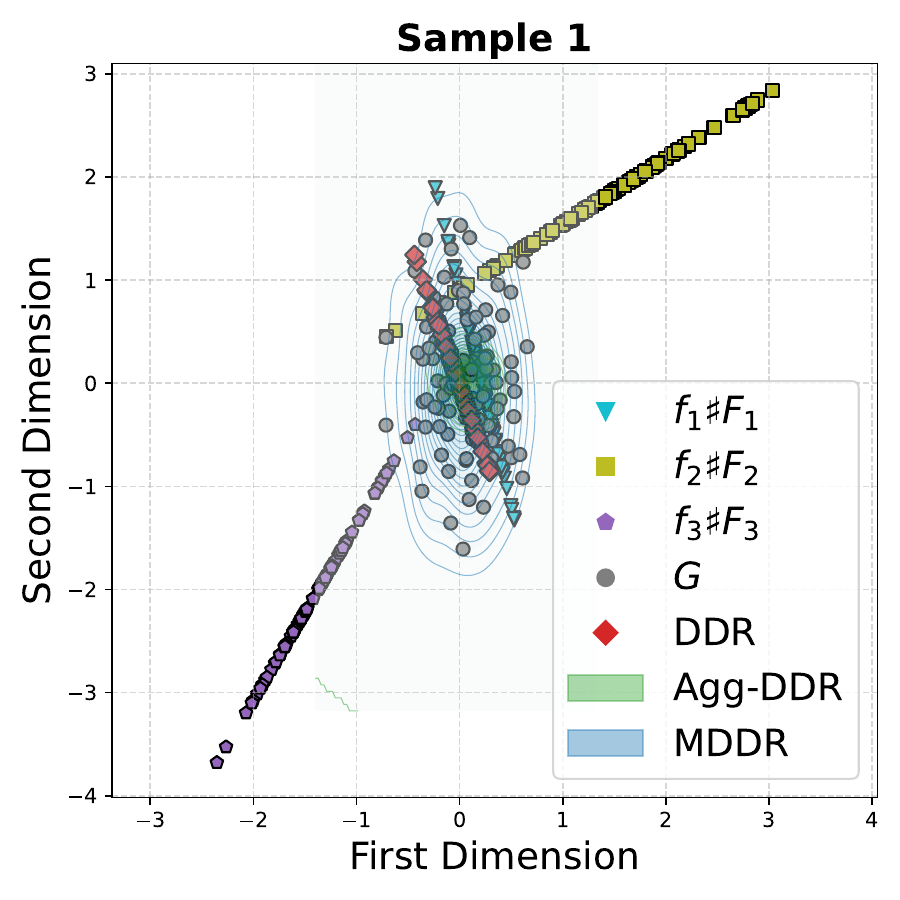} 
&
\widgraph{0.23\textwidth}{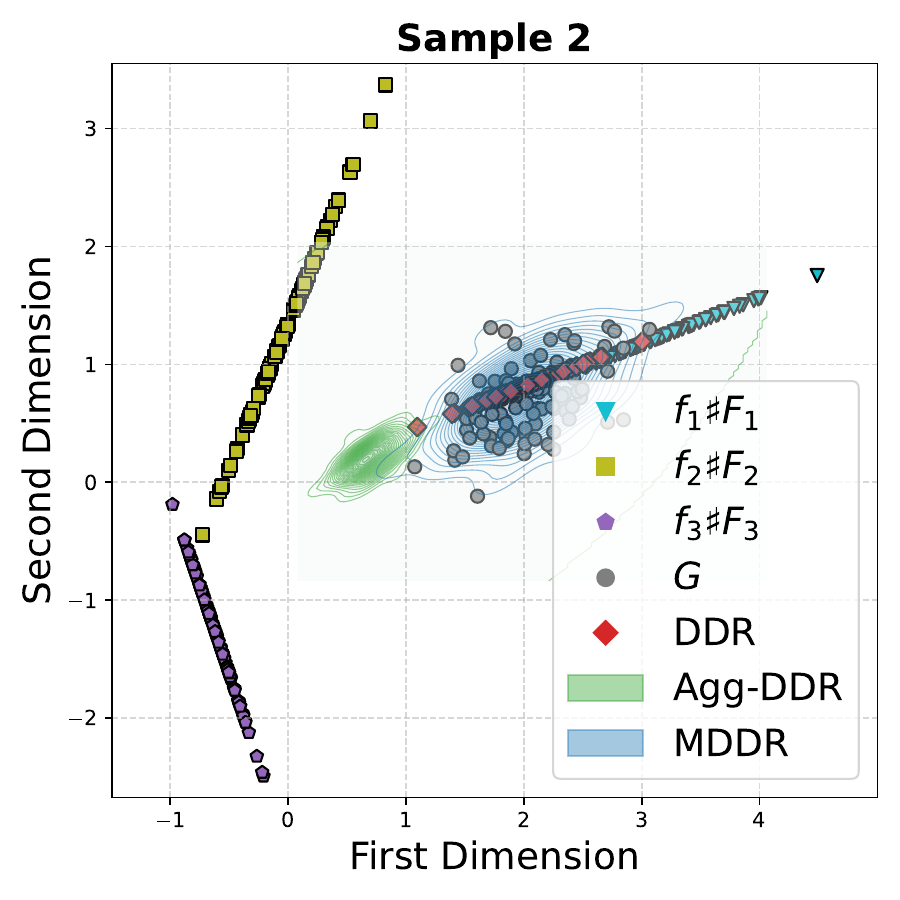} 
&
\widgraph{0.23\textwidth}{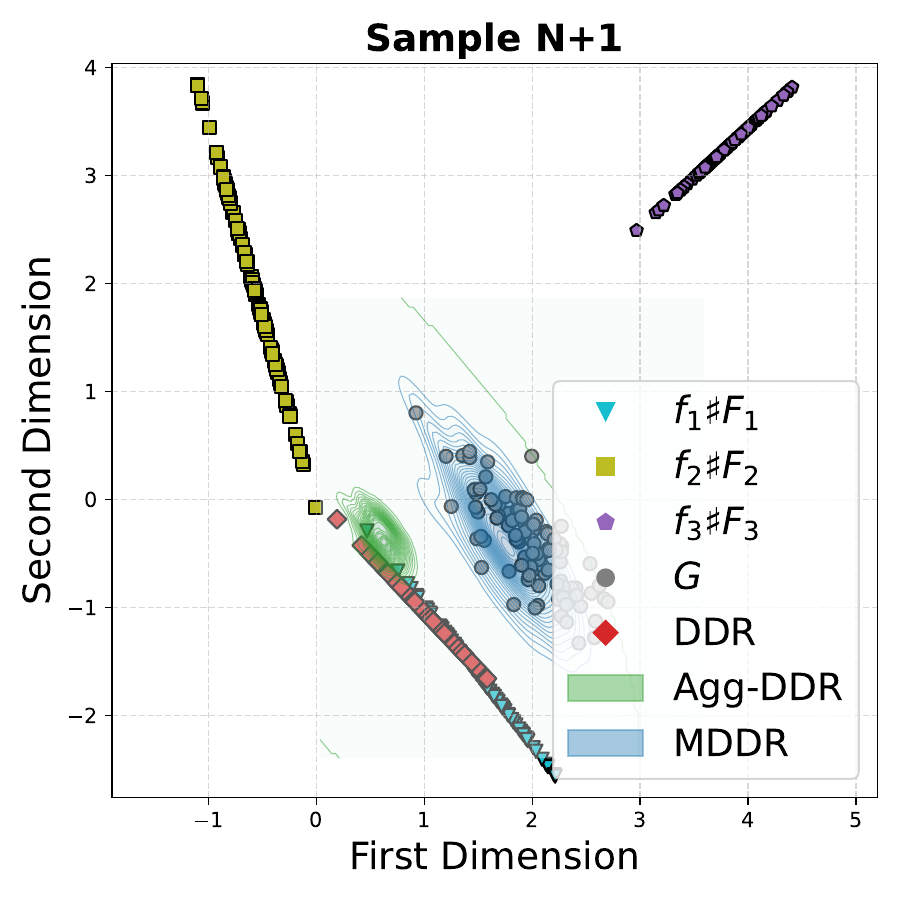} 
&
\widgraph{0.23\textwidth}{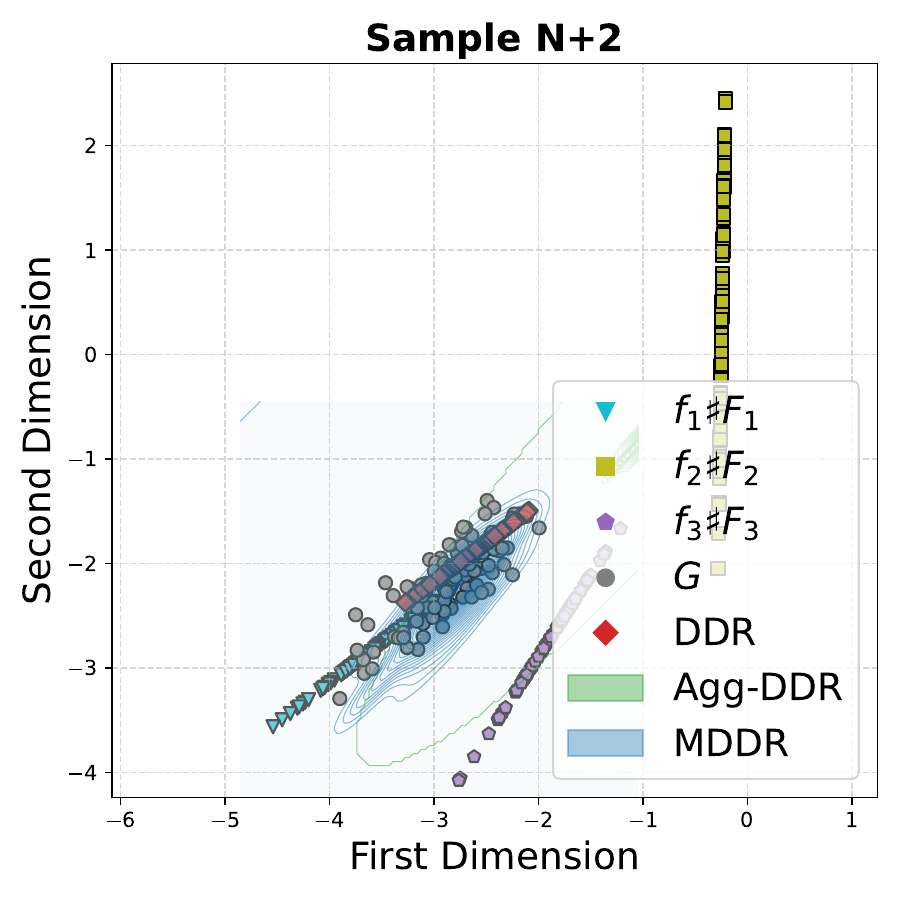}
  \end{tabular}
  \end{center}
  \vspace{-0.2 in}
  \caption{
    \footnotesize{ The first two columns display two randomly selected
      in-sample regression units  ($i=1,2$), 
      and the last two columns show two
      randomly selected out-of-sample examples 
      ($i=N+1,N+2$). \\
      The first row presents
      the observed responses  $G_i$, 
      along with the posterior mean  fitted distributions 
      $\EE(\Gt_i \mid  data)$  for DDR, aggregated DDR, and MDDR.\\
            The second row shows the observed
      responses, the fitted distributions  $\Gt_i$  for DDR, aggregated DDR,
      and MDDR under the final posterior sample of the chain, and the
      push-forward  $f_{\phi_k} \sharp \Fh_{i,k}$  of the
      predictors based on that final posterior sample.
}
} 
  \label{fig:Gaussian_simulation}
\end{figure}

We report the RE~\eqref{eq:RE} for the training set, the RE for the
testing set, and their associated 95\% highest posterior  credible intervals
(HCIs) in Table~\ref{table:simulation_result}. From the table, we
observe that MDDR performs significantly better than DDR and aggregated DDR in both
in-sample fitting and out-of-sample prediction, in all three replications. Aggregated DDR performs worse than DDR on the first predictor, since it gives equal weights to the three predictors while the first predictor has the largest true weight. This shows that combining predictors is not sufficient. The regression functions and the weights need to be learned jointly as in MDDR. An
illustration of the results is provided in
Figure~\ref{fig:Gaussian_simulation}. From the first row of the
figure, we see that MDDR yields more accurate posterior means in both
the in-sample and out-of-sample settings.
 We show the posterior mean $\EE(\Gt_i \mid F_i, data)$ of fitted
distributions (for display purposes we actually show kernel density
estimates (KDE)). 
MDDR provides consistently better
fits than DDR. The second row of the figure explains why MDDR
outperforms DDR: it shows
 $\Gt_i$ and $f_k \sharp \Fh_{i,k}$ 
from the last posterior
sample in the Markov chain. With a linear regression function, DDR can
only produce fitted distributions with a univariate structure. In
contrast, MDDR can still provide accurate fitted distributions by
leveraging barycenters, although the push-forwards of the
predictors also have univariate structures.

Another inference target of interest is the barycenter weights, which
indicate the contribution of each predictor to the
response.
 The weights  allow us to quantify the strength of
association between the predictors and the response. As discussed in Remark~\ref{remark:invariance}, for linear regression functions and $p=2$, the weights $\pi$ are not identified, while the scale-normalized weights $\tilde\pi$ might be. In this simulation, the true $A_k$ are rotations, so $\|A_1\|_F=\|A_2\|_F=\|A_3\|_F$ and the true $\tilde\pi$ equals the true $\pi=(2/3,1/6,1/6)$. The posterior means of $\tilde\pi$ in the three replications are $(0.659,0.170,0.172)$, $(0.673,0.155,0.172)$, and $(0.670,0.168,0.162)$, with Euclidean distances to the truth $0.010$, $0.014$, and $0.006$. For the first replication, the 95\% posterior credible intervals are $(0.630,0.146,0.151)$–$(0.691,0.195,0.193)$, which contain the true weights. In the next section, we will discuss how to use the
posterior of $\tilde\pi$ to construct a cell-cell communication network
for single-cell data.

\textbf{Number of units and number of samples.} Theorem~\ref{theorem:approx_posterior_consistency} involves two sample sizes: the number of units $N$ and the number of i.i.d samples $n$ per distribution. To study them separately, we fix a test set of 30 units with $n=100$ samples, and either vary $N\in\{10,20,30,50\}$ with $n=100$, or vary $n\in\{10,20,50,100\}$ with $N=50$. The test RE of MDDR (Table~\ref{table:N_n_sweep} in Supplementary Material~\ref{sec:additional_results}) decreases from $0.077$ to $0.018$ as $N$ increases and decreases from $0.019$ to $0.017$ as $n$ increases. Hence, in this setting the number of units is the main driver of the accuracy, while the plug-in error from using $n$ samples is small already for $n=10$. This is consistent with  the discussion associated with Theorem~\ref{theorem:approx_posterior_consistency}.

\textbf{Sensitivity.} We vary one hyperparameter at a time from the setting above: $w\in\{5,20\}$, the scale of the Laplace prior in $\{0.5,2\}$, the Dirichlet concentration in $\{0.1,1\}$, the number of projections $L\in\{20,50\}$, and the number of SWB iterations $T=50$ (Table~\ref{table:sensitivity} in Supplementary Material~\ref{sec:additional_results}). The test RE varies between $0.010$ and $0.035$, and the distance between the posterior mean of $\tilde\pi$ and the truth is at most $0.010$ in all settings. Only $w$ has a visible effect on RE: a larger $w$ gives a more concentrated posterior and hence a smaller average residual. We also compare different choices of the reference model in RE (Table~\ref{table:reference_sensitivity}). All intercept-only references give nearly the same RE as the fitted reference (test RE $0.020$–$0.021$), so the conclusions do not depend on this choice.

\textbf{Misspecification.} We repeat the simulation with nonlinear true regression functions $f_k(x)=R\big(\vartheta_k+c\tanh(\|x\|_2/3)\big)x$ in place of $f_{A_k}$, where $R(\vartheta)$ is the rotation by angle $\vartheta$, $A_k=R(\vartheta_k)$, and $c=1$. Points far from the origin are rotated by up to one extra radian, so each predictor is mapped to a curve that no linear function can produce, and both MDDR and DDR are misspecified. MDDR still gives the best fit, with test RE $0.105$ (95\% HCI $0.058$–$0.179$), compared with $0.306$ ($0.288$–$0.330$) for DDR and $0.428$ ($0.399$–$0.455$) for aggregated DDR (Table~\ref{table:misspecified}). As expected, the RE of MDDR is larger than in the well-specified case.

\textbf{Convergence.} We run 4 independent chains of 500 iterations on the first simulated dataset and discard the first half of each chain. The rank-normalized split $\hat R$~\citep{vehtari2021rank} of the generalized log-likelihood is $1.037$ for MDDR and at most $1.020$ for DDR showing no evidence for poor mixing or a failure of convergence.

\begin{table}[t!]
    \centering
    \scalebox{0.9}{
    \begin{tabularx}{\linewidth}{l|p{4cm}|l|p{7cm}} 
        \toprule
        Predictors & Ligands & Response & Receptors \\
        \midrule
            Monocytes & CD86, ICOSL, IL-23, TNF-$\alpha$, IL-1$\beta$, IL-6, CD70 & \multirow{3}{*}{T Cells} & \multirow{3}{=}{\RaggedRight CD3D, CD3E, CD3G, TRBC1, TRBC2, CD28, ICOS, IL-2R, IFNGR, IL-21R, PD-1, CTLA-4, CXCR5, CCR7, CXCR3, CXCR4, IL-12R, IL-15R, IFN-$\gamma$R, Tim-3}\\
       NK Cells & IL-15, IFN-$\gamma$, CD40L, FasL  & & \\
           B Cells & CD40, CD86, ICOSL, IL-6, BAFF, APRIL & & \\
        \midrule 
        T Cells& CD40L, IFN-$\gamma$ & \multirow{3}{*}{B Cells} & \multirow{3}{=}{\RaggedRight IGHM, IGHD, IGHG, IGHA, CD40, ICOSL, IL-21R, IL-6R, BAFFR, CXCR5, CCR7, IL-10R, CXCR4, PD-1, IL-2R, IL-15R, IFN-$\gamma$R, CCR1, CXCR2, ICOS, Tim-3}\\
        Monocytes &  CD40L, IL-6, TNF-$\alpha$, Tim-3, BAFF, APRIL & & \\
        NK Cells & IL-15, IFN-$\gamma$, CXCL8, CD40L & & \\
    \midrule 
        B Cells&IL-6, BAFF, TNF-$\alpha$, CD40L & \multirow{3}{*}{NK Cells} & \multirow{3}{=}{\RaggedRight IL-6R, BAFF-R, TNFR, IL-10R, CCR1, CXCR3, CD40, ICOS, PD-1, Tim-3, NKG2D, IL-2R, IL-4R, IL-21R, IFN-$\gamma$R, CD28, NKp30, NKp46, DNAM-1, NKG2A, IL-15R}\\
        T Cells &  IFN-$\gamma$, CD86 & & \\
        Monocytes &  CD40L, TNF-$\alpha$, IL-15, IL-18, Tim-3 & &
        \\
        \midrule 
        NK Cells&IL-15, IFN-$\gamma$, FasL & \multirow{3}{*}{Monocytes} & \multirow{3}{=}{\RaggedRight IL-12R, IL-15R, IFN-$\gamma$R, CXCR3, Fas, CD40, PD-1, Tim-3, CD86, TNFR1, TNFR2, LAG-3, IL-6R, CXCR4 }\\
        B Cells &  CD40L, CD70, IL-6, TNF-$\alpha$, Tim-3, BAFF, APRIL & & \\
        T Cells &  CD40L, IFN-$\gamma$, CD70 & &
        \\
        \bottomrule
    \end{tabularx}}
    \caption{List of ligands' genes and receptors' genes for 4 regression problems  (1) Predictors:  Monocytes, NK Cells, B Cells, Response: T Cells; (2) Predictors: T Cells, Monocytes, NK Cells, Response: B Cells; (3) Predictors: B Cells, T Cells, Monocytes, Response: NK Cells; (4) Predictors: NK Cells, B Cells, T Cells, Response: Monocytes.}
    \label{tab:data_summary}
\end{table}
\section{Cell-Cell Communication}
\label{sec:cell_cell_communication}

Cell–cell communication is central to understanding complex biological
systems. Interactions between cells underlie a wide range of
physiological and pathological processes.  A natural approach to
quantifying such interactions is to model how
 the distribution of ligand gene expression 
of one cell type affects  the distribution of receptor
gene expressions of another cell type. 
We apply the proposed MDDR framework to infer communication between a set of cell-types as the senders and a
cell-type as the receiver using the population-scale single-cell
dataset OneK1K~\citep{yazar2022single}. 

\begin{table}[t!]
    \centering
    \scalebox{0.8}{
    \begin{tabular}{l|l|c|c|c|c}
        \toprule
        Response & Model& Train RE & Train RE 95\% HCI & Test RE & Test RE 95\% HCI \\
        \midrule
        \multirow{5}{*}{T Cells}
            & DDR (B Cells) & 0.2398 & 0.2381-0.2413 & 0.2476 & 0.2439-0.2508\\
            & DDR (Monocytes) & 0.3542 & 0.3519-0.3563 & 0.3620 & 0.3589-0.3646\\
            & DDR (NK Cells) & 0.5019 & 0.5004-0.5036 & 0.5166 & 0.5147-0.5187\\
            & Aggregated DDR & 0.2354 & 0.2335-0.2382 & 0.2397 & 0.2372-0.2423\\
            & MDDR & \textbf{0.1561} & \textbf{0.1544-0.1578} & \textbf{0.1668} & \textbf{0.1652-0.1687} \\
        \midrule
        \multirow{5}{*}{B Cells}
            & DDR (T Cells) & 0.6108 & 0.6082-0.6124 & 0.6040 & 0.6021-0.6067\\
            & DDR (Monocytes) & 0.4480 & 0.4454-0.4513 & 0.4596 & 0.4555-0.4634\\
            & DDR (NK Cells) & 0.6249 & 0.6232-0.6267 & 0.6646 & 0.6610-0.6681\\
            & Aggregated DDR & 0.4339 & 0.4312-0.4381 & 0.4686 & 0.4648-0.4733\\
            & MDDR & \textbf{0.2927} & \textbf{0.2911-0.2948} & \textbf{0.3033} & \textbf{0.3024-0.3045} \\
        \midrule
        \multirow{5}{*}{Monocytes}
            & DDR (T Cells) & 0.5665 & 0.5638-0.5689 & 0.5299 & 0.5275-0.5321\\
            & DDR (B Cells) & 0.3925 & 0.3896-0.3961 & 0.4161 & 0.4119-0.4199\\
            & DDR (NK Cells) & 0.5730 & 0.5699-0.5766 & 0.5569 & 0.5538-0.5599\\
            & Aggregated DDR & 0.3855 & 0.3837-0.3882 & 0.3738 & 0.3720-0.3768\\
            & MDDR & \textbf{0.2578} & \textbf{0.2563-0.2593} & \textbf{0.2514} & \textbf{0.2498-0.2532} \\
        \midrule
        \multirow{5}{*}{NK Cells}
            & DDR (T Cells) & 0.8121 & 0.8098-0.8144 & 0.7724 & 0.7701-0.7749\\
            & DDR (B Cells) & 0.4794 & 0.4774-0.4821 & 0.4962 & 0.4934-0.4988\\
            & DDR (Monocytes) & 0.4737 & 0.4697-0.4766 & 0.4650 & 0.4619-0.4682\\
            & Aggregated DDR & 0.4213 & 0.4185-0.4245 & 0.4216 & 0.4175-0.4244\\
            & MDDR & \textbf{0.3046} & \textbf{0.3015-0.3076} & \textbf{0.3090} & \textbf{0.3073-0.3111} \\
        \bottomrule
    \end{tabular}}
    \caption{Relative residual errors (RE) for the single-predictor DDR models, the aggregated DDR model, and the multi-predictor MDDR model across four cell types. Reported are training and test RE with their 95\% highest credible intervals (HCI) over 50 posterior samples. All models are compared with the same fitted intercept-only reference model.}
    \label{tab:simulation_comparison}
\end{table}

We focus on four major cell types: B cells, T cells, monocytes, and NK
(natural killer) cells, and set up four regression tasks: 
(1) Predictors: monocytes, NK cells, and B cells; Response: T cells.
(2) Predictors: T cells, monocytes, and NK cells; Response: B cells.
(3) Predictors: B cells, T cells, and monocytes; Response: NK cells.
(4) Predictors: NK cells, B cells, and T cells; Response:
monocytes. We summarize the list of ligands' genes and receptors'
genes for all 4  relations  in Table~\ref{tab:data_summary}. From the
table, we can see that the number of ligands (the dimension of the
predictors) is smaller than the number of receptors (the dimension of
the responses), which is similar to our simulation setup in
Section~\ref{sec:simulation}. We select donors who have at least 90 cells per cell type, which results in 75 donors, and randomly assign 70\% of them (52 donors) to the training set and 30\% (23 donors) to the testing set. The number of cells per donor varies across cell types, with medians of 647 T cells, 129 B cells, 124 monocytes, and 166 NK cells per donor (Table~\ref{tab:cell_counts} in Supplementary Material~\ref{sec:additional_results}).

We again compare the proposed MDDR with DDR and aggregated DDR (Section~\ref{sec:simulation}), using linear regression
functions in all models. Since DDR can only handle one predictor and we only observe cell-types separately without a joint distribution, we
set up three separate DDR regression problems
for each response variable,
each with a different predictor; aggregated DDR combines the three fitted DDR models through the SWB with uniform weights. In all models, we set $w
= 100$ for the generalized likelihood and run MALA,
described in Section~\ref{subsec:posterior_inference}, for 3000 iterations
with step sizes $\eta_1 = 0.0001$ (same for DDR's stepsize) and $\eta_2 = 0.0005$. DDR takes around 1 hour while MDDR takes about 2 days on a NVIDIA A100. We use the last 100 iterations, thinned by 2, as 50 posterior samples. SW distances
are approximated using 1000 Monte Carlo projections, both during inference and for computing RE. For computing SWB, we set $T = 100$,
$\eta = 0.1$, and $(\beta_1, \beta_2) = (0.9, 0.999)$.

\begin{figure}[!t]
\begin{center}
    \begin{tabular}{cccc}
  \widgraph{0.23\textwidth}{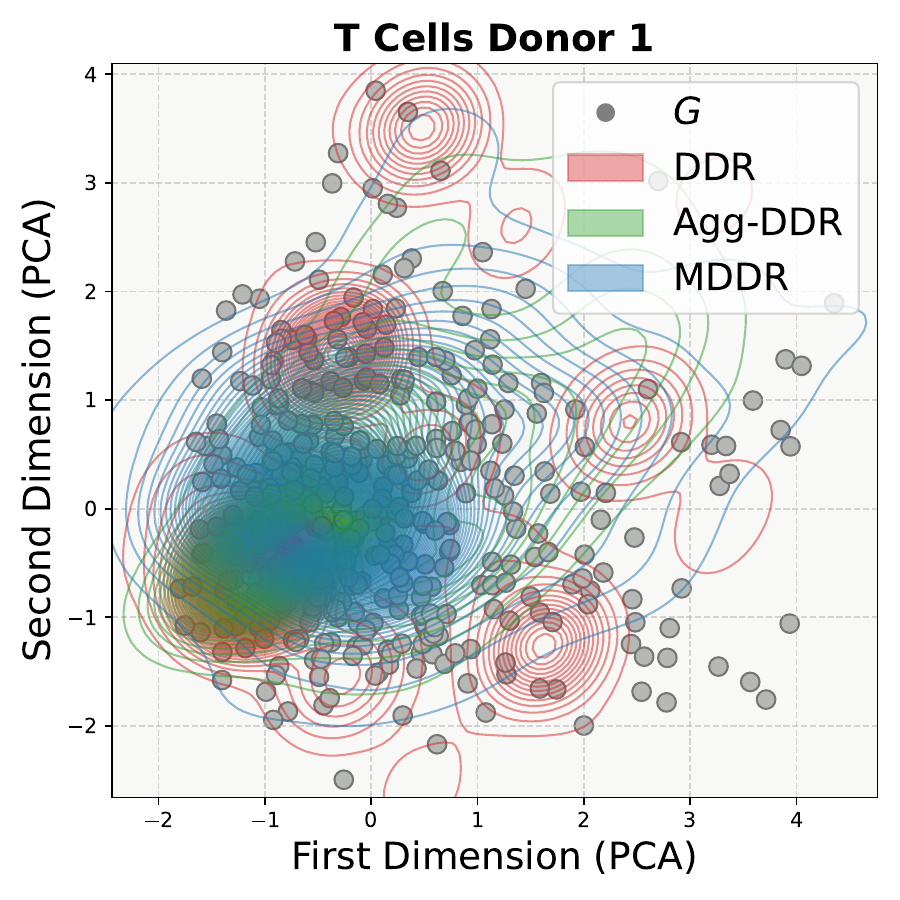} 
&
\widgraph{0.23\textwidth}{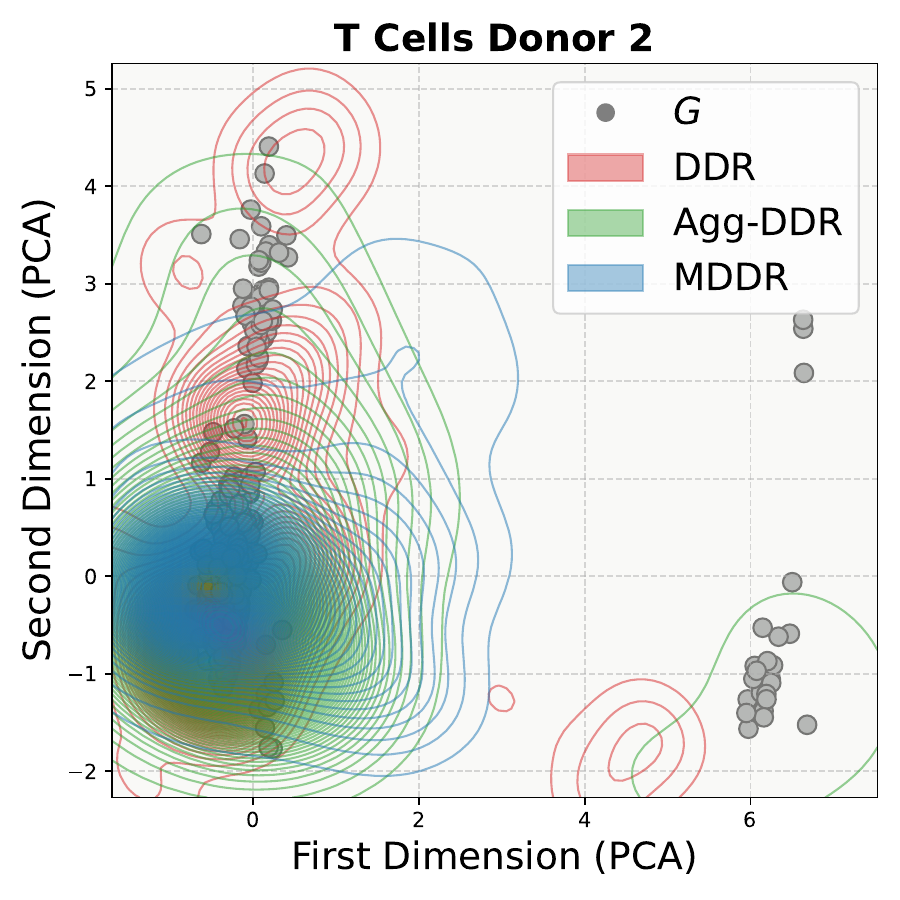} 
&
\widgraph{0.23\textwidth}{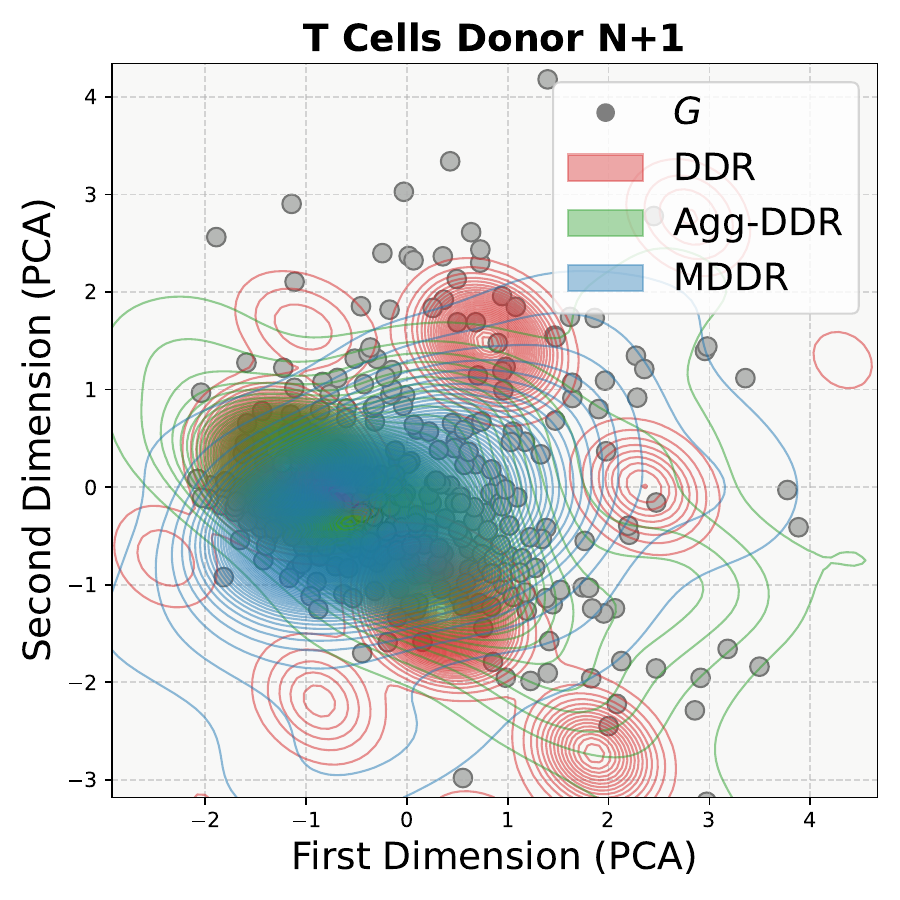} 
&
\widgraph{0.23\textwidth}{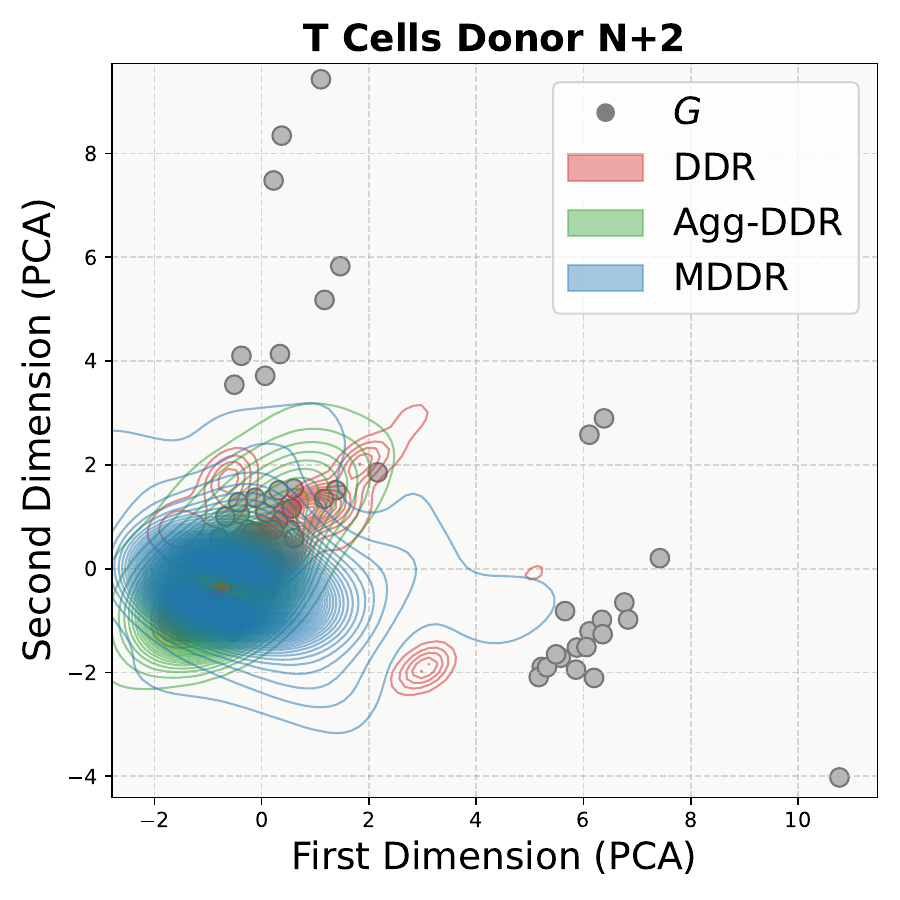} 
\\
  \widgraph{0.23\textwidth}{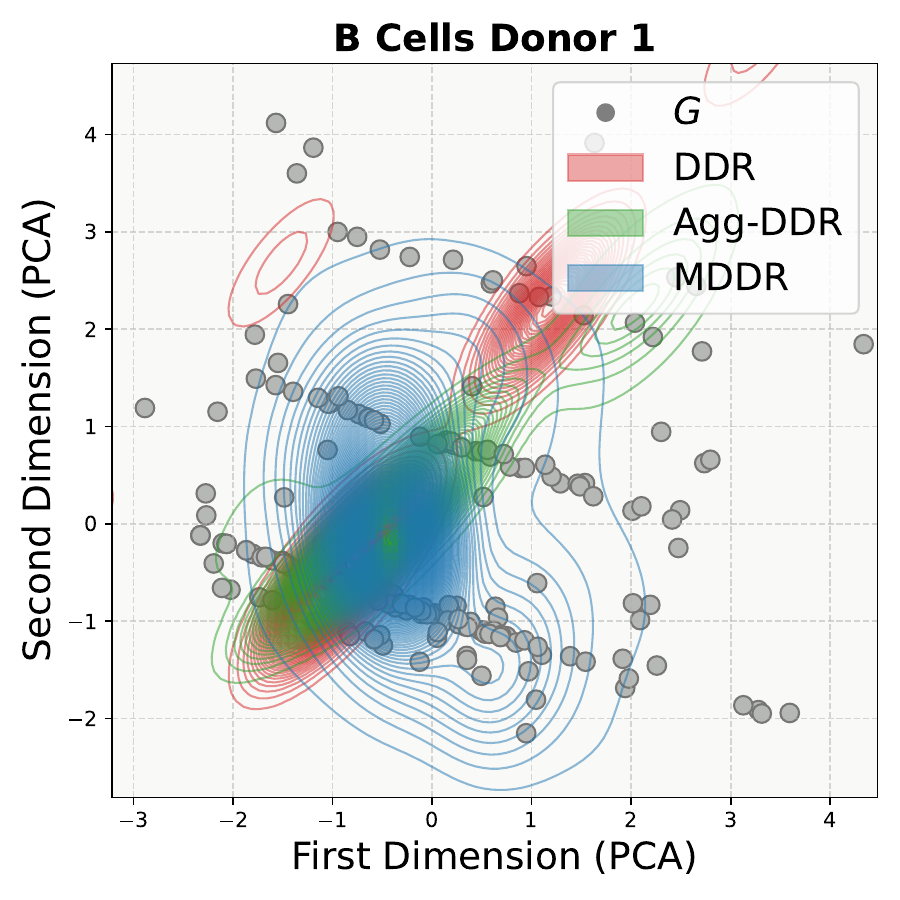} 
&
\widgraph{0.23\textwidth}{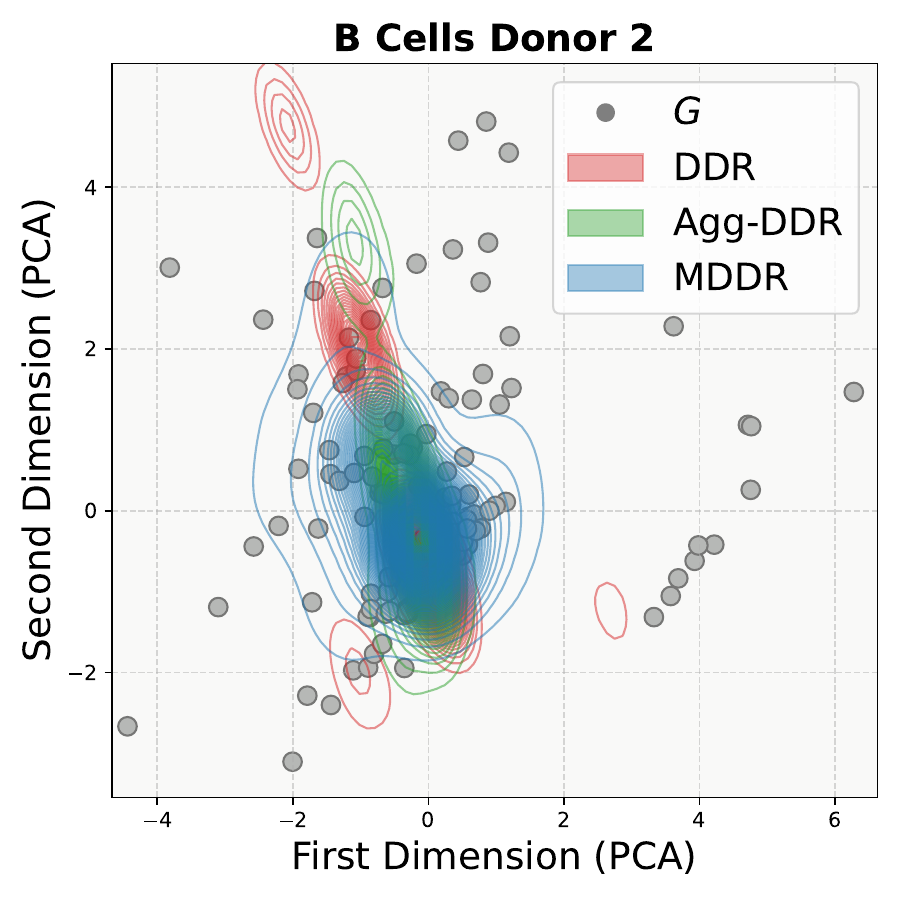} 
&
\widgraph{0.23\textwidth}{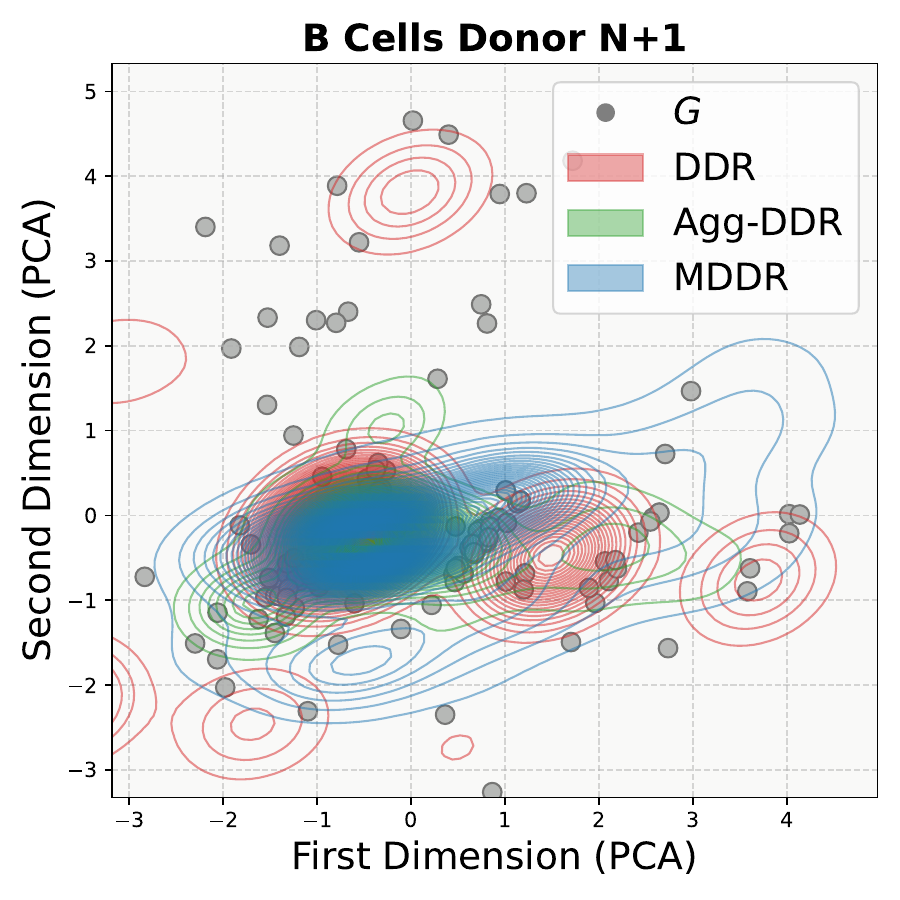} 
&
\widgraph{0.23\textwidth}{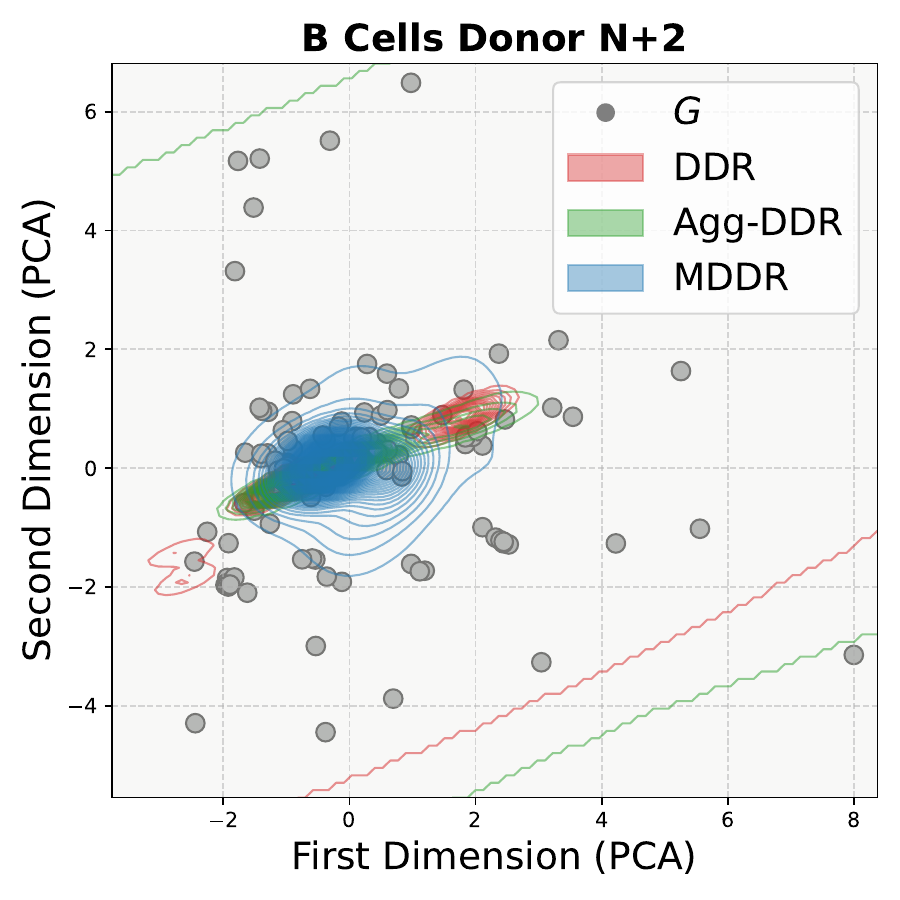}
  \end{tabular}
  \end{center}
  \vspace{-0.2 in}
  \caption{
    \footnotesize{The first two columns show two random in-sample donors and the last two columns show two random out-of-sample donors. The first row presents results for T Cells including the observed responses, and the posterior means of the KDEs of fitted distributions of DDR (with the best single predictor), aggregated DDR, and MDDR over the 50 posterior samples. The second row presents similar results for B Cells. We use PCA for visualization.
}
}
  \label{fig:Cell}
\end{figure}

Table~\ref{tab:simulation_comparison} summarizes
 residual errors under
MDDR compared with single-predictor DDR models and aggregated DDR across all four
response types. In every setting, MDDR achieves the lowest RE on both
the training and test sets. Its test RE is lower than that of the best competing model. Aggregated DDR improves over the best single-predictor DDR model for T cells, monocytes, and NK cells, but not for B cells, and it remains far from MDDR, since its regression functions are fitted separately and its weights are fixed to be uniform. The results echo similar experience with simple versus multiple
(normal linear) regression, and
underscore the advantage of jointly modeling multiple predictors, i.e., estimating the regression functions and the barycenter weights together.
MDDR more effectively captures cross–cell-type
relationships than any DDR model based on a single predictor,
particularly when the response has higher dimensionality than the
predictors.

We show the observed responses and the posterior means of
fitted DDR, aggregated DDR, and MDDR distributions
 (again, for display purposes showing KDE's)
for T cells and B cells in
Figure~\ref{fig:Cell}. MDDR provides noticeably
better fits than DDR and aggregated DDR for both training and testing donors. The fitted
distributions from MDDR capture the uncertainty in the observed
responses more accurately. While the fit could be further improved by
using more expressive regression functions beyond linear models, we
believe that linear regression functions, being both parsimonious and
computationally efficient, are sufficient in this setting. In
Figure~\ref{fig:Cell_appendix} of Supplementary
Material~\ref{sec:additional_results}, we present analogous results
for Monocytes and NK cells
 with the same conclusion:  
MDDR consistently outperforms DDR and aggregated DDR. Overall, these qualitative
observations align well with the quantitative  summaries  in
Table~\ref{tab:simulation_comparison}.

\begin{figure}[!t]
    \centering
    \begin{tabular}{cc}
         \includegraphics[width=0.45\linewidth]{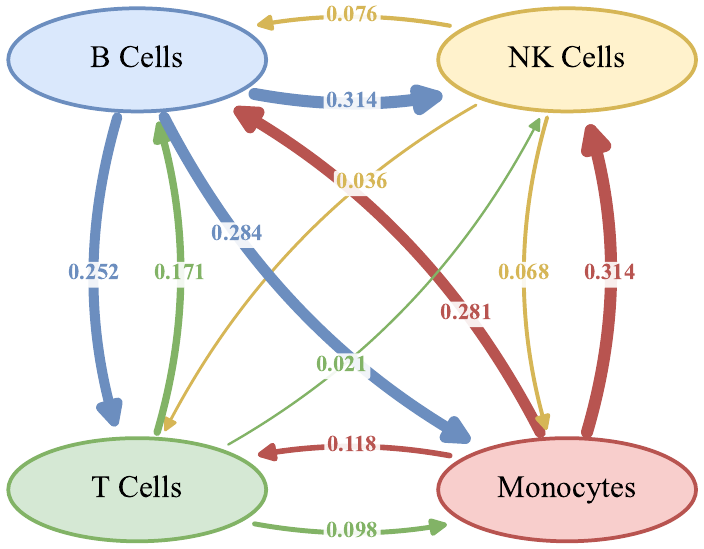}&
    \includegraphics[width=0.45\linewidth]{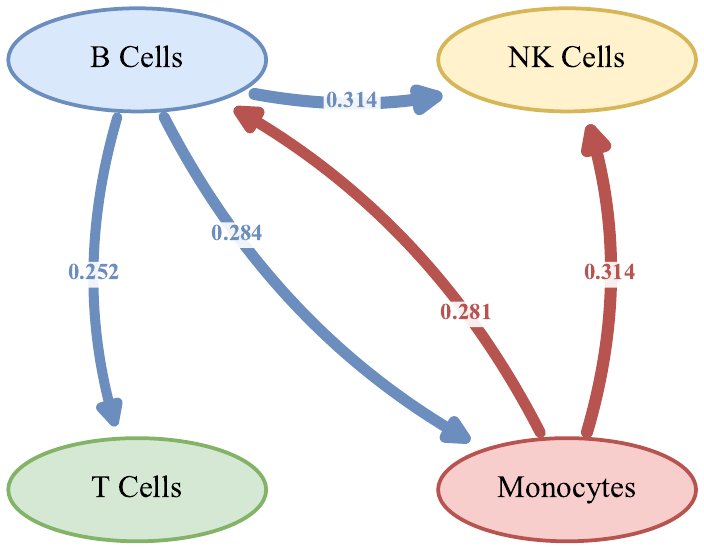} \\
    (a)&(b)
    \end{tabular}
    \caption{(a) Weighted graph of cell-cell communication under MDDR; the weight of the edge from a predictor cell type to a response cell type is the leave-one-out score $s_{k\to r}$ in \eqref{eq:edge_score}. (b) Sparse graph retaining, for each response, the edges whose score exceeds the average score of the edges into that response.}
    \label{fig:graph}
\end{figure}

Lastly, we highlight
 an important feature of inference under
the proposed MDDR model: it
naturally induces a graph structure between predictors and the
response (cell types in this case). From~\eqref{eq:MDDRmodel}
and~\eqref{eq:MDDRmodel2}, MDDR constructs the fitted distribution as
a barycenter of the push-forwards of the predictors, and each predictor contributes to the fit through both its regression function and its barycenter weight. Since the raw weights $\pi$ are not identified (Remark~\ref{remark:invariance}), we measure the contribution of a predictor $k$ to a response $r$ by how much the fit deteriorates when $k$ is removed. Let $\mathrm{PE}_{\mathrm{full}}$ be the in-sample residual error $PE(\SS)$ (Section~\ref{subsec:MDDR}) of MDDR with all three predictors, $\mathrm{PE}_{-k}$ that of MDDR fitted without predictor $k$, i.e., with the remaining two predictors, and $\mathrm{PE}_{\mathrm{ref}}$ that of the reference model. We define the score
\begin{align}
    s_{k\to r}=\frac{\log\big(\mathrm{PE}_{-k}/\mathrm{PE}_{\mathrm{full}}\big)}{\log\big(\mathrm{PE}_{\mathrm{ref}}/\mathrm{PE}_{\mathrm{full}}\big)}, \label{eq:edge_score}
\end{align}
which is the share of the improvement of MDDR over the reference model that is lost without predictor $k$: $s_{k\to r}=0$ if predictor $k$ adds nothing given the other predictors, and $s_{k\to r}=1$ if removing it brings the fit back to the reference model. The normalization by $\mathrm{PE}_{\mathrm{ref}}/\mathrm{PE}_{\mathrm{full}}$ makes the scores comparable across responses. The two-predictor models are fitted with the same sampler, and all residual errors are averaged over the last 500 MCMC iterations. Across the four regression problems, we obtain a
weighted graph whose vertices represent the cell types, in which each vertex
has three incoming edges from the other vertices (Figure~\ref{fig:graph}(a)). To obtain a sparse graph, for each response we keep the edges whose score exceeds the average score of its three incoming edges.  The resulting graph (Figure~\ref{fig:graph}(b)) is the same as the one obtained by cutting the sorted scores at their largest gap, and it does not change for any common cutoff between $0.18$ and $0.25$. From
this graph, we observe evidence
that B cells are associated with all other cell
types, and monocytes are associated with B cells and NK cells, while T cells and NK cells are not strong predictors of any other cell type. These conclusions agree with two other summaries. First, for every response, the ranking of the predictors by $s_{k\to r}$ is the same as the ranking of the single-predictor DDR models by RE (Table~\ref{tab:simulation_comparison}). Second, a graph built from the posterior means of the weights $\tilde\pi$ (Remark~\ref{remark:invariance}) and the cutoff $1/3$ contains the same edges, with one additional edge from monocytes to T cells ($\tilde\pi=0.341$) (Figure~\ref{fig:graph_pitilde} in Supplementary Material~\ref{sec:additional_results}). We emphasize that the edges represent predictive associations between the ligand expression of the predictor cell types and the receptor expression of the response cell type, given the other predictors. They suggest candidate communication patterns but do not establish causal communication, which would require experimental validation.

\section{Conclusion}
\label{sec:conclusion}

We introduced Bayesian multiple multivariate distribution–distribution regression as a novel framework for modeling complex regression relationships, where a multivariate random distribution serves as the response and multiple random distributions act as predictors.

The same structure, a distribution-valued response explained by several separately sampled distribution-valued predictors, arises beyond single-cell data: in environmental statistics (a site's temperature distribution versus those of neighboring sites), finance (one market's return distribution versus those of other markets), population health (regional outcome distributions versus subpopulation covariate distributions), and multi-source domain adaptation. In each case the predictors and the response are not jointly observed and have different sample sizes, which is the setting of MDDR.

Several limitations remain. First, the framework may be sensitive to model misspecification because it relies on parametric regression functions  $f_\phi$. Nevertheless, in both simulations and real-data applications, we found these mappings to be sufficiently flexible in practice. Second, the approximation of the sliced Wasserstein barycenter (SWB) used in our implementation may not be optimal in terms of computational efficiency or approximation accuracy. For the datasets considered in the paper, however, the current approximation proved adequate. 

Future work will aim to address these limitations. For example, the framework could be extended to a Bayesian nonparametric setting for multiple multivariate distribution–distribution regression. Additionally, alternative approximation strategies, such as more efficient iterative schemes with improved convergence or deterministic fast approximations of the SWB, could further enhance performance.

\section{Disclosure statement}\label{disclosure-statement}

There are no conflicts of interest.

\section{Data Availability Statement}\label{data-availability-statement}
We use either simulated datasets or a public dataset~\citep{yazar2022single}. 

\section{Declaration of Generative AI}

The authors used a generative AI tool (Claude, Anthropic) to assist with revising the manuscript, brainstorming approaches to theoretical proofs, and coding. All such use built on original drafts written by the authors, who reviewed and verified every AI-assisted output and take full responsibility for the content.

\phantomsection\label{supplementary-material}
\bigskip

{
\begin{center} {\large\bf SUPPLEMENTARY MATERIALS}
\end{center}
The Supplementary Materials include Appendices which contains technical proofs and additional discussion.
}

  \bibliography{bibliography.bib}

\clearpage

\appendix
\begin{center} \Large Supplementary Materials for ``Bayesian Multiple
  Multivariate Distribution-Distribution Regression"
\end{center}

\section{Computation of Sliced Wasserstein Barycenter}
\label{sec:computation_SWB}
We
discuss the computation of SWB,  which is required for the evaluation
of the likelihood in \eqref{eq:MDDRmodel2}.  
We recall the optimization problem which defines  SWB
with marginals $G_1,\ldots,G_K \in \PP_p(\mathbb{R}^d)$ and $\pi \in \Delta^K$: 
\begin{align}
    \min_{G\in \PP_p(\Re^d)} \sum_{k=1}^K \pi_k \SW_p^p(G,G_k)
\end{align}
Solving SWB is an optimization problem and is still a developing
research direction ~\citep{bonneel2015sliced,nguyen2025towards}.
SWB is a convex problem on the space of distributions i.e., the
mapping $G\to \sum_{k=1}^K \pi_k \SW_p^p(G,G_k)$ is convex
(Proposition~\ref{proposition:convexity_SWB} in Supplementary
Material~\ref{subsec:proof:proposition:convexity_SWB}). However, in
practice we often need to parameterize the barycenter for tractable
optimization, which may break convexity.

In our case, we use a plug-in estimator for the generalized likelihood,
solving SWB for
 the empirical distributions 
$G_{1}=\frac{1}{M_{G_1}} \sum_{i=1}^{M_{G_1}}
\delta_{x_{i1}},\ldots,G_k=\frac{1}{M_{G_K}} \sum_{i=1}^{M_{G_K}}
\delta_{x_{iK}}$.
It is then
natural to also restrict the barycenter to  an empirical
distribution,  i.e.,
$G=\frac{1}{M_{G}} \sum_{\ell=1}^{M_G}\delta_{z_\ell}$ (also known as the
free support barycenter~\citep{cuturi2014fast}) where
$z_1,\ldots,z_{M_G} \in \Re^d$. This restriction is justified as the number of atoms grows i.e., any (approximate) minimizer over uniform $M_G$-atom distributions converges to the SWB as $M_G\to\infty$ when the SWB is unique (Proposition~\ref{prop:free_support} in Supplementary Material~\ref{subsec:proof:prop:free_support}). We update the  support points  of the gradient: 
\begin{align}
\label{eq:gradient_z_1}
   \nabla_{z_\ell} \sum_{k=1}^K \pi_k \SW_p^p(G,G_k)&= \sum_{k=1}^K \pi_k \nabla_{z_\ell} \SW_p^p(G,G_k)\nonumber \\&= \sum_{k=1}^K \pi_k \nabla_{z_\ell} \EE_{\theta \sim \mathcal{U}(\mathbb{S}^{d-1})}[W_p^p(\theta \sharp G,\theta \sharp G_k)]  \nonumber\\
   &=\sum_{k=1}^K \pi_k  \EE_{\theta \sim \mathcal{U}(\mathbb{S}^{d-1})}[\nabla_{z_\ell}W_p^p(\theta \sharp G,\theta \sharp G_k)],
\end{align}
for any $z_\ell \in \{z_1,\ldots,z_{M_G}\}$. Let $\gamma^\star_{k,\theta}$ be the optimal transport plan between
$\theta \sharp G$ and $\theta \sharp G_k$,  implying 
\begin{align}
\label{eq:gradient_z_2}
    &\nabla_{z_\ell}W_p^p(\theta \sharp G,\theta \sharp G_k) = \nabla_{z_\ell}
\sum_{j=1}^{M_{G_k}} \gamma_{k,\theta,{ \ell} j}^\star |\theta^\top(
z_{{ \ell}} - x_{jk})|^p  \nonumber \\ 
    &= p\theta
\sum_{j=1}^{M_{G_{k}}}
\gamma_{k,\theta,\ell j}^\star 
\left|\theta^\top \big(z_{\ell}- x_{jk}\big)\right|^{p-2}
\left( \theta^\top \big(z_{\ell}- x_{jk}\big)\right),
\end{align}
when $\theta^\top \big(z_{\ell}- x_{jk})\neq 0$. A general update is of the form:
\begin{align}
    z_{\ell}^{(t)} = g\left(z_{\ell}^{(t-1)},\nabla_{z_\ell^{(t-1)}} \sum_{k=1}^K \pi_k \SW_p^p(G^{(t-1)},G_k) \right),
\end{align}
where $G^{(t)}=\frac{1}{M_{G}} \sum_{i=1}^{M_G}\delta_{z_\ell^{(t)}}$,  $g:\Re^d\times \Re^d \to \Re^d$ is the update rule, and $z_1^{(0)},\ldots,z_{M_G}^{(0)}$ are randomly initialized. We update $T>0$ iterations to obtain an approximation $G^{(T)}$ of SWB. Overall, the time complexity is $\mathcal{O}(T M_{max} \log M_{max})$ where $M_{max} = \max\{M_{G_1},\ldots,M_{G_K},M_{G}\}$.

For the upcoming discussion of efficient posterior inference we
will use the Jacobian $\frac{\partial z_\ell^{(T)}(\phi)}{\partial\phi}$.
Specifically, let 
$f_\phi = (f_{1,\phi_1},\ldots,f_{k,\phi_k})$ be parametric regression functions
with $\phi=(\phi_1,\ldots,\phi_K)$, $G_{1}=\frac{1}{M_{G_1}}
\sum_{i=1}^{M_{G_1}}
\delta_{f_{1,\phi_1}(x_{i1})},\ldots,G_k=\frac{1}{M_{G_K}}
\sum_{i=1}^{M_{G_k}} \delta_{f_{K,\phi_K}(x_{iK})}$. We consider the
Jacobian matrix $ \frac{\partial z_{\ell}^{(T)}(\phi)}{\partial \phi} $,
which is a function of $\phi$. Let
$h_\ell^{(t-1)}(\phi)=\nabla_{z_\ell^{(t-1)}} \sum_{k=1}^K \pi_k
\SW_p^p(G^{(t-1)},G_k)$ (see
\eqref{eq:gradient_z_1}-\eqref{eq:gradient_z_2}). By the update rule,
we have $z_{\ell}^{(T)}(\phi) = g\left(z_{\ell}^{(T-1)} (\phi),h_\ell^{(T-1)}
(\phi)\right)$. The Jacobian can be written as: 
\begin{align}
   \frac{\partial z_{\ell}^{(T)}(\phi)}{\partial \phi}=  \frac{\partial z_{\ell}^{(T)}(\phi)}{\partial z_{\ell}^{(T-1)}(\phi)} \frac{\partial z_{\ell}^{(T-1)}(\phi)}{\partial \phi }  + \frac{\partial z_{\ell}^{(T)}(\phi)}{\partial h_{\ell}^{(T-1)}(\phi)} \frac{\partial h_{\ell}^{(T-1)}(\phi)}{\partial \phi},
\end{align}
where $\frac{\partial z_{\ell}^{(T)}(\phi)}{\partial z_{\ell}^{(T-1)}(\phi)}$  and $\frac{\partial z_{\ell}^{(T)}(\phi)}{\partial h_{\ell}^{(T-1)}(\phi)}$ depends on the update rule $g$ (the exact form for a specific $g$ will be provided later), $\frac{\partial z_{\ell}^{(T-1)}(\phi)}{\partial \phi } $ can be computed recursively, and 
\begin{align}
   & \frac{\partial h_{\ell}^{(T-1)}(\phi)}{\partial \phi} 
   = \frac{\partial}{\partial \phi}\left(\sum_{k=1}^K \pi_k  \EE\left[p\theta
\sum_{j=1}^{M_{G_{k}}}
\gamma_{k,\theta,\ell j}^{\star(T-1)}
\left|\theta^\top \big(z_{\ell}^{(T-1)}(\phi)- f_{k,\phi_k}(x_{jk})\big)\right|^{p-2} \right.\right.\nonumber \\&\left.\left. \quad \quad 
\left( \theta^\top \big(z_{\ell}^{(T-1)}(\phi)- f_{k,\phi_k}(x_{jk})\big)\right) \right]\right) \nonumber \\
&=\sum_{k=1}^K \pi_k  \EE\left[p(p-1)\,\theta
\sum_{j=1}^{M_{G_k}}
\gamma_{k,\theta,\ell j}^{\star(T-1)}
\big|\theta^\top \big(z_\ell^{(T-1)}(\phi) - f_{k,\phi_k}(x_{jk})\big)\big|^{p-2} \nonumber \right.\\& \quad \quad  \left.
\left(\theta^\top \left(\frac{\partial z_\ell^{(T-1)}(\phi)}{\partial \phi}-\frac{\partial f_{k,\phi_k}(x_{jk})}{\partial \phi}\right)\right) \right],
\end{align}
where $\gamma_{k,\theta}^{\star(T-1)}$ is the optimal transport plan between
$\theta \sharp G^{(T-1)}$ and $\theta \sharp G_k$.
In the above, the expectation is with respect to $\theta \sim \mathcal{U}(\mathbb{S}^{d-1})$ and  $\frac{\partial f_{k,\phi_k}(x_{jk})}{\partial \phi}$ depends on the parameterization of the regression function $f_{k,\phi_k}$ which will be discussed later.

Similarly, we will also need
$\frac{\partial  z_{\ell}^{(T)}(\pi)}{ \partial \pi }$, which can be
derived as: 
\begin{align}
  \frac{\partial  z_{\ell}^{(T)}(\pi)}{ \partial \pi }
  =\frac{\partial z_{\ell}^{(T)}(\pi)}{\partial z_{\ell}^{(T-1)}(\pi)} \frac{\partial z_{\ell}^{(T-1)}(\pi)}{\partial \pi }  + \frac{\partial z_{\ell}^{(T)}(\pi)}{\partial h_{\ell}^{(T-1)}(\pi)} \frac{\partial h_{\ell}^{(T-1)}(\pi)}{\partial \pi},
\end{align}
where $\frac{\partial z_{\ell}^{(T)}(\pi)}{\partial z_{\ell}^{(T-1)}(\pi)}$  and $\frac{\partial z_{\ell}^{(T)}(\pi)}{\partial h_{\ell}^{(T-1)}(\pi)}$ depend on the update rule $g$, $\frac{\partial z_{\ell}^{(T-1)}(\pi)}{\partial \pi } $ can be computed recursively, and 
\begin{align}
   & \frac{\partial h_{\ell}^{(T-1)}(\pi)}{\partial \pi} 
   = \frac{\partial}{\partial \pi}\left(\sum_{k=1}^K \pi_k  \EE\left[p\theta
\sum_{j=1}^{M_{G_{k}}}
\gamma_{k,\theta,\ell j}^{\star,(T-1)}
\left|\theta^\top \big(z_{\ell}^{(T-1)}(\pi)- f_{k,\phi_k}(x_{jk})\big)\right|^{p-2} \right.\right. \nonumber\\&\quad \left.\left.
\left( \theta^\top \big(z_{\ell}^{(T-1)}(\pi)- f_{k,\phi_k}(x_{jk})\big)\right) \right]\right) = \left(\mathcal{G}_1,\ldots,\mathcal{G}_K\right),
\end{align}
with 
\begin{align}
    &\mathcal{G}_k\nonumber = \EE\left[p\theta
\sum_{j=1}^{M_{G_{k}}}
\gamma_{k,\theta,\ell j}^{\star,(T-1)}
\left|\theta^\top \big(z_{\ell}^{(T-1)}(\pi)- f_{k,\phi_k}(x_{jk})\big)\right|^{p-2}
\left( \theta^\top \big(z_{\ell}^{(T-1)}(\pi)- f_{k,\phi_k}(x_{jk})\big)\right) \right] \nonumber\\
&\quad + \sum_{k'=1}^K \pi_{k'} \,
\EE \Bigg[ p(p-1) \sum_{j=1}^{M_{G_{k'}}}
\gamma_{k',\theta,\ell j}^{\star,(T-1)} \,
\big|\theta^\top ( z_\ell^{(T-1)}(\pi) - f_{k',\phi_{k'}}(x_{jk'}) ) \big|^{p-2} \, 
\theta \theta^\top \Bigg] 
\frac{\partial z_\ell^{(T-1)}(\pi)}{\partial \pi_k}.
\end{align}

In practice, we use Monte Carlo samples
$\theta_1,\ldots,\theta_L\overset{i.i.d}{\sim
}\mathcal{U}(\mathbb{S}^{d-1})$ to approximate
 the expectation in  the gradients.
For the update rule $g$, we use
Adam~\citep{kingma2014adam}, which can be described as: 
\begin{align}
    &z_{\ell}^{(t)} = g\left(z_{\ell}^{(t-1)},h_\ell^{(t-1)}\right) = z_{\ell}^{(t-1)}- \eta \frac{\hat{m}_\ell^{(t)}}{\sqrt{\hat{v}^{(t)}_\ell} + \epsilon}, \nonumber \\
    &\hat{m}^{(t)}_\ell = \frac{m^{(t)}_\ell}{1 - \beta_1^t}, \quad \hat{v}^{(t)}_\ell = \frac{v^{(t)}_\ell}{1 - \beta_2^t}, \nonumber\\
    &m^{(t)}_\ell= \beta_1 m^{(t-1)}_\ell + (1 - \beta_1) h^{(t-1)}_\ell, \quad v^{(t)}_\ell = \beta_2 v^{(t-1)}_\ell + (1 - \beta_2) (h^{(t-1)}_\ell)^2,
\end{align}
where $\eta >0$ is the step size, $\beta_1\in[0,1)$ and $\beta_2\in[0,1)$ are momentum parameters, $\epsilon>0$ for avoiding $0$ in the denominator.  For this update rule, we have:
\begin{align}
    \frac{\partial z_\ell^{(t)}}{ \partial z_\ell^{(t-1)}}=I,\quad    \frac{\partial z_\ell^{(t)}}{ \partial h_\ell^{(t-1)}} = -\eta\, \diag\left(\frac{\partial\hat{m}_\ell^{(t)} }{\partial h_\ell^{(t-1)}} \frac{1}{\sqrt{\hat{v}^{(t)}_\ell} + \epsilon} + \hat{m}_\ell^{(t)} \frac{\partial \left(\sqrt{\hat{v}^{(t)}_\ell} + \epsilon \right)^{-1}}{\partial h_\ell^{(t-1)}}\right),
\end{align}
where 
\begin{align}
    &\frac{\partial\hat{m}_\ell^{(t)} }{\partial h_\ell^{(t-1)}} = \frac{1-\beta_1}{1-\beta_1^t}I, \quad  \frac{\partial \left(\sqrt{\hat{v}^{(t)}_\ell} + \epsilon)^{-1} \right)}{\partial h_\ell^{(t-1)}} = -\frac{1}{\left(\sqrt{\hat{v}^{(t)}_\ell} + \epsilon\right)^{2}} \frac{1}{2\sqrt{\hat{v}^{(t)}_\ell}} \frac{\partial \hat{v}^{(t)}_\ell}{\partial h_\ell^{(t-1)}}, \\
    &\frac{\partial \hat{v}^{(t)}_\ell}{\partial h_\ell^{(t-1)}}=\frac{2(1-\beta_2)}{1-\beta_2^t} \diag( h_\ell^{(t-1)}).
 \end{align}
The moment states $m^{(t-1)}_\ell$ and $v^{(t-1)}_\ell$ also depend on $\phi$ (and $\pi$). In practice, we compute the exact Jacobian by automatic differentiation through the $T$ unrolled updates. We note that there is room to improve the update rule, for example by using a Riemannian Silver step size~\citep{park2025acceleration}. However, we keep the update rule as simple as possible while retaining good practical performance. 

\section{Proofs}

\subsection{Preliminaries}
\label{subsec:preliminaries}
We collect some facts which will be used in the later proofs. For $G\in\PP(B_R)$ and $\theta\in\mathbb{S}^{d-1}$,  we have $\theta\sharp G$ is supported in $[-R,R]$ since $|\theta^\top x|\leq\|x\|_2\leq R$. Hence, for all $G,H\in\PP(B_R)$,
\begin{align}
    \SW_p(G,H)\leq D=2R. \label{eq:D_bound}
\end{align}
For all $G,G',H,H'\in\PP(B_R)$, using the above bound   with the inequality $|a^p-b^p| \leq p\max\{a,b\}^{p-1}|a-b|$ for $a,b\geq0$ and the triangle inequality of $\SW_p$, we have
\begin{align}
    \big|\SW_p^p(G,H)-\SW_p^p(G',H')\big|\leq pD^{p-1}\big(\SW_p(G,G')+\SW_p(H,H')\big). \label{eq:Lip}
\end{align}
Moreover, if $G_m\to G$ weakly in $\PP(B_R)$, then $\theta\sharp G_m\to\theta\sharp G$ weakly for every $\theta$, hence $W_p(\theta\sharp G_m,\theta\sharp G)\to0$ since all supports lie in $[-R,R]$~\citep[Theorem~6.9]{villani2009optimal}, and the dominated convergence theorem gives
\begin{align}
    \SW_p(G_m,G)\to0. \label{eq:weak_SW}
\end{align}
Finally, $\PP(B_R)$ is compact with respect to weak convergence (Prokhorov's theorem), and for $F_1,\ldots,F_K\in\PP(B_R)$ and $\pi\in\Delta^K$, the objective $J(G)=\sum_{k=1}^K\pi_k\SW_p^p(G,F_k)$ is continuous on $\PP(B_R)$ by \eqref{eq:Lip} and \eqref{eq:weak_SW}. Hence, a minimizer of $J$ over $\PP(B_R)$ exists.

\subsection{Proof of Theorem~\ref{theorem:sample_complexity}}
\label{subsec:proof:theorem:sample_complexity}
Using the mean value theorem for the exponential function, we have
\begin{align}
    e^x-e^y = e^c (x-y),
\end{align}
for any $x,y \in \mathbb{R}$ and some $c \in [x,y]$. Therefore, we have
\begin{align}
    |e^x-e^y| \leq \max\{e^x,e^y\} |x-y|.
\end{align}
Apply the above inequality to our case, we obtain:
 \begin{align}
        &\EE\left[\left|\exp \{-w  \SW_p^p(\bar{F},G)\} - \exp \{-w  \SW_p^p(\hat{\bar{F}},\hat{G})\} \right|\right]   \nonumber\\
        &\leq \EE\left[\max\{\exp\{-w  \SW_p^p(\bar{F},G)\},\exp\{-w  \SW_p^p(\hat{\bar{F}},\hat{G})\}\} \left| w(  \SW_p^p(\bar{F},G)-  \SW_p^p(\hat{\bar{F}},\hat{G}))\right| \right] \nonumber\\
        &\leq w\EE\left[\left|  \SW_p^p(\bar{F},G)-  \SW_p^p(\hat{\bar{F}},\hat{G})\right| \right], \label{eq:sample_1}
    \end{align}
where the last inequality is due to the fact that $\SW_p^p(\bar{F},G)\geq 0$ and $\SW_p^p(\hat{\bar{F}},\hat{G})\geq 0$ which implies $\exp(-w  \SW_p^p(\bar{F},G))\leq 1$ and $\exp(-w  \SW_p^p(\hat{\bar{F}},\hat{G}))\leq 1$.  Using Jensen's inequality, we have:
\begin{align}
    w\EE\left[\left|  \SW_p^p(\bar{F},G)-  \SW_p^p(\hat{\bar{F}},\hat{G})\right| \right]
    \leq w\EE\left[ \mathbb{E}_\theta \left[ \left| W_p^p(\theta \sharp \bar{F},\theta \sharp G)-  W_p^p(\theta \sharp \hat{\bar{F}},\theta \sharp \hat{G})\right| \right]\right].
\end{align}
All distributions involved, including $\bar F$ and $\hat{\bar F}$, belong to $\PP(B_R)$, so all their projections are supported in $[-R,R]$, an interval of length $D=2R$. For one-dimensional $\mu,\nu,\mu',\nu'$ supported in $[-R,R]$, the closed form of the one-dimensional Wasserstein distance and $|a^p-b^p|\leq pD^{p-1}|a-b|$ for $a,b\in[0,D]$ give
\begin{align}
    \big|W_p^p(\mu,\nu)-W_p^p(\mu',\nu')\big|
    &\leq pD^{p-1}\int_0^1\big(|Q_\mu-Q_{\mu'}|+|Q_\nu-Q_{\nu'}|\big)\mathrm{d}t\\
    &= pD^{p-1}\big(W_1(\mu,\mu')+W_1(\nu,\nu')\big).
\end{align}
Therefore, we further have:
\begin{align}
   \EE\left[ \mathbb{E}_\theta \left[ \left| W_p^p(\theta \sharp \bar{F},\theta \sharp G)-  W_p^p(\theta \sharp \hat{\bar{F}},\theta \sharp \hat{G})\right| \right]\right]
   \leq pD^{p-1}\EE\left[ \SW_1(\bar{F},\hat{\bar{F}})+  \SW_1( G, \hat{G})\right]. \label{eq:sample_2}
\end{align}
 The sample complexity of $\SW_1$ has been investigated in~\citet{nadjahi2020statistical,nguyen2021distributional,manole2022minimax,nietert2022statistical,boedihardjo2025sharp}. We provide a proof based on compact support:
\begin{align}
    \EE\left[\SW_1( G, \hat{G})\right]  &= \EE\left[\EE_\theta \left[W_1( \theta \sharp G, \theta \sharp \hat{G})\right]\right] \nonumber  \\
    &= \EE\left[\EE_\theta \left[\int_{-R}^{R}| F_{\theta \sharp G}(x) -  F_{\theta \sharp \hat{G}}(x) |\mathrm{dx}\right] \right] \nonumber \\
    &\leq 2R \EE\left[\EE_\theta \left[\sup_x| F_{\theta \sharp G}(x) -  F_{\theta \sharp \hat{G}}(x) |\right] \right] \nonumber \\
    &=  2R \EE_\theta \left[\EE\left[\sup_x| F_{\theta \sharp G}(x) -  F_{\theta \sharp \hat{G}}(x) |\right] \right],\label{eq:sample_4}
\end{align}
by Fubini's theorem. From Dvoretzky--Kiefer--Wolfowitz's inequality~\citep{dvoretzky1956asymptotic,massart1990tight}, we know that
\begin{align}
    \mathbb{P}(\sup_x| F_{\theta \sharp G}(x) -  F_{\theta \sharp \hat{G}}(x) |>\epsilon) \leq 2\exp(-2n\epsilon^2),
\end{align}
 for any $\epsilon>0$ and $\theta \in \mathbb{S}^{d-1}$. Therefore,
 \begin{align}
     \EE\left[\sup_x| F_{\theta \sharp G}(x) -  F_{\theta \sharp \hat{G}}(x) |\right]& = \int_0^\infty  \mathbb{P}(\sup_x| F_{\theta \sharp G}(x) -  F_{\theta \sharp \hat{G}}(x) |>\epsilon) \mathrm{ d}\epsilon \nonumber \\
     &\leq \int_0^\infty 2\exp(-2n\epsilon^2) \mathrm{d}\epsilon 
     = \sqrt{\frac{\pi}{2n}},
 \end{align}
 due to the Gaussian integral. By  Fubini's theorem,
 \begin{align}
     \EE_\theta \left[\EE\left[\sup_x| F_{\theta \sharp G}(x) -  F_{\theta \sharp \hat{G}}(x) |\right] \right] \leq \sqrt{\frac{\pi}{2n}}.\label{eq:sample_5}
 \end{align}
Hence, $\EE[\SW_1(F,\hat F)] \leq 2R\sqrt{\pi/(2n)}$ for any $F\in\PP(B_R)$. Since $W_p^p \leq D^{p-1}W_1$ for distributions supported in $[-R,R]$, we have
$\SW_p^p(F,\hat F) \leq D^{p-1}\SW_1(F,\hat F)$, and Jensen's inequality gives
\begin{align}
\label{eq:SWp_rate}
    \EE\big[\SW_p(F,\hat F)\big] \leq \big(\EE[\SW_p^p(F,\hat F)]\big)^{1/p}
    \leq 2R\Big(\frac{\pi}{2n}\Big)^{1/(2p)} .
\end{align}
By \eqref{eq:stability}, Jensen's inequality for $t\mapsto t^{1/p}$, and \eqref{eq:SWp_rate},
\begin{align}
    &\EE\big[\SW_1(\bar F,\hat{\bar F})\big] \leq \EE\big[\SW_p(\bar F,\hat{\bar F})\big]
    \leq \left(\kappa\sum_{k=1}^K \pi_k \EE\big[\SW_p(F_k,\hat F_k)\big]\right)^{1/p} \nonumber\\
    &\leq (2R\kappa)^{1/p}\Big(\frac{\pi}{2n}\Big)^{1/(2p^2)} . 
\end{align}
Substituting this together with $\EE[\SW_1(G,\hat{G})] \leq 2R\sqrt{\pi/(2n)}$ into
\eqref{eq:sample_1} and \eqref{eq:sample_2}, and using $(\pi/2)^{s}\leq(\pi/2)^{1/2}$ for $s\leq 1/2$ and $n^{-1/2}\leq n^{-1/(2p^2)}$, we obtain
\begin{align}
    & \EE\left[\left|\exp \{-w  \SW_p^p(\bar{F},G)\} - \exp \{-w  \SW_p^p(\hat{\bar{F}},\hat{G})\} \right|\right]  \nonumber\\
    &\leq wp(2R)^{p-1}\big((2R\kappa)^{1/p}+2R\big)\sqrt{\frac{\pi}{2}} \frac{1}{n^{1/(2p^2)}}, 
\end{align}
which completes the proof.

\subsection{Proof of Theorem~\ref{theorem:posterior_consistency}}
\label{subsec:proof:theorem:posterior_consistency}

We first establish two lemmas. The first one states that the SWB depends continuously on the marginals and the weights whenever the limiting SWB is unique.
\begin{lemma}[Continuity of the SWB]
\label{lemma:continuity_SWB}
Let $F_1,\ldots,F_K \in \PP(B_R)$ and $\pi\in\Delta^K$ be such that $\bar{F}=\SWB_p(F_1,\ldots,F_K,\pi)$ is unique, let $F_k^{(m)}\in\PP(B_R)$ with $\SW_p(F_k^{(m)},F_k)\to 0$ as $m\to\infty$ for all $k=1,\ldots,K$, let $\pi^{(m)}\in\Delta^K$ with $\pi^{(m)}\to\pi$, and let $\bar{F}^{(m)}$ be any minimizer defining $\SWB_p(F_1^{(m)},\ldots,F_K^{(m)},\pi^{(m)})$. We have $\SW_p(\bar{F}^{(m)},\bar{F})\to 0$ as $m\to\infty$.
\end{lemma}
\begin{proof}
We write $J(G)=\sum_{k=1}^K\pi_k\SW_p^p(G,F_k)$ and $J^{(m)}(G)=\sum_{k=1}^K\pi_k^{(m)}\SW_p^p(G,F_k^{(m)})$ for $G\in\PP(B_R)$. For any $G\in\PP(B_R)$, we have
\begin{align}
    \big|J^{(m)}(G)-J(G)\big|
    &\leq\sum_{k=1}^K\pi_k^{(m)}\big|\SW_p^p(G,F_k^{(m)})-\SW_p^p(G,F_k)\big|+\sum_{k=1}^K\big|\pi_k^{(m)}-\pi_k\big|\SW_p^p(G,F_k) \nonumber\\
    &\leq pD^{p-1}\max_{k\in\{1,\ldots,K\}}\SW_p(F_k^{(m)},F_k)+D^p\|\pi^{(m)}-\pi\|_1=:\epsilon_m, \label{eq:uniform_J}
\end{align}
where the first inequality is due to the triangle inequality, and the second one is due to \eqref{eq:Lip} and \eqref{eq:D_bound}. By assumption, $\epsilon_m\to0$. Let $(m_j)_{j\geq1}$ be any subsequence. Since $\PP(B_R)$ is compact, there are a further subsequence, which we still denote by $(m_j)_{j\geq 1}$, and $H\in\PP(B_R)$ such that $\bar F^{(m_j)}\to H$ weakly, and hence $\SW_p(\bar F^{(m_j)},H)\to0$ by \eqref{eq:weak_SW}. We have
\begin{align}
    J(H)
    &\leq J(\bar F^{(m_j)})+pD^{p-1}\SW_p(\bar F^{(m_j)},H)
    \leq J^{(m_j)}(\bar F^{(m_j)})+\epsilon_{m_j}+pD^{p-1}\SW_p(\bar F^{(m_j)},H) \nonumber\\
    &\leq J^{(m_j)}(\bar F)+\epsilon_{m_j}+pD^{p-1}\SW_p(\bar F^{(m_j)},H) \nonumber\\
    &\leq J(\bar F)+2\epsilon_{m_j}+pD^{p-1}\SW_p(\bar F^{(m_j)},H),
\end{align}
where the first inequality is due to \eqref{eq:Lip}, the second and the last ones are due to \eqref{eq:uniform_J}, and the third one is due to the optimality of $\bar F^{(m_j)}$ for $J^{(m_j)}$. Letting $j\to\infty$, we obtain $J(H)\leq J(\bar F)$, i.e., $H$ is a minimizer of $J$, and $H=\bar F$ by uniqueness. Hence, every subsequence of $\big(\SW_p(\bar F^{(m)},\bar F)\big)_{m\geq1}$ has a further subsequence converging to $0$, so the whole sequence converges to $0$. We conclude the proof here.
\end{proof}

The second lemma establishes a modulus of continuity for $ \big| \SW_p^p(\Gt(\phi,\pi), G) - \SW_p^p(\Gt(\phi',\pi'), G) \big|$, which is later used for proving uniform law of large numbers for the empirical risk.
\begin{lemma}
\label{lemma:lipchitz}
Under Assumptions~\ref{assumption:compact} and~\ref{assumption:continuity}, there is a modulus of continuity $\gamma:\Re_+\to \Re_+$ with $\lim_{t\to 0}\gamma(t)=0$ such that, for all $(F_1,\ldots,F_K,G) \in \mathrm{supp}(P)$ and $(\phi,\pi),(\phi',\pi')\in\Phi\times\Delta^K$,
\begin{align}
   \big| \SW_p^p(\Gt(\phi,\pi), G) - \SW_p^p(\Gt(\phi',\pi'), G) \big|
    \leq \gamma(\|\phi - \phi'\|_p^p+\|\pi-\pi'\|_p^p).
\end{align}
Moreover, for fixed $(\phi,\pi)$, the map $(F_1,\ldots,F_K,G)\mapsto\SW_p^p(\Gt(\phi,\pi), G)$ is continuous on $\mathrm{supp}(P)$.
\end{lemma}

\begin{proof}
For $z=(F_1,\ldots,F_K,G)\in\mathrm{supp}(P)$ and $\vartheta=(\phi,\pi)\in\Phi\times\Delta^K$, let $L(\vartheta,z):=\SW_p^p(\Gt(\phi,\pi),G)$, which is well-defined by the uniqueness of the SWB in Assumption~\ref{assumption:continuity}. By Assumption~\ref{assumption:continuity}, all push-forwards $f_{k,\phi_k}\sharp F_k$ and $G$ belong to $\PP(B_R)$. For two regression functions, we have
\begin{align}
    \SW_p^p(f_{k,\phi_k}\sharp F_k, f_{k,\phi_k'}\sharp F_k)
    &=\EE_{\theta \sim \mathcal{U}(\mathbb{S}^{d-1})}\left[\inf_{\gamma \in \Pi(F_k,F_k)} \EE_{(X,Y) \sim \gamma}[|\theta^\top f_{k,\phi_k}(X)-\theta^\top f_{k,\phi_k'}(Y)|^p]\right] \nonumber \\
    &\leq \EE_{\theta \sim \mathcal{U}(\mathbb{S}^{d-1})}\left[ \EE_{X\sim F_k }[|\theta^\top f_{k,\phi_k}(X)-\theta^\top f_{k,\phi_k'}(X)|^p]\right] \nonumber  \\
    &\leq \sup_{x\in\mathcal{X}_k}\| f_{k,\phi_k}(x)- f_{k,\phi_k'}(x)\|_2^p, \label{eq:pushforward}
\end{align}
where the first inequality is due to the choice of the identity coupling, and the second
one is due to the Cauchy--Schwarz inequality and $\|\theta\|_2=1$.

We first show that $L$ is continuous on $(\Phi\times\Delta^K)\times\mathrm{supp}(P)$. Let $(\phi^{(m)},\pi^{(m)},z^{(m)})\to(\phi,\pi,z)$, where $z^{(m)}=(F_1^{(m)},\ldots,F_K^{(m)},G^{(m)})$ converges weakly to $z=(F_1,\ldots,F_K,G)$. For every $k$, by the triangle inequality and \eqref{eq:pushforward},
\begin{align}
    &\SW_p\big(f_{k,\phi_k^{(m)}}\sharp F_k^{(m)},f_{k,\phi_k}\sharp F_k\big)
\nonumber\\&\leq\sup_{x\in\mathcal{X}_k}\big\|f_{k,\phi_k^{(m)}}(x)-f_{k,\phi_k}(x)\big\|_2+\SW_p\big(f_{k,\phi_k}\sharp F_k^{(m)},f_{k,\phi_k}\sharp F_k\big)\to0,
\end{align}
where the first term converges to $0$ since $(\phi_k,x)\mapsto f_{k,\phi_k}(x)$ is uniformly continuous on the compact set $\Phi_k\times\mathcal{X}_k$ (Heine--Cantor theorem), and the second one converges to $0$ by the continuous mapping theorem and \eqref{eq:weak_SW}. By the uniqueness of the SWB in Assumption~\ref{assumption:continuity}, Lemma~\ref{lemma:continuity_SWB} gives $\SW_p(\Gt^{(m)},\Gt)\to0$, where $\Gt^{(m)}$ and $\Gt$ denote the SWB at $(\phi^{(m)},\pi^{(m)},z^{(m)})$ and $(\phi,\pi,z)$. Together with $\SW_p(G^{(m)},G)\to0$ by \eqref{eq:weak_SW}, \eqref{eq:Lip} gives $L(\phi^{(m)},\pi^{(m)},z^{(m)})\to L(\phi,\pi,z)$. Hence, $L$ is continuous, which proves the second claim.

Since $\Phi\times\Delta^K$ is compact by Assumption~\ref{assumption:compact}, and $\mathrm{supp}(P)$ is a closed subset of the compact set $\PP(\mathcal{X}_1)\times\ldots\times\PP(\mathcal{X}_K)\times\PP(B_R)$ (Prokhorov's theorem) by Assumption~\ref{assumption:continuity}, $L$ is uniformly continuous on $(\Phi\times\Delta^K)\times\mathrm{supp}(P)$ (Heine--Cantor theorem). Therefore,
\begin{align}
    \gamma(t):=\sup\Big\{\big|L(\vartheta,z)-L(\vartheta',z)\big|:\ z\in\mathrm{supp}(P),\ \|\phi-\phi'\|_p^p+\|\pi-\pi'\|_p^p\leq t\Big\}
\end{align}
satisfies $\lim_{t\to0}\gamma(t)=0$, which proves the first claim. We conclude the proof here.
\end{proof}

\begin{lemma}[Uniform Law of Large Numbers]
\label{lemma:uniform_law}
Under assumptions~\ref{assumption:compact} and~\ref{assumption:continuity},
\begin{align}
    \sup_{(\phi,\pi) \in \Phi \times \Delta^K} |R_N(\phi,\pi) - R(\phi,\pi)| \overset{a.s}{\to} 0 \quad \text{as } N \to \infty.
\end{align}
\end{lemma}
\begin{proof}
Since $\Phi$ is compact by Assumption~\ref{assumption:compact} and the simplex $\Delta^K$ is already compact, there exists a $\delta$-net $\{ (\phi^{(1)},\pi^{(1)}), \ldots, (\phi^{(H)},\pi^{(H)})\} \subset \Phi \times \Delta^K$ such that for any $(\phi,\pi) \in \Phi \times \Delta^K$, there exists $(\phi^{(h)},\pi^{(h)}) \in \{ (\phi^{(1)},\pi^{(1)}), \ldots, (\phi^{(H)},\pi^{(H)})\} $ with
\begin{align}
\|\phi - \phi^{(h)}\|_p^p + \|\pi -  \pi^{(h)}\|_p^p\leq \delta.
\end{align}

Using Lemma~\ref{lemma:lipchitz}, we have for each $i = 1, \ldots, N$:
\begin{align}
    \big| \SW_p^p(\Gt_i(\phi,\pi), G_i) - \SW_p^p(\Gt_i(\phi^{(h)},\pi^{(h)}), G_i) \big|
    \leq \gamma(\delta),
\end{align}
where $\lim_{\delta\to 0}\gamma(\delta)=0$.
Averaging over $i$ and taking expectation, this implies
\begin{align}
   &\frac{1}{N}\sum_{i=1}^N  \big| \SW_p^p(\Gt_i(\phi,\pi), G_i) - \SW_p^p(\Gt_i(\phi^{(h)},\pi^{(h)}), G_i) \big|
    \leq \gamma(\delta), \text{ and} \\
   & \EE\left[ \big| \SW_p^p(\Gt(\phi,\pi), G) - \SW_p^p(\Gt(\phi^{(h)},\pi^{(h)}), G) \big| \right]\leq \gamma(\delta)
\end{align}
By  Jensen's inequality, we have:
\begin{align}
   &\left|\frac{1}{N}\sum_{i=1}^N   \SW_p^p(\Gt_i(\phi,\pi), G_i) - \frac{1}{N}\sum_{i=1}^N \SW_p^p(\Gt_i(\phi^{(h)},\pi^{(h)}), G_i)
    \right|\leq \gamma(\delta), \text{ and} \\
   & \left|\EE\left[ \SW_p^p(\Gt(\phi,\pi), G) - \SW_p^p(\Gt(\phi^{(h)},\pi^{(h)}), G) \right] \right|\leq \gamma(\delta)
\end{align}
which is equivalent to
\begin{align}
    |R_N(\phi,\pi) - R_N(\phi^{(h)},\pi^{(h)})| \leq \gamma(\delta), \text{ and} \quad |R(\phi,\pi) - R(\phi^{(h)},\pi^{(h)})| \leq \gamma(\delta).
 \end{align}
Using triangle inequality, we have:
\begin{align}
    &\sup_{(\phi,\pi) \in \Phi\times \Delta^K} |R_N(\phi,\pi) - R(\phi,\pi)|  \nonumber \\
    &\leq \sup_{(\phi,\pi) \in \Phi\times \Delta^K}  \left(|R_N(\phi,\pi) -R_N(\phi^{(h)},\pi^{(h)})| +|R_N(\phi^{(h)},\pi^{(h)}) - R(\phi,\pi)|  \nonumber \right)\\
    &\leq \nonumber  \sup_{(\phi,\pi) \in \Phi\times \Delta^K}  \left( |R_N(\phi,\pi) -R_N(\phi^{(h)},\pi^{(h)})|+ |R_N(\phi^{(h)},\pi^{(h)}) - R(\phi^{(h)},\pi^{(h)})|  \right. \\& \left.\quad+|R(\phi^{(h)},\pi^{(h)}) - R(\phi,\pi)| \right) \nonumber \\
    &\leq \max_{h = 1,\ldots,H} |R_N(\phi^{(h)},\pi^{(h)}) - R(\phi^{(h)},\pi^{(h)})| + 2 \gamma( \delta). \label{eq:key-net-bound}
\end{align}

Let
$
Z_N \;:=\; \sup_{(\phi,\pi)\in\Phi\times \Delta^K}\,|R_N(\phi,\pi)-R(\phi,\pi)| $ and $
A_N^\varepsilon \;:=\; \{\,Z_N>\varepsilon\,\},\ \ \varepsilon>0 .
$
Almost sure convergence $Z_N\overset{a.s}{\to} 0$ is equivalent to
\begin{align}
\forall\,\varepsilon>0:\quad \mathbb{P}\!\left(\limsup_{N\to\infty} A_N^\varepsilon\right)=0,
\end{align}
where $\displaystyle \limsup_{N\to\infty} A_N^\varepsilon:=\bigcap_{m=1}^\infty\bigcup_{N\ge m}A_N^\varepsilon$. For a fixed $\varepsilon>0$,  we choose $\delta>0$ so small that $2\gamma(\delta)<\varepsilon/2$. Then \eqref{eq:key-net-bound} implies the event inclusion
\begin{align}
A_N^\varepsilon
\;\subseteq\;
\bigcup_{h=1}^H
\left\{\,|R_N(\phi^{(h)},\pi^{(h)})-R(\phi^{(h)},\pi^{(h)})|>\frac{\varepsilon}{2}\right\}
=: \bigcup_{h=1}^H B_{N,h}^{\varepsilon/2}.
\end{align}
Therefore,
\begin{align}
\limsup_{N\to\infty} A_N^\varepsilon
\;\subseteq\;
\bigcup_{h=1}^H \left( \limsup_{N\to\infty} B_{N,h}^{\varepsilon/2} \right).
\end{align}

For fixed $(\phi,\pi)$, the loss $\SW_p^p(\Gt(\phi,\pi),G)$ is bounded by $D^p$ and is a continuous, hence measurable, function of $(F_1,\ldots,F_K,G)\in\mathrm{supp}(P)$ by Lemma~\ref{lemma:lipchitz}. By the strong law of large numbers applied to each fixed $(\phi^{(h)},\pi^{(h)})$,
\begin{align}
|R_N(\phi^{(h)},\pi^{(h)})-R(\phi^{(h)},\pi^{(h)})|\overset{a.s}{\to}0,
\end{align}
hence, for every $h$ and every $\eta>0$,
\begin{align}
\mathbb{P}\!\left(\limsup_{N\to\infty} \{\,|R_N(\phi^{(h)},\pi^{(h)})-R(\phi^{(h)},\pi^{(h)})|>\eta\,\}\right)=0.
\end{align}
Taking $\eta=\varepsilon/2$ and using the finite union bound, we get
\begin{align}
\mathbb{P}\!\left(\limsup_{N\to\infty} A_N^\varepsilon\right)
\;\le\;
\sum_{h=1}^H
\mathbb{P}\!\left(\limsup_{N\to\infty} B_{N,h}^{\varepsilon/2}\right)
=0.
\end{align}

Since this holds for every $\varepsilon>0$, it follows that
\begin{align}
\forall\,\varepsilon>0:\quad \mathbb{P}\!\left(\limsup_{N\to\infty} A_N^\varepsilon\right)=0,
\end{align}
which is equivalent to $Z_N\to 0$ almost surely, i.e.,
\begin{align}
\sup_{(\phi,\pi)\in\Phi\times\Delta^K}|R_N(\phi,\pi)-R(\phi,\pi)| \overset{a.s}{\to} 0,
\end{align}
which completes the proof of the uniform law of large numbers for the risks.
\end{proof}

With the proved uniform law of large numbers, we now discuss the proof of posterior consistency. As in the proof of Lemma~\ref{lemma:uniform_law}, $|R(\phi,\pi) - R(\phi',\pi')| \leq \gamma(\|\phi-\phi'\|_p^p+\|\pi-\pi'\|_p^p)$, so $R$ is continuous on the compact set $\Theta$, and $\Theta_0$ is non-empty and compact. We denote $R^\star = \min_{(\phi,\pi)\in\Theta}R(\phi,\pi)$. For all $\epsilon>0$, we define the ``bad'' set of parameters:
\begin{align}
S_\epsilon := \{ (\phi,\pi) \in \Theta : \rho\big((\phi,\pi),\Theta_0\big) \geq \epsilon\}.
\end{align}
If $S_\epsilon$ is empty, there is nothing to prove. Otherwise, $S_\epsilon$ is compact since it is a closed subset of $\Theta$, $R$ is continuous, and $R>R^\star$ on $S_\epsilon$ since $S_\epsilon\cap\Theta_0=\emptyset$. Hence,
\begin{align}
    \Delta(\epsilon) := \min_{(\phi,\pi)\in S_\epsilon} R(\phi,\pi) - R^\star >0.
\end{align}
By the uniform law of large numbers (Lemma~\ref{lemma:uniform_law}), for any $\eta>0$, the event
\begin{align}
E_N := \Big\{\sup_{(\phi,\pi)\in\Theta}|R_N(\phi,\pi)-R(\phi,\pi)| \le \eta \Big\}
\end{align}
occurs for all sufficiently large $N$ almost surely. On $E_N$:

\begin{itemize}
\item For $(\phi,\pi) \in S_\epsilon$,
\begin{align}
R_N(\phi,\pi) \ge R(\phi,\pi) - \eta \ge R^\star + \Delta(\epsilon)  - \eta.
\end{align}
\item Let $(\phi_0,\pi_0)\in\Theta_0$ be given by Assumption~\ref{assumption:prior}. For $ (\phi,\pi) \in B_\delta(\phi_0,\pi_0)$ (``good" set), choose $\delta$ small so that
\begin{align}
\sup_{ (\phi,\pi) \in B_\delta(\phi_0,\pi_0)} (R(\phi,\pi)-R^\star) \le \eta,
\end{align}
which is possible since $R$ is continuous and $R(\phi_0,\pi_0)=R^\star$. Then
\begin{align}
R_N(\phi,\pi) \le R(\phi,\pi) + \eta \le R^\star + 2\eta.
\end{align}
\end{itemize}

On $E_N$, the posterior mass ratio of ``bad" set and ``good" set
\begin{align}
    \frac{p_N(S_\epsilon)}{p_N(B_\delta(\phi_0,\pi_0))}
    &= \frac{\int_{S_\epsilon} \exp(-w N R_N(\phi,\pi)) \, \mathrm{d} p(\phi,\pi)}
           {\int_{B_\delta(\phi_0,\pi_0)} \exp(-w N R_N(\phi,\pi)) \, \mathrm{d} p(\phi,\pi)} \nonumber \\
    &\leq \frac{p(\Theta) \, \exp(-w N (R^\star + \Delta(\epsilon)  - \eta))}{p(B_\delta(\phi_0,\pi_0)) \, \exp(-w N (R^\star + 2\eta))} \nonumber \\
    &= \frac{1}{p(B_\delta(\phi_0,\pi_0))} \exp(-w N (\Delta(\epsilon)  - 3\eta)), \label{eq:posterior_ratio}
\end{align}
where $p(\phi,\pi)$ is the prior measure. Since $p_N(B_\delta(\phi_0,\pi_0))\leq 1$, we have
\begin{align}
    p_N(S_\epsilon) \leq C \exp(-w N (\Delta(\epsilon)  - 3\eta)),
\end{align}
where $C=\frac{1}{p(B_\delta(\phi_0,\pi_0))} >0$ due to Assumption~\ref{assumption:prior}.
Choose $\eta < \Delta(\epsilon) /4$ so that $\Delta(\epsilon) -3\eta > \Delta(\epsilon) /4>0$, then the right-hand side decays exponentially in $N$. Hence, $p_N(S_\epsilon) \to 0$ almost surely, i.e., the posterior concentrates on the set of risk minimizers $\Theta_0$, which completes the proof.

\subsection{Proof of Theorem~\ref{theorem:approx_posterior_consistency}}
\label{subsec:proof:theorem:approx_posterior_consistency}
We use the notation of the proof of Theorem~\ref{theorem:posterior_consistency}. For unit $i$ and $\vartheta=(\phi,\pi)$, let $\mathcal{M}_i^{(n)}(\vartheta)$ denote the set of all minimizers defining $\SWB_p(f_{1,\phi_1}\sharp\hat F_{i1}^{(n)},\ldots,f_{K,\phi_K}\sharp\hat F_{iK}^{(n)},\pi)$, and define the plug-in error
\begin{align}
    e_i(n):=\sup_{\vartheta\in\Phi\times\Delta^K}\ \sup_{H\in\mathcal{M}_i^{(n)}(\vartheta)}\big|\SW_p^p(H,\hat G_i^{(n)})-\SW_p^p\{\Gt_i(\phi,\pi),G_i\}\big|. \label{eq:plugin_error}
\end{align}
Since $\widehat{\Gt}_i(\phi,\pi)\in\mathcal{M}_i^{(n)}(\vartheta)$, we have $|\hat R_N^{(n)}(\vartheta)-R_N(\vartheta)|\leq\frac1N\sum_{i=1}^Ne_i(n)$ for all $\vartheta$. The variables $e_1(n),\ldots,e_N(n)$ are i.i.d and bounded by $D^p$ by \eqref{eq:D_bound}. We assume that $e_i(n)$ is measurable. Let $\bar e_n:=\EE[e_1(n)]$, and since $\bar e_n$ depends only on the law of the unit and its first $n$ samples, we may assume without loss of generality that the samples of each unit are the first $n$ terms of an infinite i.i.d sequence.

We now show that  $\bar e_n\to0$. Almost surely, we have $(F_{11},\ldots,F_{1K},G_1)\in\mathrm{supp}(P)$, and $\hat F_{1k}^{(n)}\to F_{1k}$ and $\hat G_1^{(n)}\to G_1$ weakly for all $k$~\citep{varadarajan1958convergence}. Fix such an outcome and suppose that $e_1(n)\not\to0$. Then there are $\epsilon>0$, $n_l\to\infty$, $\vartheta_l=(\phi_l,\pi_l)$ and $H_l\in\mathcal{M}_1^{(n_l)}(\vartheta_l)$ with $\big|\SW_p^p(H_l,\hat G_1^{(n_l)})-\SW_p^p\{\Gt_1(\phi_l,\pi_l),G_1\}\big|\geq\epsilon$, and by compactness we may assume $\vartheta_l\to\vartheta^\ast=(\phi^\ast,\pi^\ast)$. As in the proof of Lemma~\ref{lemma:lipchitz}, $\SW_p(f_{k,\phi_{l,k}}\sharp\hat F_{1k}^{(n_l)},f_{k,\phi^\ast_k}\sharp F_{1k})\to0$ for every $k$. By the uniqueness of $\Gt_1(\phi^\ast,\pi^\ast)$ in Assumption~\ref{assumption:continuity}, Lemma~\ref{lemma:continuity_SWB} gives $\SW_p(H_l,\Gt_1(\phi^\ast,\pi^\ast))\to0$, and similarly $\SW_p(\Gt_1(\phi_l,\pi_l),\Gt_1(\phi^\ast,\pi^\ast))\to0$. Moreover, $\SW_p(\hat G_1^{(n_l)},G_1)\to0$ by \eqref{eq:weak_SW}. Therefore,
\begin{align}
    \epsilon\leq pD^{p-1}\Big(\SW_p\big(H_l,\Gt_1(\phi^\ast,\pi^\ast)\big)+\SW_p\big(\Gt_1(\phi^\ast,\pi^\ast),\Gt_1(\phi_l,\pi_l)\big)+\SW_p(\hat G_1^{(n_l)},G_1)\Big)\to0,
\end{align}
where the inequality is due to \eqref{eq:Lip} and the triangle inequality, which is a contradiction. Hence, $e_1(n)\to0$ almost surely, and $\bar e_n\to0$ by the dominated convergence theorem.

By the bound above and Lemma~\ref{lemma:uniform_law},
\begin{align}
    &\sup_{\vartheta\in\Phi\times\Delta^K}\big|\hat R_N^{(n)}(\vartheta)-R(\vartheta)\big|\leq\frac{1}{N}\sum_{i=1}^Ne_i(n)+\Delta_N,\nonumber \\
    &\Delta_N:=\sup_{\vartheta\in\Phi\times\Delta^K}|R_N(\vartheta)-R(\vartheta)|\overset{a.s}{\to}0. \label{eq:approx_risk}
\end{align}
Let $\epsilon>0$ with $S_\epsilon\neq\emptyset$ (otherwise there is nothing to prove), and let $S_\epsilon$, $\Delta(\epsilon)$, $R^\star$ and $\delta$ be as in the proof of Theorem~\ref{theorem:posterior_consistency} with $\eta=\Delta(\epsilon)/4$. On the event $E_N:=\{\frac{1}{N}\sum_{i=1}^Ne_i(n_N)\leq\eta/2\}\cap\{\Delta_N\leq\eta/2\}$, \eqref{eq:approx_risk} gives $\sup_\vartheta|\hat R_N^{(n_N)}-R|\leq\eta$, and the argument leading to \eqref{eq:posterior_ratio} gives $\hat p_N^{(n_N)}(S_\epsilon)\leq c_N:=\exp(-wN\Delta(\epsilon)/4)/p(B_\delta(\phi_0,\pi_0))$. Moreover, for each $N$, the variables $e_1(n_N),\ldots,e_N(n_N)$ are i.i.d and take values in $[0,D^p]$, so Hoeffding's inequality~\citep{hoeffding1963probability} gives
\begin{align}
    \mathbb{P}\left(\frac{1}{N}\sum_{i=1}^Ne_i(n_N)>\bar e_{n_N}+\frac{\eta}{4}\right)\leq\exp\left(-\frac{N\eta^2}{8D^{2p}}\right),
\end{align}
which is summable in $N$. By the Borel--Cantelli lemma, which does not require independence across $N$, almost surely $\frac{1}{N}\sum_{i=1}^Ne_i(n_N)\leq\bar e_{n_N}+\eta/4$ for all large $N$. Since $\bar e_{n_N}\to0$, we have $\bar e_{n_N}\leq\eta/4$ for all large $N$, and $\Delta_N\to0$ almost surely. Hence, almost surely, $E_N$ occurs for all large $N$, so $\hat p_N^{(n_N)}(S_\epsilon)\leq c_N\to0$. For a fixed $n$, the strong law of large numbers gives $\frac1N\sum_{i=1}^Ne_i(n)\to\bar e_n$ almost surely. With the same argument with $\{R\geq R^\star+2\bar e_n+\tau\}$ in place of $S_\epsilon$, $\eta<\tau/3$, and $\sup_\vartheta|\hat R_N^{(n)}-R|\leq\bar e_n+\eta$ for all large $N$, we can show that $\hat p_N^{(n)}(\{R\geq R^\star+2\bar e_n+\tau\})\to0$ almost surely for every $\tau>0$. We conclude the proof here.

\subsection{Convexity of Sliced Wasserstein Barycenter}
\label{subsec:proof:proposition:convexity_SWB}
We discuss the convexity of SWB in the distribution space in the following proposition.
\begin{proposition}[Convexity of Sliced Wasserstein Barycenter]
\label{proposition:convexity_SWB} Given $F_{1},\ldots,F_K,G_0,G_1 \in \PP_p(\mathbb{R}^d)$ and $(\pi_1,\ldots,\pi_K) \in \Delta^K$, the mapping \begin{align}G \to \sum_{k=1}^K \pi_k \SW_p^p (G,F_k)\end{align} is convex in $G$ i.e.,
\begin{align}
    \sum_{k=1}^K \pi_k \SW_p^p (G_t,F_k) \leq t \sum_{k=1}^K \pi_k \SW_p^p (G_0,F_k) +(1-t) \sum_{k=1}^K \pi_k \SW_p^p (G_1,F_k),
\end{align}
for $G_t = t G_0 + (1-t)G_1$ with $t\in [0,1]$.
\end{proposition}

\begin{proof}
Fix $k$ and $\theta\in\mathbb{S}^{d-1}$, and let $\gamma_0$ and $\gamma_1$ be the optimal couplings between $\theta\sharp G_0$ and $\theta\sharp F_k$, and between $\theta\sharp G_1$ and $\theta\sharp F_k$, respectively. Since $\theta\sharp G_t=t\,\theta\sharp G_0+(1-t)\,\theta\sharp G_1$, $\gamma_t = t\gamma_0+(1-t)\gamma_1$ is a coupling between $\theta\sharp G_t$ and $\theta\sharp F_k$. We have:
    \begin{align}
        W_p^p(\theta\sharp G_t,\theta\sharp F_k)
        &\leq \EE_{(x,y) \sim \gamma_t}[|x-y|^p] \nonumber \\
        &= t \EE_{(x,y)\sim \gamma_0}[|x-y|^p]  +(1-t)\EE_{(x,y)\sim \gamma_1}[|x-y|^p] \nonumber \\
        &= tW_p^p(\theta\sharp G_0,\theta\sharp F_k) +(1-t)W_p^p(\theta\sharp G_1,\theta\sharp F_k).
    \end{align}
Taking the expectation over $\theta \sim \mathcal{U}(\mathbb{S}^{d-1})$, multiplying by $\pi_k$ and summing over $k$, we obtain the desired inequality.
\end{proof}

\subsection{A sufficient condition for \eqref{eq:stability}}
\label{subsec:stability_sufficient}
For $F_1,\ldots,F_K\in\PP(B_R)$ and $\pi\in\Delta^K$, let $J(G)=\sum_{k=1}^K\pi_k\SW_p^p(G,F_k)$. We say that $(F_1,\ldots,F_K,\pi)$ satisfies the quadratic growth condition with constant $\lambda>0$ if a minimizer $\bar F$ of $J$ over $\PP(B_R)$ satisfies
\begin{align}
    J(G)-J(\bar F) \;\geq\; \lambda\,\SW_p^p(G,\bar F),
    \qquad \text{for all } G \in \PP(B_R).
    \label{eq:QG}
\end{align}
\begin{lemma}
\label{lemma:QG_stability}
If $(F_1,\ldots,F_K,\pi)$ satisfies \eqref{eq:QG} with constant $\lambda>0$, then \eqref{eq:stability} holds with $\kappa=2pD^{p-1}/\lambda$.
\end{lemma}
\begin{proof}
Let $J'(G) = \sum_{k=1}^K \pi_k \SW_p^p(G,F_k')$, so that $\bar{F}'$ minimizes $J'$. We have
\begin{align}
    \lambda\,\SW_p^p(\bar{F},\bar{F}')
    &\leq J(\bar{F}') - J(\bar{F}) \nonumber \\
    &= \big[J(\bar{F}')-J'(\bar{F}')\big] + \big[J'(\bar{F}')-J'(\bar{F})\big] + \big[J'(\bar{F})-J(\bar{F})\big] \nonumber \\
    &\leq \sum_{k=1}^K \pi_k \big|\SW_p^p(\bar{F}',F_k) - \SW_p^p(\bar{F}',F_k')\big|
     + \sum_{k=1}^K \pi_k \big|\SW_p^p(\bar{F},F_k') - \SW_p^p(\bar{F},F_k)\big| \nonumber \\
    &\leq 2pD^{p-1}\sum_{k=1}^K \pi_k \SW_p(F_k,F_k'),
\end{align}
where the first inequality is due to \eqref{eq:QG} at the point $G=\bar F'$, the second one is due to the optimality of $\bar F'$ for $J'$, and the last one is due to \eqref{eq:Lip}. Dividing by $\lambda$ completes the proof.
\end{proof}

For $\theta\in\mathbb{S}^{d-1}$ and a distribution $F$, we write $Q_{\theta\sharp F}$ for the quantile function of $\theta\sharp F$.
\begin{proposition}
\label{proposition:QG_sufficient}
Let $p=2$, $F_1,\ldots,F_K\in\PP(B_R)$, $\pi\in\Delta^K$, and $\bar{Q}_\theta:=\sum_{k=1}^K\pi_kQ_{\theta\sharp F_k}$. If there is $\bar F\in\PP(B_R)$ such that $Q_{\theta\sharp \bar F}=\bar{Q}_\theta$ for almost every $\theta\in\mathbb{S}^{d-1}$, then $J(G)-J(\bar F)=\SW_2^2(G,\bar F)$ for all $G\in\PP(B_R)$. Hence, $\bar F=\SWB_2(F_1,\ldots,F_K,\pi)$ and \eqref{eq:QG} holds with $\lambda=1$. In particular, this holds if $F_k=(s_kx+b_k)\sharp\mu$ for a distribution $\mu$, $s_k>0$ and $b_k\in\Re^d$, with $\bar F=(\bar s x+\bar b)\sharp\mu$, $\bar s=\sum_{k=1}^K\pi_ks_k$ and $\bar b=\sum_{k=1}^K\pi_kb_k$.
\end{proposition}
\begin{proof}
Using the closed form of the one-dimensional Wasserstein distance and $\sum_{k=1}^K\pi_k(a-b_k)^2=(a-\sum_{k=1}^K\pi_kb_k)^2+\sum_{k=1}^K\pi_k(b_k-\sum_{j=1}^K\pi_jb_j)^2$, we have
\begin{align}
    J(G)
    =\EE_{\theta}\left[\int_0^1\big(Q_{\theta\sharp G}(t)-\bar Q_\theta(t)\big)^2\mathrm{d}t\right]+V
    =\SW_2^2(G,\bar F)+V,
\end{align}
where $V:=\EE_{\theta}[\int_0^1\sum_{k=1}^K\pi_k(Q_{\theta\sharp F_k}(t)-\bar Q_\theta(t))^2\mathrm{d}t]$ does not depend on $G$, and the second equality is due to $Q_{\theta\sharp\bar F}=\bar Q_\theta$. Taking $G=\bar F$ gives $J(\bar F)=V$, which proves the first claim. For the example, $Q_{\theta\sharp F_k}=s_kQ_{\theta\sharp\mu}+\theta^\top b_k$ since $s_k>0$, so $\bar Q_\theta=\bar sQ_{\theta\sharp\mu}+\theta^\top\bar b=Q_{\theta\sharp\bar F}$, and $\bar F\in\PP(B_R)$ by the convexity of $B_R$. We conclude the proof here.
\end{proof}

\begin{remark}
\label{remark:QG}
For $d\geq 2$, the condition of Proposition~\ref{proposition:QG_sufficient} fails in general. For example, let $d=2$, $F_1=\frac{1}{2}(\delta_{e_1}+\delta_{-e_1})$, $F_2=\frac{1}{2}(\delta_{e_2}+\delta_{-e_2})$ and $\pi=(\frac{1}{2},\frac{1}{2})$. Then $\bar Q_\theta$ is the quantile function of $\frac{1}{2}(\delta_{-c(\theta)}+\delta_{c(\theta)})$ with $c(\theta)=(|\theta_1|+|\theta_2|)/2$. If $Q_{\theta\sharp G}=\bar{Q}_\theta$ for some $G$, then, by continuity in $\theta$, $\theta^\top M\theta=c(\theta)^2$ for all $\theta$ with $M=\EE_{X\sim G}[XX^\top]$. The choices $\theta=e_1,e_2$ give $M_{11}=M_{22}=1/4$, and $\theta=(1,\pm1)/\sqrt{2}$ gives $1/4\pm M_{12}=1/2$, which is impossible. Hence, for $d\geq 2$, the SWB is in general not obtained by averaging projected quantile functions, and Proposition~\ref{proposition:QG_sufficient} does not cover all marginals.
\end{remark}

\subsection{Invariance in Remark~\ref{remark:invariance}}
\label{subsec:proof:invariance}
Let $p=2$, $\tilde F_k=f_k\sharp F_k$, $\pi'_k=\pi_k/s_k$ and $f'_k=s_kf_k+t_k$ with $s_k>0$, $\sum_{k=1}^K\pi'_k=1$ and $\sum_{k=1}^K\pi'_kt_k=0$. Write $Q_G=Q_{\theta\sharp G}$ and $Q_k=Q_{\theta\sharp\tilde F_k}$. Since $s_k>0$, the quantile function of $\theta\sharp(f'_k\sharp F_k)$ is $s_kQ_k+\theta^\top t_k$, and hence
\begin{align}
    \sum_{k=1}^K\pi'_k\SW_2^2(G,f'_k\sharp F_k)
    &=\EE_\theta\left[\int_0^1Q_G^2\,\mathrm{d}t-2\int_0^1Q_G\sum_{k=1}^K\pi'_k(s_kQ_k+\theta^\top t_k)\,\mathrm{d}t\right]+c_1 \nonumber\\
    &=\EE_\theta\left[\int_0^1Q_G^2\,\mathrm{d}t-2\int_0^1Q_G\sum_{k=1}^K\pi_kQ_k\,\mathrm{d}t\right]+c_1\\
    &
    =\sum_{k=1}^K\pi_k\SW_2^2(G,\tilde F_k)-c_2+c_1,
\end{align}
where $c_1$ and $c_2$ do not depend on $G$, the first equality is due to $\sum_{k=1}^K\pi'_k=1$, and the second one is due to $\pi'_ks_k=\pi_k$ and $\sum_{k=1}^K\pi'_kt_k=0$. Hence, both objectives have the same minimizers, i.e., the SWB is unchanged. For linear $f_k(x)=A_kx+b_k$, $f'_k(x)=s_kA_kx+s_kb_k+t_k$, and $\pi'_k\|s_kA_k\|_F=\pi_k\|A_k\|_F$ and $\sum_{k=1}^K\pi'_k(s_kb_k+t_k)=\sum_{k=1}^K\pi_kb_k$, which shows the invariance of $\tilde\pi$ and $\bar b$.

\subsection{Free-support approximation}
\label{subsec:proof:prop:free_support}
For $M\geq 1$, let $\mathcal{G}_M:=\{\frac{1}{M}\sum_{\ell=1}^M\delta_{z_\ell}: z_1,\ldots,z_M\in B_R\}$ and $J(G)=\sum_{k=1}^K\pi_k\SW_p^p(G,F_k)$.
\begin{proposition}
\label{prop:free_support}
Let $F_1,\ldots,F_K\in\PP(B_R)$ and $\pi\in\Delta^K$ be such that the SWB ($\bar F$) is unique, $\epsilon_M\to0$ and $\hat G_M\in\mathcal{G}_M$ with $J(\hat G_M)\leq\inf_{G\in\mathcal{G}_M}J(G)+\epsilon_M$, we have $\SW_p(\hat G_M,\bar F)\to0$ as $M\to\infty$.
\end{proposition}
\begin{proof}
Empirical distributions of i.i.d samples from $\bar F$ belong to $\mathcal{G}_M$ and converge weakly to $\bar F$ almost surely~\citep{varadarajan1958convergence}, so there are $G_M\in\mathcal{G}_M$ with $\SW_p(G_M,\bar F)\to0$ by \eqref{eq:weak_SW}. We have
\begin{align}
    J(\bar F)\leq J(\hat G_M)\leq J(G_M)+\epsilon_M\leq J(\bar F)+pD^{p-1}\SW_p(G_M,\bar F)+\epsilon_M,
\end{align}
where the first inequality is due to the optimality of $\bar F$, the second one is due to the definition of $\hat G_M$, and the last one is due to \eqref{eq:Lip}. Hence, $J(\hat G_M)\to J(\bar F)$. Since $\PP(B_R)$ is compact, every subsequence of $(\hat G_M)$ has a further subsequence converging weakly, hence in $\SW_p$ by \eqref{eq:weak_SW}, to some $H$ with $J(H)=J(\bar F)$ by the continuity of $J$. By uniqueness, $H=\bar F$. Therefore, $\SW_p(\hat G_M,\bar F)\to0$, which completes the proof.
\end{proof}
The objective $J$ is convex under mixtures (Proposition~\ref{proposition:convexity_SWB}) but not as a function of the atom locations; e.g., for $p=2$, $K=1$, $M=2$ and $F_1=\frac{1}{2}(\delta_{-e_1}+\delta_{e_1})$, the atoms $(-e_1,e_1)$ and $(e_1,-e_1)$ both give $J=0$, while their average $(0,0)$ gives $J=1/d>0$.

\begin{figure}[!t]
\begin{center}
    \begin{tabular}{cccc}
  \widgraph{0.23\textwidth}{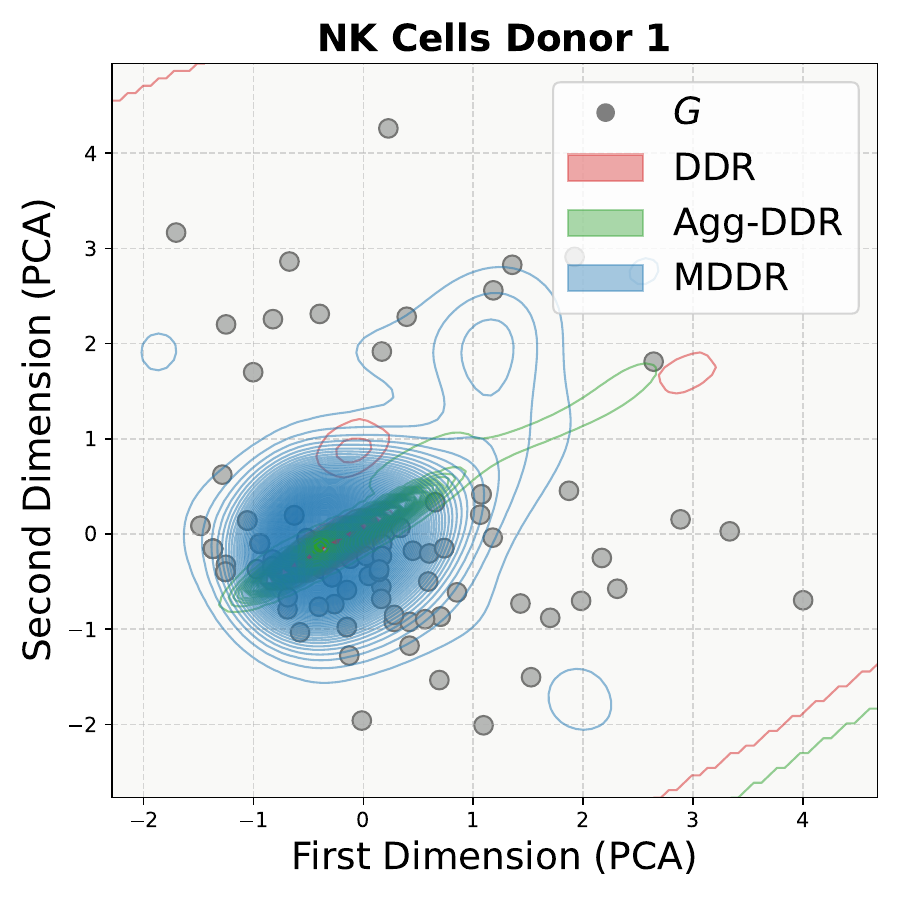}
&
\widgraph{0.23\textwidth}{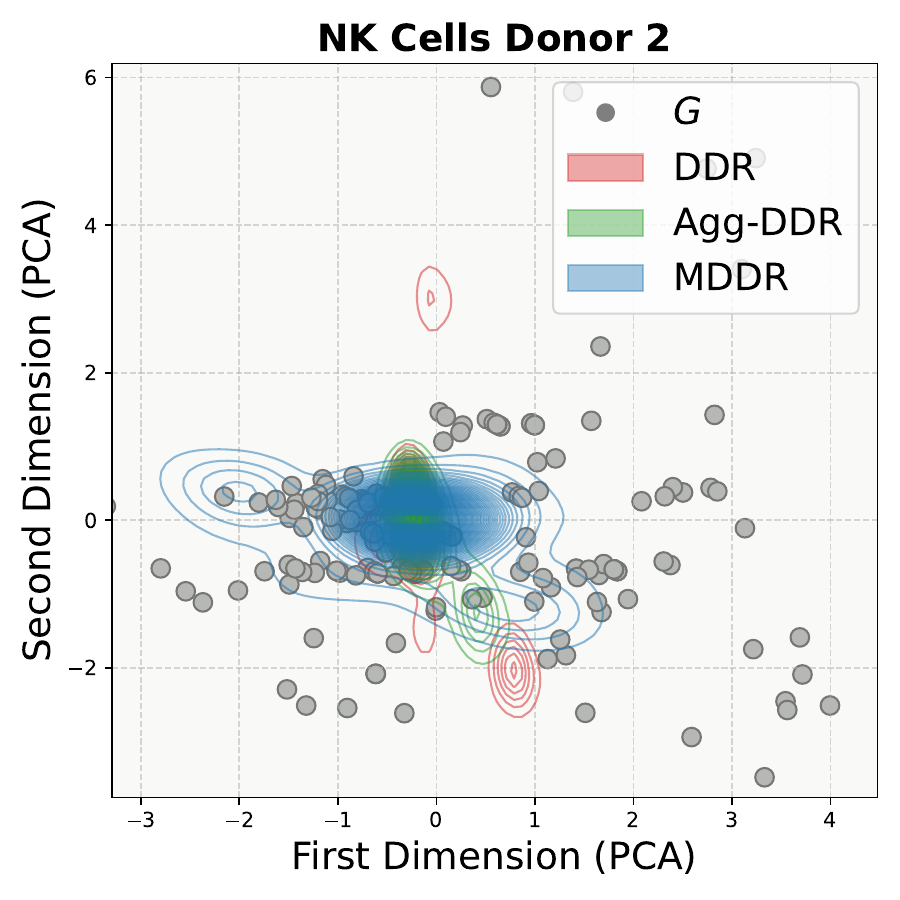}
&
\widgraph{0.23\textwidth}{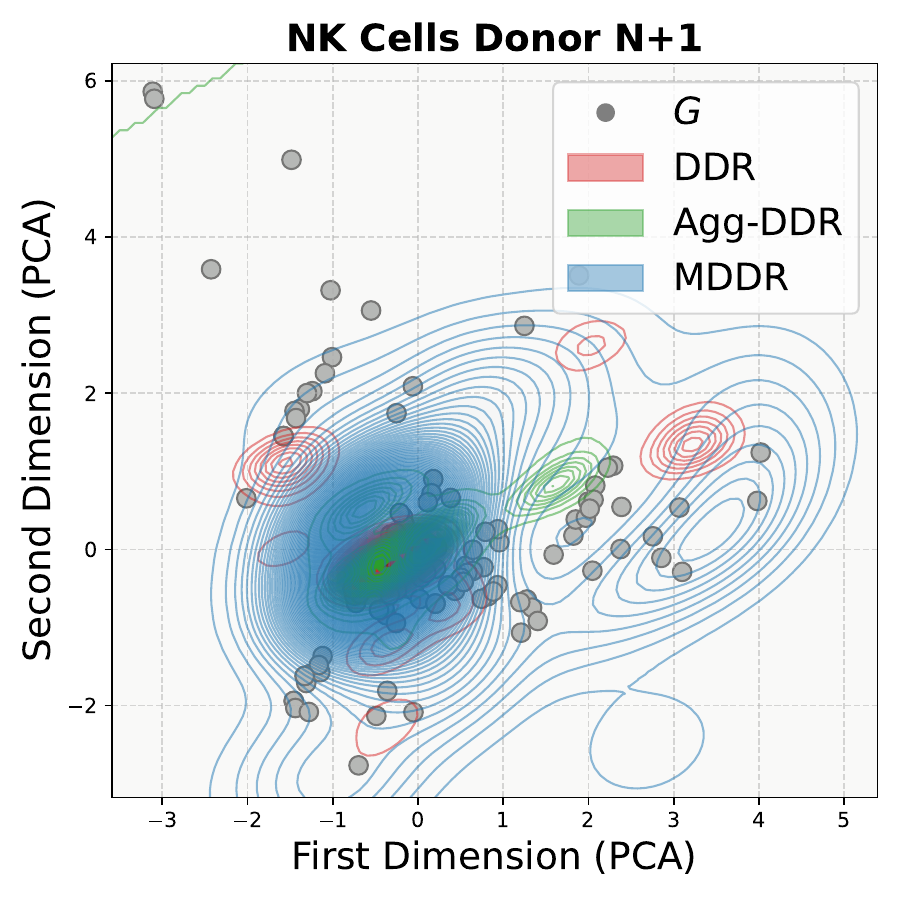}
&
\widgraph{0.23\textwidth}{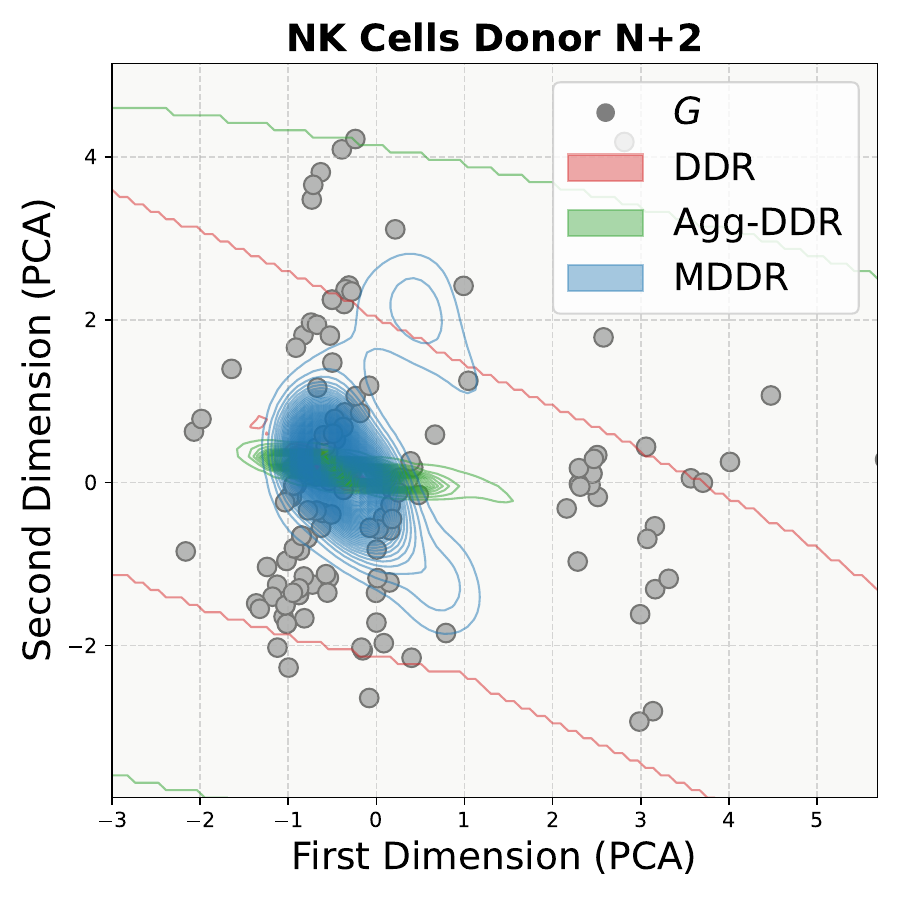}
\\
  \widgraph{0.23\textwidth}{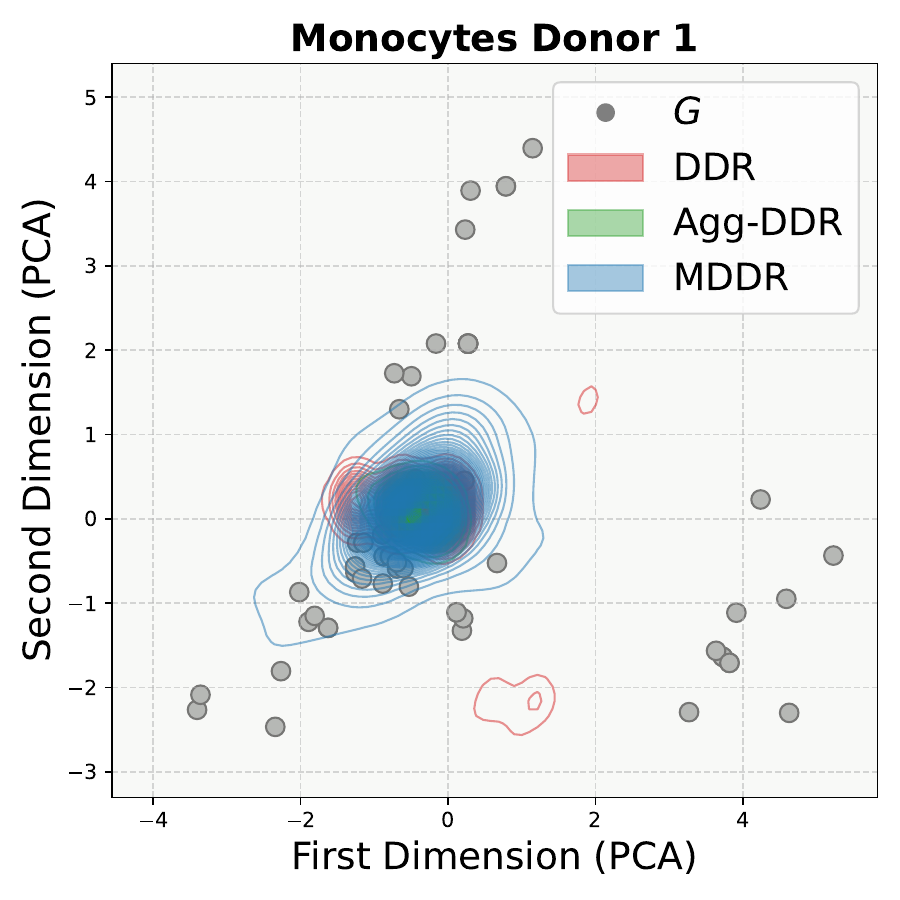}
&
\widgraph{0.23\textwidth}{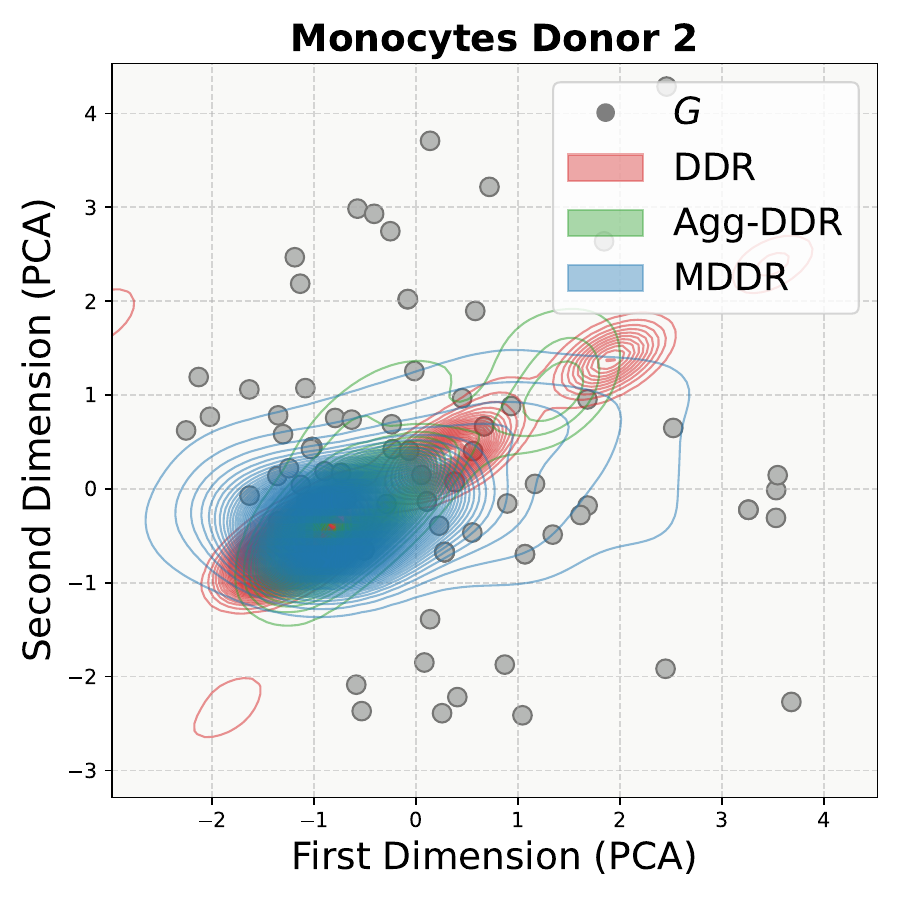}
&
\widgraph{0.23\textwidth}{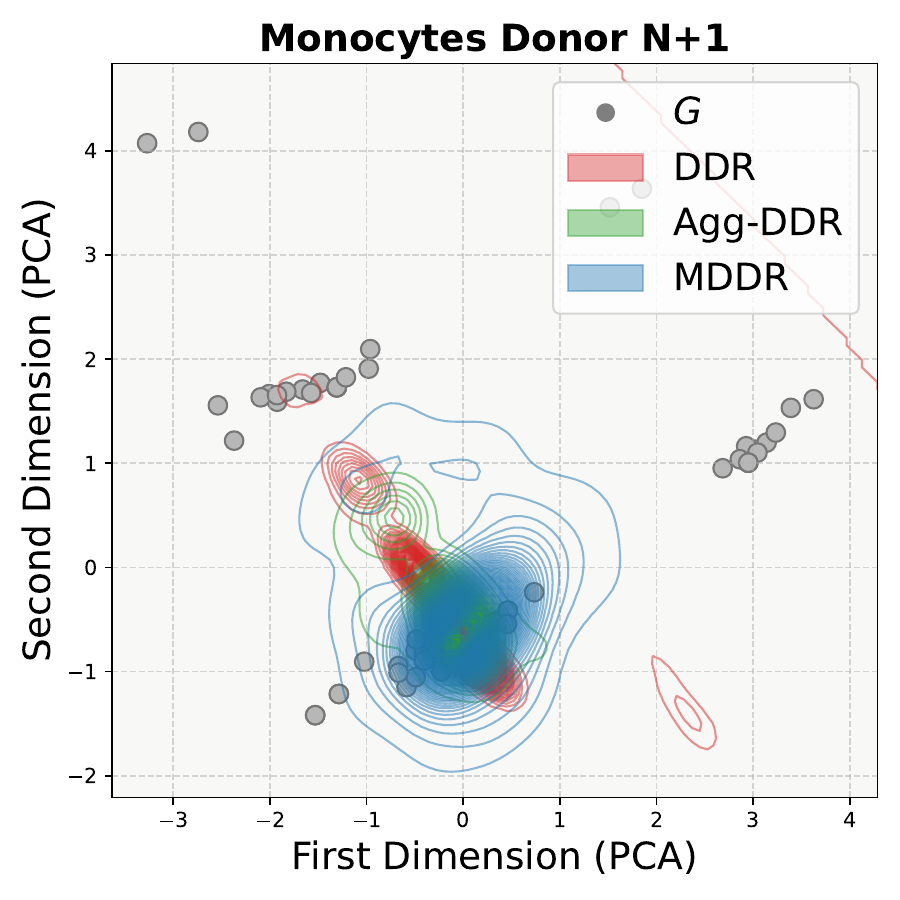}
&
\widgraph{0.23\textwidth}{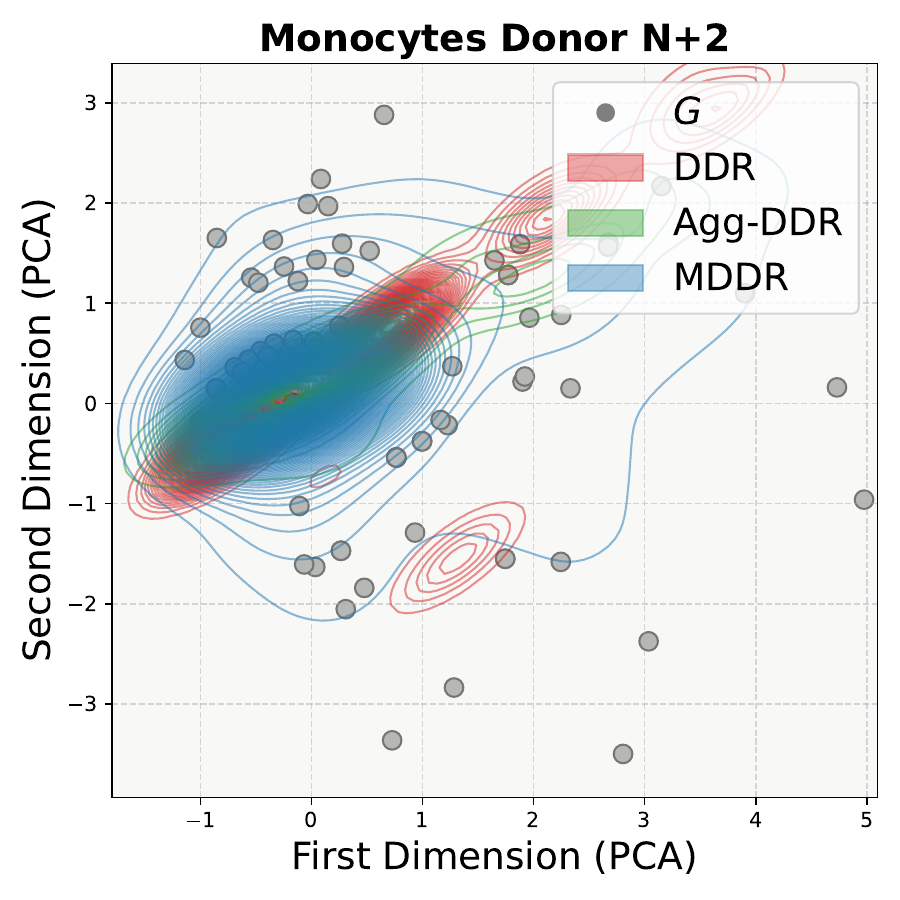}
  \end{tabular}
  \end{center}
  \vspace{-0.2 in}
  \caption{
    \footnotesize{The first two columns show two random in-sample donors and the last two columns show two random out-of-sample donors. The first row presents results for NK Cells including the observed responses, and the posterior means of the KDEs of fitted distributions of DDR (with the best single predictor), aggregated DDR, and MDDR over the 50 posterior samples. The second row presents similar results for Monocytes. 
}
}
  \label{fig:Cell_appendix}
\end{figure}
\section{Additional Results}
\label{sec:additional_results}
As mentioned in the main text, we show the observed responses and the posterior means of the KDEs of the fitted DDR, aggregated DDR, and MDDR distributions for NK cells and Monocytes in Figure~\ref{fig:Cell_appendix}. From these figures, MDDR again leads to better fits than DDR and aggregated DDR for both training and testing donors. For the simulation experiments discussed in Section~\ref{sec:simulation}, Table~\ref{table:N_n_sweep} is about varying the number of units $N$ and samples per distribution $n$, Table~\ref{table:sensitivity} discusses the sensitivity of MDDR with respect to its hyperparameters, Table~\ref{table:reference_sensitivity} provides comparison of different reference models for RE, and Table~\ref{table:misspecified} concerns the misspecified case. For the single-cell data analysis presented in Section~\ref{sec:cell_cell_communication}, Table~\ref{tab:cell_counts} presents the number of cells per each donor, Figure~\ref{fig:Cell_appendix} displays the fitted distributions for NK and monocyte cells, and Figure~\ref{fig:graph_pitilde} displays the communication graph constructed based on the identified weights $\tilde\pi$.

\begin{table}[!h]
    \centering
    \begin{tabular}{c|c|c||c|c|c}
    \toprule
      $N$ ($n=100$) & Test RE & 95\% HCI & $n$ ($N=50$) & Test RE & 95\% HCI \\
      \midrule
      10 & 0.0774 & 0.0392-0.1147 & 10 & 0.0186 & 0.0078-0.0289 \\
      20 & 0.0340 & 0.0162-0.0520 & 20 & 0.0178 & 0.0083-0.0288 \\
      30 & 0.0257 & 0.0115-0.0379 & 50 & 0.0179 & 0.0078-0.0238 \\
      50 & 0.0179 & 0.0078-0.0238 & 100 & 0.0172 & 0.0071-0.0242 \\
      \bottomrule
    \end{tabular}
    \caption{Test RE of MDDR on a fixed test set of 30 units when varying the number of training units $N$ (left) and the number of samples per distribution $n$ (right).}
    \label{table:N_n_sweep}
\end{table}

\begin{table}[!h]
    \centering
    \begin{tabular}{l|c|c|c|c|c}
    \toprule
      Setting & Train RE & Test RE & Test RE 95\% HCI & Posterior mean of $\tilde\pi$ & $\|\bar{\tilde\pi}-\tilde\pi\|_2$ \\
      \midrule
      Baseline & 0.0177 & 0.0204 & 0.0079-0.0373 & (0.659, 0.170, 0.172) & 0.010 \\
      $w=5$ & 0.0293 & 0.0354 & 0.0124-0.0610 & (0.662, 0.169, 0.169) & 0.006 \\
      $w=20$ & 0.0089 & 0.0096 & 0.0049-0.0159 & (0.662, 0.168, 0.170) & 0.006 \\
      Laplace scale 0.5 & 0.0166 & 0.0193 & 0.0079-0.0371 & (0.664, 0.168, 0.169) & 0.004 \\
      Laplace scale 2 & 0.0154 & 0.0177 & 0.0073-0.0308 & (0.659, 0.169, 0.172) & 0.009 \\
      Dirichlet 0.1 & 0.0163 & 0.0197 & 0.0084-0.0373 & (0.660, 0.170, 0.170) & 0.008 \\
      Dirichlet 1 & 0.0150 & 0.0182 & 0.0055-0.0296 & (0.659, 0.170, 0.172) & 0.010 \\
      $L=20$ & 0.0169 & 0.0201 & 0.0077-0.0324 & (0.659, 0.169, 0.172) & 0.010 \\
      $L=50$ & 0.0160 & 0.0176 & 0.0071-0.0324 & (0.661, 0.169, 0.170) & 0.007 \\
      $T=50$ & 0.0212 & 0.0216 & 0.0058-0.0381 & (0.663, 0.167, 0.169) & 0.004 \\
      \bottomrule
    \end{tabular}
    \caption{Sensitivity of MDDR to one hyperparameter at a time (first simulated dataset). The baseline uses $w=10$, Laplace scale 1, Dirichlet concentration 0.01, $L=100$, and $T=100$. The true $\tilde\pi$ is $(2/3,1/6,1/6)$.}
    \label{table:sensitivity}
\end{table}

\begin{table}[!h]
    \centering
    \begin{tabular}{l|c|c|c}
    \toprule
      Reference model & Train RE & Test RE & Test RE 95\% HCI \\
      \midrule
      Intercept only& 0.0177 & 0.0204 & 0.0079-0.0373 \\
       MDDR intercepts, uniform weights & 0.0169 & 0.0195 & 0.0079-0.0356 \\
      MDDR intercepts, MDDR weights & 0.0173 & 0.0210 & 0.0080-0.0388 \\
      MDDR regression functions, uniform weights & 0.0691 & 0.0614 & 0.0175-0.1147 \\
      Pooled training responses & 0.0049 & 0.0074 & 0.0031-0.0135 \\
      \bottomrule
    \end{tabular}
    \caption{RE of MDDR under different reference models (first simulated dataset). }
    \label{table:reference_sensitivity}
\end{table}

\begin{table}[!h]
    \centering
    \begin{tabular}{c|c|c|c|c}
    \toprule
         Model&  Train RE & Train RE 95\% HCI & Test RE & Test RE 95\% HCI  \\
         \midrule
        DDR & 0.2828 & 0.2661-0.3010 & 0.3059 & 0.2881-0.3300 \\
        Aggregated DDR & 0.4227 & 0.4041-0.4421 & 0.4275 & 0.3991-0.4553 \\
        MDDR& \textbf{0.0739}&\textbf{0.0458-0.1129}&\textbf{0.1055}&\textbf{0.0576-0.1790} \\
        \bottomrule
    \end{tabular}
    \caption{Misspecified simulation with nonlinear true regression functions ($\kappa=1$). All models use linear regression functions.}
    \label{table:misspecified}
\end{table}

\begin{table}[!h]
    \centering
    \begin{tabular}{l|c|c|c|c|c}
    \toprule
      Cell type & Minimum & First quartile & Median & Third quartile & Maximum \\
      \midrule
      T Cells & 202 & 522 & 647 & 734 & 1104 \\
      B Cells & 90 & 105 & 129 & 164 & 296 \\
      Monocytes & 90 & 103 & 124 & 152 & 353 \\
      NK Cells & 92 & 127 & 166 & 208 & 581 \\
      \bottomrule
    \end{tabular}
    \caption{Number of cells per donor for each cell type, over the 75 donors used in Section~\ref{sec:cell_cell_communication}.}
    \label{tab:cell_counts}
\end{table}

\begin{figure}[!h]
    \centering
    \begin{tabular}{cc}
         \includegraphics[width=0.45\linewidth]{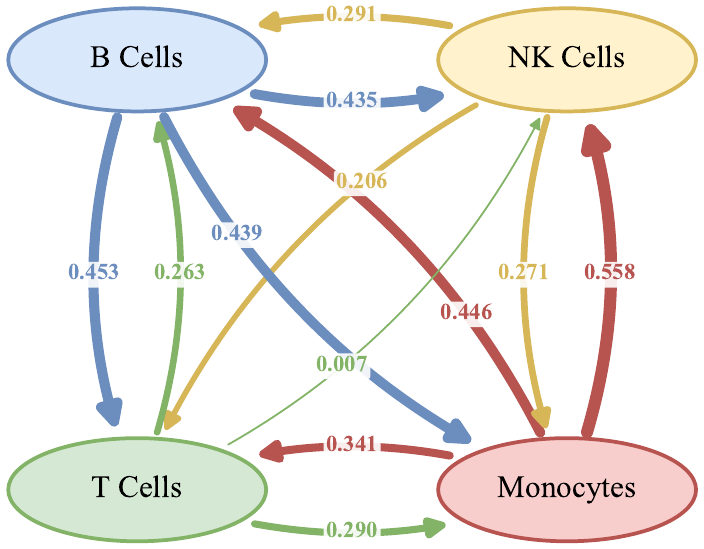}&
    \includegraphics[width=0.45\linewidth]{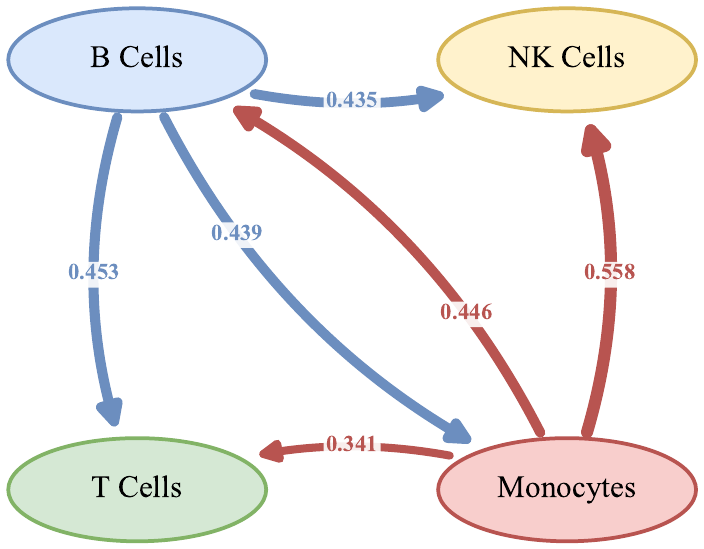} \\
    (a)&(b)
    \end{tabular}
    \caption{(a) Weighted graph using the posterior means of the identified barycenter weights $\tilde\pi$ (Remark~\ref{remark:invariance}) over the last 500 MCMC iterations. (b) Sparse graph retaining the edges with $\tilde\pi>1/3$.}
    \label{fig:graph_pitilde}
\end{figure}

\end{document}